\documentclass[a4paper,11pt]{article}
\usepackage{cite}
\usepackage{hyperref}
\hypersetup{
	colorlinks=true,
	linkcolor=blue,
	citecolor=blue,
	urlcolor=blue
}
\usepackage[top=1 in,bottom=1 in,left=1.0 in,right=1.0 in]{geometry}
\usepackage{amsthm}
\usepackage{amssymb,amsmath}
\usepackage{scalerel}

\numberwithin{equation}{section}

\newcommand{\bsb}{\begin{subequations}}
	\newcommand{\esb}{\end{subequations}}

\usepackage{amsmath}
\usepackage{amsfonts}
\usepackage{latexsym}
\usepackage{amssymb}
\usepackage{appendix}
\usepackage{mathrsfs}
\usepackage{graphicx}
\usepackage{bm}
\usepackage{arydshln}
\usepackage{tikz}
\usepackage{multirow}
\usepackage{ulem}

\usepackage{color}
\newcommand{\bred}{\begin{color}{red}}
\newcommand{\ecl}{\end{color}}
\newcommand{\bblue}{\begin{color}{blue}}
\newcommand{\bgre}{\begin{color}{green}}
\newcommand{\bora}{\begin{color}{orange}}

\begin{document}
				
	
\begin{titlepage}
	\null
	\begin{flushright}
		arXiv:2607.17749
		\\
		July, 2026
	\end{flushright}
	
	\vskip 1cm
	\begin{center}
		
		{\LARGE \bf Asymptotic Equivalence Between Quasi-Grammian 
			\\
			and Quasi-Wronskian $N$-Soliton Solutions of the 
			\\
			\vskip 0.3cm
			Anti--Self--Dual Yang--Mills Equation}
		
		\vskip 1cm
		\normalsize
		
		{\large 
			Shan-Chi Huang \footnote{E-mail:x18003x@math.nagoya-u.ac.jp}
		}
		
		\vskip 0.7cm
		
		{Graduate School of Mathematics, Nagoya University,\\
			Nagoya, 464-8602, JAPAN}
		
		
		
		\vskip 1.3cm
		

	\end{center}

\begin{abstract}
Asymptotic equivalence between the quasi-Grammian and quasi-Wronskian representations of $N$-soliton solutions in the $J$-matrix formulation of the anti-self-dual Yang--Mills (ASDYM) equation is established up to a constant matrix factor. This formulation, known as the Yang equation, serves as the equation of motion of the four-dimensional Wess--Zumino--Witten (WZW$_4$) model and is equivalent to the ASDYM equation.
To visualize the solitonic behavior, the action density of the WZW$_4$ model is evaluated for $\mathrm{G}=\mathrm{U}(2)$, demonstrating that the quasi-Grammian and quasi-Wronskian representations exhibit the same asymptotic soliton profiles, while the phase shift factors associated with $N$-soliton collisions are obtained explicitly. Hence, by virtue of the particle-like nature of solitons, the two representations describe the same class of ASDYM $N$-soliton solutions.
These solitons can be regarded as a four-dimensional analogue of KP/KdV-type multi-solitons in fluid dynamics, suggesting a possible connection between the ASDYM equation and higher-dimensional Sato theory. Exact quasi-Grammian $N$-soliton solutions are also presented for $N\leq4$.
\end{abstract}
	\end{titlepage}

\clearpage
\baselineskip 6mm

\tableofcontents


\section{Introduction}
Due to their intrinsic nonlinearity, nonlinear partial differential equations are generally difficult to solve exactly. Nevertheless, there exists a distinguished class of nonlinear equations, known as integrable systems (or integrable equations), which share highly similar mathematical structures \cite{DM-2000, Hirota-2004, JM-1983, MJD-2000, MGK-1968, Sato-1981}. As a result, these systems admit a unified set of mathematical techniques \cite{AKNS-1974, CFYG-1978, DM-2009, UN-1983, MS-1991, SOM-1998, Hirota-2004, BV-2001, RS-2002, GN-2007, GHN-2009, L2Z-2022, Kakei-2023, Zhang-2026} for constructing exact solutions. A defining characteristic of integrable systems is the existence of soliton solutions, which arise from nonlinear effects. Typically, a soliton\cite{DJ-1989} is a stable, localized wave exhibiting particle-like behavior, characterized by the preservation of its shape and velocity even after nonlinear interactions with other solitons. This remarkable property implies that the energy density of a solitonic wave does not dissipate during propagation. Owing to these features, solitons have been extensively studied in a wide range of physical and applied contexts \cite{Kodama-2017-18, MG-2006, BV-2001, MSu-2004, Vachaspati-2023, Takahashi-2025}. Despite their wide applicability, the fundamental origin of complete integrability remains not fully understood. It is widely accepted that the four-dimensional anti--self--dual Yang--Mills (ASDYM) equation plays a central role in this context. More specifically, the ASDYM equation is regarded as a master equation, from which many lower-dimensional classical integrable equations can be derived through dimensional reduction\cite{Ward-1985, MW-1996, LLZ-2025, LMZ-2026}. On the other hand, the ASDYM equation admits an equivalent formulation known as the Yang equation \cite{Yang-1977, BFNY-1978}, which serves as the equation of motion (Euler--Lagrange equation) of the four-dimensional Wess--Zumino--Witten (WZW$_{4}$) model \cite{Nair-1991, NS-1992, LMNS-1996, IKU-1997, IKUX-1997}. From this perspective, a wide variety of soliton solutions of integrable equations  can be described in a unified manner within the framework of the WZW$_{4}$ model \cite{HH-2024}.

Meanwhile, it has long been known that many lower-dimensional integrable systems can be unified within the framework of Sato theory \cite{Sato-1981}. According to Sato theory, a broad class of integrable equations, most notably the Kadomtsev--Petviashvili (KP) equation, admit solutions formulated in terms of the so-called $\tau$-function. In particular, the $\tau$-function describing the standard $N$-soliton solution (where 
$N$ denotes the soliton number) possesses two equivalent representations \cite{CLM-2010}: the $N$-th order Wronskian determinant and the 
$N$-th order Grammian determinant. It was further shown in \cite{MS-1991} that these two representations can be constructed via the Darboux transformation and the binary Darboux transformation, respectively, thereby establishing a direct correspondence between them. Although the ASDYM equation is not directly formulated within the framework of Sato theory, recent studies \cite{NGO-2000, GHHN-2020, HH-2022, LQYZ-2022, HHK-2023, LQZ-2023, Ohta-2024, LHHZ-2025} have revealed that its solutions admit four-dimensional $\tau$-function-like objects, which can be formulated in terms of quasideterminants\cite{GR-1991, GGRW-2005} and play an analogous role in a noncommutative setting \cite{EGR-1997, GN-2007, GNS-2008}. 
More specifically, the ASDYM equation admits two classes of exact solutions, which can be expressed in terms of Wronskian-type quasideterminants (quasi-Wronskians) and Grammian-type quasideterminants (quasi-Grammians). Furthermore, we showed in \cite{HH-2022, HHK-2023} that the $N$-th order quasi-Wronskian representation yields a distinguished class of solutions to the ASDYM equation that can be regarded as four-dimensional $N$-solitons, serving as a natural analogue of the universal $N$-soliton solutions observed in lower-dimensional integrable systems. This naturally raises the question of whether the $N$-th order quasi-Grammian representation also describes the same class of $N$-soliton solutions. In a previous work \cite{LHHZ-2025}, we verified the equivalence of the two representations explicitly for the one- and two-soliton cases. In the present paper, we extend this result and demonstrate that the equivalence holds for general $N$-soliton solutions through asymptotic analysis.

In Section 2, we briefly review the complexified ASDYM equation, its equivalent formulation as the Yang equation, and the corresponding WZW$_4$ model as necessary background.
In Section 3, we briefly review the background on quasideterminants relevant to our applications.
In Section 4, we present the quasi-Grammian solution of the ASDYM equation in Yang's $J$-matrix formulation for $\mbox{G}=\mbox{GL}(n,\mathbb{C})$ and discuss its unitarity on each real space.
In Section 5, we show that the quasi-Grammian and quasi-Wronskian representations of the ASDYM $N$-soliton are asymptotically equivalent up to a constant matrix factor. Consequently, they represent the same class of ASDYM $N$-soliton solutions in the WZW$_4$ model.
Section 6 is devoted to concluding remarks.

\section{ASDYM Equation and WZW$_4$ Model}

For convenience and to allow a unified treatment, we introduce a four-dimensional complex space $\mathbb{CM}$ with coordinates $(z,\widetilde z, w, \widetilde w)$ equipped with the flat metric \cite{MW-1996}
\begin{eqnarray}
	\mathrm{d}s^2=g_{mn}\mathrm{d}z^{m}\mathrm{d}z^{n}=2(\mathrm{d}z \mathrm{d}\widetilde z -\mathrm{d}w \mathrm{d}\widetilde w),
	~~~m,n=1,2,3,4,
\end{eqnarray}
where
\begin{eqnarray}
	\label{com_metric}
	g_{mn}:=\left(
	\begin{array}{cccc}
		0 & 1 & 0 & 0 \\
		1 & 0 & 0 & 0 \\ 
		0 & 0 & 0 & -1\\
		0 & 0 & -1 & 0
	\end{array}
	\right),
	~(z^{1},z^{2},z^{3},z^{4}):=(z,\widetilde z, w, \widetilde w).
\end{eqnarray}
Note that four-dimensional real spaces can be recovered by imposing appropriate reality conditions on the complex coordinates $(z,\widetilde z, w, \widetilde w)$ as follows:
\begin{eqnarray}
	&(\mathbb{E}^{4})&	
	\label{Reality condition_E}
	\left(\begin{array}{cc}
		z & w \\
		\widetilde{w} & \widetilde{z}
	\end{array}\right)_{\widetilde{z}=\overline{z}, \widetilde{w}=-\overline{w}}
	=
	\frac{1}{\sqrt{2}}
	\left(\begin{array}{cc}
		x^{1}+ix^{2} & x^{3}+ix^{4} \\
		-(x^{3}-ix^{4}) & x^{1}-ix^{2}
	\end{array}\right),
	\smallskip \\
	&&~~~\mathrm{d}s^2=\eta_{\mu\nu}^{(\mathbb{E}^{4})}\mathrm{d}x^{\mu}\mathrm{d}x^{\nu},~\eta_{\mu\nu}^{(\mathbb{E}^{4})}:=\mbox{diag}(1,1,1,1).\label{E}
	\medskip \\
	&(\mathbb{M}^{3,1})&
	\label{Reality condition_M}
	\left(\begin{array}{cc}
		z & w \\
		\widetilde{w} & \widetilde{z}
	\end{array}\right)_{z,\widetilde z \in \mathbb{R},~\widetilde w= \overline w}
	=
	\frac{1}{\sqrt{2}}
	\left(\begin{array}{cc}
		x^{0}+x^{1} & x^{2}+ix^{3} \\
		x^{2}-ix^{3} & x^{0}-x^{1}
	\end{array}\right),
	\smallskip \\
	&&~~~\mathrm{d}s^2=\eta_{\mu\nu}^{(\mathbb{M}^{3,1})}\mathrm{d}x^{\mu}\mathrm{d}x^{\nu},~\eta_{\mu\nu}^{(\mathbb{M}^{3,1})}:=\mbox{diag}(1,-1,-1,-1) \label{M}.
	\medskip \\	
	\label{Reality condition_U}
	&(\mathbb{U}_1^{2,2})&
	\left(
	\begin{array}{cc}
		z & w \\
		\widetilde{w} & \widetilde{z}
	\end{array}\right)_{z, \widetilde{z}, w, \widetilde{w} \in \mathbb{R}}
	=
	\displaystyle{\frac{1}{\sqrt{2}}}
	\left(\begin{array}{cc}
		x^{1}+x^{3} & x^{2}+x^{4} \\
		-(x^{2}-x^{4}) & x^{1}-x^{3} 
	\end{array}\right), 
	\smallskip \\
	\label{Reality condition_U_2}
	&(\mathbb{U}_2^{2,2})&
	\left(
	\begin{array}{cc}
		z & w \\
		\widetilde{w} & \widetilde{z}
	\end{array}\right)_{\widetilde{z}=\overline{z}, \widetilde{w}=\overline{w}}
	=
	\displaystyle{\frac{1}{\sqrt{2}}}
	\left(\begin{array}{cc}
		x^{1}+ix^{2} & x^{3}-ix^{4} \\
		x^{3}+ix^{4} & x^{1}-ix^{2}
	\end{array}\right), 	
	\medskip \\
	&&~~~\mathrm{d}s^2=\eta_{\mu\nu}^{(\mathbb{U}^{2,2})}\mathrm{d}x^{\mu}\mathrm{d}x^{\nu},~\eta_{\mu\nu}^{(\mathbb{U}^{2,2})}:=\mbox{diag}(1,1,-1,-1)
	.\label{U}
\end{eqnarray}
We now consider a gauge theory on $\mathbb{CM}$ with gauge group $G=\mathrm{GL}(n,\mathbb{C})$.
The covariant derivatives and field strengths are defined by
\begin{eqnarray}
	\label{Field strength_complex}
	F_{mn} := [D_{m}, D_{n}]
	= \partial_{m} A_{n}- \partial_{n} A_{m}+[A_m, A_n], ~~~~
	D_{m}:=\partial_{m} +A_{m},
\end{eqnarray}
where the indices $m, n$ correspond to the complex coordinates $z, \widetilde{z}, w, \widetilde{w}$ and the gauge fields $A_m$  
take values in the Lie algebra $\mathfrak{gl}(n,\mathbb{C})$. 
The anti-self-dual Yang-Mills (ASDYM) equations are given by
\begin{eqnarray}
	\label{ASDYM_complex}
	F_{zw}=[D_{z}, D_{w}]=0, ~
	F_{\widetilde z \widetilde w}=
	\left[D_{\widetilde z}, D_{\widetilde w}\right]=0, ~
	F_{z \widetilde z}-F_{w \widetilde w}=\left[D_{z}, D_{\widetilde{z}}\right] - \left[D_{w}, ~ D_{\widetilde{w}}\right] =0.
\end{eqnarray}
The first two equations of \eqref{ASDYM_complex} is equivalent to the existence of $n \times n$ invertible matrices $h$ and $\widetilde{h}$ such that 
\begin{eqnarray}
	\label{Dh=0}
	D_z h = D_{w}h=0, ~ D_{\widetilde{z}}\widetilde{h}=D_{\widetilde{w}}\widetilde{h}=0.
\end{eqnarray}
Equivalently, the gauge fields can be expressed as
\begin{eqnarray}
	\label{Gauge fields in terms of h}
	A_{z}=-(\partial_{z}h)h^{-1},~
	A_{w}=-(\partial_{w}h)h^{-1},~
	A_{\widetilde{z}}=-(\partial_{\widetilde{z}}\widetilde{h})\widetilde{h}^{-1},~
	A_{\widetilde{w}}=-(\partial_{\widetilde{w}}\widetilde{h})\widetilde{h}^{-1}.~ 
\end{eqnarray}   
Introducing the matrix $J:={\widetilde{h}}^{-1}h$ \cite{BFNY-1978}, known as Yang's $J$-matrix, one verifies that
\begin{eqnarray}
	\partial_{\widetilde{z}}\left[(\partial_{z}J)J^{-1}\right]-
	\partial_{\widetilde{w}}\left[(\partial_{w}J)J^{-1}\right]
	=\widetilde{h}^{-1}(F_{w\widetilde{w}}-F_{z\widetilde{z}})\widetilde{h}.
\end{eqnarray}
By virtue of the third equation in \eqref{ASDYM_complex}, the ASDYM equations \eqref{ASDYM_complex} are satisfied if and only if the Yang equation \cite{BFNY-1978} 
\begin{eqnarray}
	\partial_{\widetilde{z}}\left[(\partial_{z}J)J^{-1}\right]-
	\partial_{\widetilde{w}}\left[(\partial_{w}J)J^{-1}\right] = 0   
\end{eqnarray}
holds.
On the other hand, the Yang equation serves as the equation of motion (Euler--Lagrange equation) of the four-dimensional Wess--Zumino--Witten (WZW$_4$) model \cite{Nair-1991, NS-1992, LMNS-1996, IKU-1997, IKUX-1997}, defined on the four-dimensional real spaces $\mathbb{E}^{4}$, $\mathbb{M}^{3,1}$, $\mathbb{U}_{1}^{2,2}$, and $\mathbb{U}_{2}^{2,2}$. The WZW$_4$ action, denoted by $S_{\mathrm{WZW}_4}$, consists of two parts: the nonlinear sigma model (NL$\sigma$M) term $S_{\sigma}$ and the Wess--Zumino (WZ) term $S_{\mathrm{WZ}}$. For practical purposes, we present the $\mathrm{WZW}_4$ action in its component form rather than its differential-form representation:
\begin{eqnarray}
	S_{\text{WZW$_4$}}&\!\!\!=\!\!\!\!&S_{\sigma} + S_{\text{WZ}},
	\\
	S_{\sigma}&\!\!\!\!:=\!\!\!\!&   \label{NLSM action}
	\frac{-1}{16\pi}\int_{\text{M$_4$}}
	\!\!\mbox{Tr}\!\left[ 
	\begin{array}{l}
		~~\!(\partial_{z}J)J^{-1}(\partial_{z}J)J^{-1} + (\partial_{\widetilde{z}}J)J^{-1}(\partial_{\widetilde{z}}J)J^{-1}
		\\
		\!\!-(\partial_{w}J)J^{-1}(\partial_{w}J)J^{-1} -
		(\partial_{\widetilde{w}}J)J^{-1}(\partial_{\widetilde{w}}J)J^{-1}
	\end{array}
	\!\!\!\right]
	dz d\widetilde{z} dw d\widetilde{w}, ~~~~~~~~
	\\
	S_{\text{WZ}}&\!\!\!\!:=\!\!\!\!&   \label{WZ action}
	\frac{-1}{16\pi}\int_{\text{M$_4$}}
	\!\left\{
	\begin{array}{l}
		~~\!\mbox{Tr}\!\left[
		(\partial_{z}J)J^{-1}
		[~\!(\partial_{w}J)J^{-1}, ~\!(\partial_{\widetilde{w}}J)J^{-1}~\!]
		~\!\right]\!z
		\smallskip \\
		\!\!+\mbox{Tr}\!\left[
		(\partial_{\widetilde{z}}J)J^{-1}
		[~\!(\partial_{w}J)J^{-1}, ~\!(\partial_{\widetilde{w}}J)J^{-1}~\!] 
		~\!\right]\!\widetilde{z}
		\smallskip \\
		\!\!-\mbox{Tr}\!\left[
		(\partial_{w}J)J^{-1}
		[~\!(\partial_{z}J)J^{-1}, ~\!(\partial_{\widetilde{z}}J)J^{-1}~\!] 
		~\!\right]\!w
		\smallskip \\
		\!\!-\mbox{Tr}\!\left[
		(\partial_{\widetilde{w}}J)J^{-1}
		[~\!(\partial_{z}J)J^{-1}, ~\!(\partial_{\widetilde{z}}J)J^{-1}~\!] 
		~\!\right]\!\widetilde{w}
	\end{array}		
	\!\!\right\}dz d\widetilde{z} dw d\widetilde{w}.
\end{eqnarray}
Note that $M_{4}$ denotes any one of the four real spaces $\mathbb{E}^{4}$, $\mathbb{M}^{3,1}$,  $\mathbb{U}_{1}^{2,2}$, and $\mathbb{U}_{2}^{2,2}$. In each case, the complexified coordinates $(z,\widetilde{z},w,\widetilde{w})$ are subject to the corresponding reality conditions (cf. \eqref{Reality condition_E}, \eqref{Reality condition_M}, \eqref{Reality condition_U}, and \eqref{Reality condition_U_2}).

In particular, for $J\in\mathrm{GL}(2, \mathbb{C})$ parametrized as
\begin{eqnarray}
	J:=
	\frac{1}{\Delta}
	\left(
	\begin{array}{cc}
		\Delta^{(11)} & \Delta^{(12)}\\
		\Delta^{(21)} & \Delta^{(22)}
	\end{array}
	\right),
\end{eqnarray}
with $\partial_m|J|=0$ for all $m=z,\widetilde{z},w,\widetilde{w}$, the following identities \cite{HHK-2023}
\begin{eqnarray}
	&&\mbox{Tr}\left[(\partial_m J)J^{-1}(\partial_n J)J^{-1}\right]\nonumber \\
	&\!\!\!\!=\!\!\!\!&
	\frac{1}{|J|\Delta^{2}}
	\left\{
	\left|
	\begin{array}{cc}
		\partial_m \Delta^{(11)} & \partial_m \Delta^{(12)} \\
		\partial_n \Delta^{(21)} & \partial_n \Delta^{(22)}
	\end{array}
	\right|
	+
	\left|
	\begin{array}{cc}
		\partial_n \Delta^{(11)} & \partial_n \Delta^{(12)} \\
		\partial_m \Delta^{(21)} & \partial_m \Delta^{(22)}
	\end{array}
	\right|
	-
	2|J|(\partial_m \Delta)(\partial_n \Delta)
	\right\}, \label{Tr(A_m A_n)}
\end{eqnarray}
and 
\begin{eqnarray}
	\label{Tr(A_m A_n A_p)}	
	\mbox{Tr}\left[(\partial_m J)J^{-1}(\partial_n J)J^{-1}(\partial_p J)J^{-1}\right]   
	=
	\frac{1}{2|J|^2\Delta^{4}}
	\left|
	\begin{array}{cccc}
		\Delta^{(11)} & \Delta^{(12)} & \Delta^{(21)} & \Delta^{(22)} \\
		\partial_m \Delta^{(11)} & \partial_m \Delta^{(12)} & \partial_m \Delta^{(21)} & \partial_m \Delta^{(22)} \\
		\partial_n \Delta^{(11)} & \partial_n \Delta^{(12)} & \partial_n \Delta^{(21)} & \partial_n \Delta^{(22)} \\
		\partial_p \Delta^{(11)} & \partial_p \Delta^{(12)} & \partial_p \Delta^{(21)} & \partial_p \Delta^{(22)}
	\end{array}
	\right| 
\end{eqnarray}
hold.	
Note that, using the row permutation property of the determinant in \eqref{Tr(A_m A_n A_p)} and comparing the resulting expression with \eqref{WZ action}, the WZ term can be further simplified as
\begin{eqnarray}
	\label{WZ action_2}	
	S_{\text{WZ}}=
	\frac{-1}{8\pi}\int_{\text{M$_4$}}
	\!\left\{
	\begin{array}{l}
		~~\!\mbox{Tr}\!\left[
		(\partial_{z}J)J^{-1}
		(\partial_{w}J)J^{-1}(\partial_{\widetilde{w}}J)J^{-1}
		\right]\!z
		\smallskip \\
		\!\!+\mbox{Tr}\!\left[
		(\partial_{\widetilde{z}}J)J^{-1}
		(\partial_{w}J)J^{-1}(\partial_{\widetilde{w}}J)J^{-1} 
		\right]\!\widetilde{z}
		\smallskip \\
		\!\!-\mbox{Tr}\!\left[
		(\partial_{w}J)J^{-1}
		(\partial_{z}J)J^{-1}(\partial_{\widetilde{z}}J)J^{-1}
		\right]\!w
		\smallskip \\
		\!\!-\mbox{Tr}\!\left[
		(\partial_{\widetilde{w}}J)J^{-1}
		(\partial_{z}J)J^{-1}(\partial_{\widetilde{z}}J)J^{-1}
		\right]\!\widetilde{w}
	\end{array}		
	\!\!\right\}dz d\widetilde{z} dw d\widetilde{w}. 
\end{eqnarray}

\section{Quasideterminants}
In this section, we briefly review several equivalent definitions of quasideterminants and summarize the main properties needed for our applications. For a detailed introduction of quasideterminants, please refer to \cite{GR-1991,GGRW-2005} (Cf: \cite{GN-2007, GNS-2008, Huang-2021}). First of all, if we consider an $n \times n$ invertible matrix $X:=(x_{i,j})_{n \times n}$ over a noncommutative ring $\mathcal{R}$, then the
$(i,j)$-th quasideterminant of matrix $X$, denoted by
	\begin{eqnarray}
		\label{quasideterminant}
		\left| X \right|_{ij}
		=
		\left|
		\begin{array}{ccccc}
			x_{1,1} & \!\!\!\cdots  & x_{1,j} &  \!\!\!\cdots & \!\!\!x_{1, n} \\
			\vdots &  \!\!\!\ddots &  \vdots &  \!\!\!\ddots & \!\!\!\vdots \\
			x_{i,1} & \!\!\!\cdots &  \fbox{$x_{i,j}$} &  \!\!\!\cdots & \!\!\!x_{i, n} 
			 \\
			\vdots &  \!\!\!\ddots &  \vdots &  \!\!\!\ddots & \!\!\!\vdots  \\
			x_{n,1} & \!\!\!\cdots &  x_{n,j} &  \!\!\!\cdots & \!\!\!x_{n,n}
		\end{array}\right|~,
	\end{eqnarray}
	is defined by 
	\begin{eqnarray}
		\label{Quasideterminant_defn}
		\left| X \right|_{ij}
		:= x_{i,j} - R_{i,~\!\widehat{j}} \left(X_{\widehat{i},~\! \widehat{j}} \right)^{-1}C_{~\!\widehat{i},~\!j},
	\end{eqnarray}
	where $X_{\widehat{i},,\widehat{j}}$ denotes the submatrix obtained by removing the $i$-th row and the $j$-th column from $X$, and is assumed to be invertible over $\mathcal{R}$. $R_{i,~\!\widehat{j}}$ denotes the $i$-th row of $X$ with the $j$-th entry removed, and $C_{~\!\widehat{i},~\!j}$ denotes the $j$-th column of $X$ with $i$-th entry removed.   
By definition \eqref{Quasideterminant_defn}, one can choose a convenient representation for quasideterminant as
\begin{eqnarray}
\label{Quasideterminant_canonical form}
\left| X \right|_{ij}:=
\left|
\begin{array}{cc}
X_{\widehat{i},~\! \widehat{j}} & C_{~\!\widehat{i},~\!j} \\
R_{i,~\!\widehat{j}} & \fbox{$x_{i,j}$} 
\end{array}
\right|.
\end{eqnarray} 
Using the inverse formula for a $2\times 2$ block matrix, one can also show that quasideterminants are expressed in terms of the entries of the inverse matrix of $X$ as follows:
\begin{eqnarray}
\label{Quasideterminat_inverse matrix}	
\vert X \vert_{ij}=((X^{-1})_{ji})^{-1}.
\end{eqnarray} 
Note that if the entries of $X$ are commutative, one immediately obtains
\begin{eqnarray}
\label{Quasideterminant_commutative limit}
|X|_{ij}=(-1)^{i+j}\frac{\mbox{det}X}{\mbox{det}X_{\widehat{i}, \widehat{j}}}.
\end{eqnarray}
For applicability, we denote $X_{\widehat{i},\widehat{j}}$, $C_{\widehat{i},j}$, $R_{i,\widehat{j}}$, and $x_{ij}$ by $A$, $B$, $C$, and $d$, respectively, and assume that $d$ is an $m \times m$ matrix with commutative entries. By Eqs.~\eqref{Quasideterminant_defn}, \eqref{Quasideterminant_canonical form}, and \eqref{Quasideterminant_commutative limit}, the quasideterminant can be written in the following equivalent forms:
\begin{eqnarray}
	&& \!\!\!\!\left|
	\begin{array}{cc}
		A & B \\
		C & \fbox{$d$}
	\end{array}
	\right|
	\\
	&\!\!\!\!:=\!\!\!\!&
	\label{Quasideterminant_defn_concrete form}
	\left|
	\begin{array}{cc}
		A &  \begin{array}{ccc} B_1 & \!\!\!\dots & \!\!\!B_m \end{array}
		\\
		\begin{array}{c} C_1 \\ \vdots \\ C_m  \end{array} & 
		\fbox{
			$\begin{array}{ccc}
				d_{1,1} & \!\!\!\cdots & \!\!\!d_{1,m} \\
				\vdots & \!\!\!\ddots & \!\!\!\vdots \\
				d_{m,1} & \!\!\!\cdots & \!\!\!d_{m,m}
			\end{array}$
		}
	\end{array}
	\right|
	=\left(
	\begin{array}{cccc}
		d_{1,1} & \!\!\!\cdots & \!\!\!d_{1,m} \\
		\vdots & \!\!\!\ddots & \!\!\!\vdots \\
		d_{m,1} & \!\!\!\cdots & \!\!\!d_{m,m}
	\end{array}
	\!\right)
	-
	\left(
	\begin{array}{c} \!\!C_1 \\ \vdots \\ \!\!C_m  \end{array}
	\!\!\right)
	\!A^{-1}
	\!\left(
	\begin{array}{ccc} \!\!B_1 & \!\!\!\dots & \!\!\!\!B_m \end{array}
	\!\!\right)
	\\
	&\!\!\!\!=\!\!\!\!&
	\label{Quasideterminant_defn_expansion}
	\left(
	\begin{array}{cccc}
		\!\!\left|
		\begin{array}{cc}
			A & B_1 \\
			C_1 & \fbox{$d_{1,1}$}
		\end{array}
		\right|
		& 
		\!\!\! \cdots
		&
		\!\!\!\left|
		\begin{array}{cc}
			A & B_m \\
			C_1 & \fbox{$d_{1,m}$}
		\end{array}
		\right|
		\medskip \\
		\!\!\vdots & \!\!\!\ddots & \!\!\!\vdots 
		\\
		\!\!\left|
		\begin{array}{cc}
			A & B_1 \\
			C_m & \fbox{$d_{m,1}$}
		\end{array}
		\right|
		& 
		\!\!\!\cdots
		&
		\!\!\!\left|
		\begin{array}{cc}
			A & B_m \\
			C_m & \fbox{$d_{m,m}$}
		\end{array}
		\right|
	\end{array}
	\!\!\right)
	=\frac{1}{|A|}
	\left(
	\begin{array}{ccc}
		\!\!\left|
		\begin{array}{cc}
			A & B_1 \\
			C_1 & d_{1,1}
		\end{array}
		\right|
		& 
		\!\!\!\cdots
		&
		\!\!\!\left|
		\begin{array}{cc}
			A & B_m \\
			C_1 & d_{1,m}
		\end{array}
		\right|
		\medskip \\
		\!\!\vdots & \!\!\!\ddots & \!\!\!\vdots 
		\\
		\!\!\left|
		\begin{array}{cc}
			A & B_1 \\
			C_m & d_{m,1}
		\end{array}
		\right|
		& 
		\!\!\!\cdots
		&
		\!\!\!\left|
		\begin{array}{cc}
			A & B_m \\
			C_m & d_{m,m}
		\end{array}
		\right|
	\end{array}
	\!\!\right). ~~~~~~
\end{eqnarray}
Next, we present some elementary operations of quasideterminants, which are similar to, but not exactly the same as, those of determinants.
\newtheorem{prop_2.1}{Proposition}[section]
\begin{prop_2.1}
	{\bf (Row (Column) Permutation Rule)}  \label{Prop_2.3} \\
	Interchanging two rows (or two columns) of a quasideterminant leaves it unchanged.
\end{prop_2.1}
\newtheorem{prop_2.2}[prop_2.1]{Proposition} 
\begin{prop_2.2} {\bf{(Right (Left) Multiplication Rule)}\label{Prop_2.4}} 		
		\begin{eqnarray}
			\left|
			\begin{array}{ccccc}
				x_{1,1}\Lambda_{1} & \!\! \cdots & \!\! x_{1,j}\Lambda_{j} & \!\!\cdots & \!\! x_{1,n}\Lambda_{n} \\
				\vdots &  \!\!\ddots & \!\!\vdots & \!\!\ddots & \!\!\vdots  \\
				x_{i,1}\Lambda_{1} & \!\!\cdots & \!\!\fbox{$x_{i,j}\Lambda_{j}$} & \!\!\cdots & \!\!x_{i,n}\Lambda_{n} \\
				\vdots &  \!\!\ddots & \!\!\vdots & \!\!\ddots & \!\!\vdots  \\
				x_{n,1}\Lambda_{1} & \!\!\cdots & \!\!x_{n,j}\Lambda_{j} & \!\!\cdots & \!\!x_{n,n}\Lambda_{n}
			\end{array}\right|
			=
			\left|
			\begin{array}{ccccc}
				x_{1,1} & \!\! \cdots & \!\! x_{1,j} & \!\! \cdots & \!\! x_{1,n} \\
				\vdots &  \!\! \ddots & \!\! \vdots & \!\! \ddots & \!\! \vdots  \\
				x_{i,1} & \!\! \cdots & \!\! \fbox{$x_{i,j}$} & \!\! \cdots & \!\! x_{i,n} \\
				\vdots &  \!\! \ddots & \!\! \vdots & \!\! \ddots & \!\! \vdots  \\
				x_{n,1} & \!\! \cdots & \!\! x_{n,j} & \!\! \cdots & \!\! x_{n,n}
			\end{array}\right|\Lambda_{j},
		\end{eqnarray}
	and
		\begin{eqnarray}
			\label{left multiplication law}
			\left|
			\begin{array}{ccccc}
				\Lambda_{1}x_{1,1} & \!\!\cdots & \!\!\Lambda_{1}x_{1,j} & \!\!\cdots & \!\!\Lambda_{1} x_{1,n} \\
				\vdots &  \!\!\ddots & \!\!\vdots & \!\!\ddots & \!\!\vdots  \\
				\Lambda_{i}x_{i,1} & \!\!\cdots & \!\!\fbox{$\Lambda_{i}x_{i,j}$} & \!\!\cdots & \!\!\Lambda_{i}x_{i,n} \\
				\vdots &  \!\!\ddots & \!\!\vdots & \!\!\ddots & \!\!\vdots  \\
				\Lambda_{n}x_{n,1} & \!\!\cdots & \!\!\Lambda_{n}x_{n,j} & \!\!\cdots & \!\!\Lambda_{n}x_{n,n}
			\end{array}\right|	
			=
			\Lambda_{i} \left|
			\begin{array}{ccccc}
				x_{1,1} & \!\! \cdots & \!\! x_{1,j} & \!\!\cdots & \!\!x_{1,n} \\
				\vdots &  \!\! \ddots & \!\! \vdots & \!\!\ddots & \!\!\vdots  \\
				x_{i,1} & \!\!\cdots & \!\!\fbox{$x_{i,j}$} & \!\!\cdots & \!\!x_{i,n} \\
				\vdots &  \!\!\ddots & \!\!\vdots & \!\!\ddots & \!\!\vdots  \\
				x_{n,1} & \!\!\cdots & \!\!x_{n,j} & \!\!\cdots & \!\!x_{n,n}
			\end{array}\right| ~~~~
		\end{eqnarray}
hold provided that $\Lambda_{k}$ ($k=1, ~2, ..., ~n$) is invertible. 	
\end{prop_2.2}
\newtheorem{prop_2.3}[prop_2.1]{Prposition}
\begin{prop_2.3}
	{\bf (Row (Column) Operation)} \label{Prop_2.5} \\
Suppose that the $l$-th row (or column) of a quasideterminant does not involve the boxed entry.
Then adding the $l$-th row (or column) to the $k$-th row (or column) leaves the quasideterminant unchanged.
\end{prop_2.3}
A general rule for reducing a higher-order quasideterminant to one of lower order is known as the noncommutative version of Sylvester's theorem \cite{GR-1991, GGRW-2005} (Cf:\cite{GN-2007}).
The simplest form of this theorem is given by the following noncommutative analogue of the Jacobi identity.
\newtheorem{prop_2.4}[prop_2.1]{Propositon}
\begin{prop_2.4}
	{\bf (Noncommutative Jacobi Identity)} \label{Jacobi identity}
	\begin{eqnarray}
		\label{Jacobi Identity}
		\left| X \right|_{i,j}
		=
		\left|
		\begin{array}{cc:c:c:c:cc}
			\!\!x_{1,1} & \!\!\!\! \cdots & \! x_{1,m} & \! \cdots & \! x_{1,j} & \! \cdots & \!\!\!\!\! x_{1,n} \\
			\!\!\vdots & \!\!\!\! \ddots & \! \vdots & \! \ddots & \! \vdots & \! \ddots & \!\!\!\!\! \vdots \\
			\hdashline
			\!\!x_{\ell,1} & \!\!\!\! \cdots & \! x_{\ell,m} & \! \cdots & \! x_{\ell,j} & \! \cdots & \!\!\!\!\! x_{\ell,n} \\
			\hdashline
			\!\!\vdots & \!\!\!\! \ddots & \! \vdots & \! \ddots & \! \vdots & \! \ddots & \!\!\!\!\! \vdots \\
			\hdashline
			\!\!x_{i,1} & \!\!\!\! \cdots & \! x_{i,m} & \! \cdots & \! \fbox{$x_{i,j}$} & \! \cdots & \!\!\!\!\! x_{i,n} \\
			\hdashline
			\!\!\vdots & \!\!\!\! \ddots & \! \vdots & \! \ddots & \! \vdots & \! \ddots & \!\!\!\!\! \vdots \\
			\!\!x_{n,1} & \!\!\!\! \cdots & \! x_{n,m} & \! \cdots & \! x_{n,j} & \! \cdots & \!\!\!\!\! x_{n,n} 
		\end{array}
		\right| 
		= 
		\left|\begin{array}{cc}
	     \left|
	     X_{\widehat{i},~\!\widehat{j}}
	     \right|_{\ell,m}  
	     &
		\left|
		X_{\widehat{i},~\!\widehat{m}}
		\right|_{\ell,j}
		\medskip \\
		\left|
		X_{\widehat{\ell},~\!\widehat{j}}
		\right|_{i,m}
		&
		\fbox{
		$\left|
		X_{\widehat{\ell},~\!\widehat{m}}
		\right|_{i,j}$
	}
	\end{array}
\right|.
	\end{eqnarray}
\end{prop_2.4}
The following homological relations provide a powerful tool for displacing the boxed entry horizontally or vertically in nontrivial quasideterminant calculations.
\newtheorem{prop_2.5}[prop_2.1]{Proposition}
\begin{prop_2.5}
	{\bf (Homological Relations \cite{GNS-2008}, Cf: \cite{GR-1991})}  \label{Prop_2.5}
		\begin{eqnarray}
			\label{Homological relation}
			\left|
			\begin{array}{ccccc}
				\!\! x_{1,1} & \!\!\! \cdots & \!\!\!\! x_{1,j} & \!\!\! \cdots & \!\!\!\! x_{1,n} \\
				\!\! \vdots  & \!\!\! \ddots & \!\!\!\! \vdots  & \!\!\! \ddots & \!\!\!\! \vdots \\
				\!\! x_{m,1} & \!\!\! \cdots & \!\!\!\! x_{m,j} & \!\!\! \cdots & \!\!\!\! x_{m,n} \\
				\!\! \vdots  & \!\!\! \ddots & \!\!\!\! \vdots  & \!\!\! \ddots &  \!\!\!\! \vdots  \\
				\!\! x_{\ell,1} & \!\!\! \cdots & \!\!\!\! \fbox{$x_{\ell,j}$}  &  \!\!\! \cdots & \!\!\!\! x_{\ell,n} \\
				\!\! \vdots  & \!\!\! \ddots & \!\!\!\! \vdots  & \!\!\! \ddots &  \!\!\!\! \vdots  \\
				\!\! x_{n,1} & \!\!\! \cdots & \!\!\!\! x_{n,j} & \!\!\! \cdots & \!\!\!\! x_{n,n}
			\end{array}\right|
			=
			\left|
			\begin{array}{ccccc}
				\!\! x_{1,1} & \!\!\! \cdots & \!\!\!\! 0        & \!\!\! \cdots & \!\!\!\! x_{1,n} \\ 
				\!\! \vdots  & \!\!\! \ddots & \!\!\!\! \vdots   & \!\!\! \ddots & \!\!\!\! \vdots \\
				\!\! x_{m,1} & \!\!\! \cdots & \!\!\!\! 1        & \!\!\! \cdots & \!\!\!\! x_{m,n} \\
				\!\! \vdots  & \!\!\! \ddots & \!\!\!\! \vdots   & \!\!\! \ddots & \!\!\!\! \vdots  \\
				\!\! x_{\ell,1} & \!\!\! \cdots & \!\!\!\! \fbox{$0$} & \!\!\! \cdots & \!\!\!\! x_{\ell,n} \\
				\!\! \vdots  & \!\!\! \ddots & \!\!\!\! \vdots   & \!\!\! \ddots & \!\!\!\! \vdots  \\
				\!\! x_{n,1} & \!\!\! \cdots & \!\!\!\! 0        & \!\!\! \cdots & \!\!\!\! x_{n,n}
			\end{array}\right|
			\left|
			\begin{array}{ccccc}
				\!\! x_{1,1} & \!\!\! \cdots & \!\!\!\! x_{1,j} & \!\!\! \cdots & \!\!\!\! x_{1,n} \\
				\!\! \vdots  & \!\!\! \ddots & \!\!\!\! \vdots  & \!\!\! \ddots & \!\!\!\! \vdots \\
				\!\! x_{m,1} & \!\!\! \cdots & \!\!\!\! \fbox{$x_{m,j}$} &  \!\!\! \cdots & \!\!\!\! x_{m,n} \\
				\!\! \vdots  & \!\!\! \ddots & \!\!\!\! \vdots  & \!\!\! \ddots & \!\!\!\! \vdots  \\
				\!\! x_{\ell,1} & \!\!\! \cdots & \!\!\!\! x_{\ell,j} & \!\!\! \cdots & \!\!\!\! x_{\ell,n} \\
				\!\! \vdots  & \!\!\! \ddots & \!\!\!\! \vdots  & \!\!\! \ddots & \!\!\!\! \vdots  \\
				\!\! x_{n,1} & \!\!\! \cdots & \!\!\!\! x_{n,j} & \!\!\! \cdots & \!\!\!\! x_{n,n}
			\end{array}\right|,
		\end{eqnarray}
		\begin{eqnarray}
			\label{Homological relation 2}
			&&\left| 
			\begin{array}{ccccccc}
				\!\! x_{1,1} & \!\!\! \cdots & \!\!\!\! x_{1,m} & \!\!\! \cdots & \!\!\!\! x_{1,\ell} & \!\!\! \cdots & \!\!\!\! x_{1,n} \\
				\!\! \vdots  & \!\!\! \ddots & \!\!\!\! \vdots  & \!\!\! \ddots & \!\!\!\! \vdots & \!\!\! \ddots & \!\!\!\! \vdots \\
				\!\! x_{i,1} & \!\!\! \cdots & \!\!\!\! x_{i,m} & \!\!\! \cdots & \!\!\!\! \fbox{$x_{i,\ell}$} & \!\!\! \cdots & \!\!\!\! x_{i,n} \\
				\!\! \vdots  & \!\!\! \ddots & \!\!\!\! \vdots  & \!\!\! \ddots & \!\!\!\! \vdots & \!\!\! \ddots & \!\!\!\! \vdots \\
				\!\! x_{n,1} & \!\!\! \cdots & \!\!\!\! x_{n,m} & \!\!\! \cdots & \!\!\!\! x_{n,\ell} & \!\!\! \cdots & \!\!\!\! x_{n,n}
			\end{array}
			\right|  \nonumber \\
			&=&
			\left| 
			\begin{array}{ccccccc}
				\!\! x_{1,1} & \!\!\! \cdots & \!\!\!\! x_{1,m} & \!\!\! \cdots & \!\!\!\! x_{1,\ell} & \!\!\! \cdots & \!\!\!\! x_{1,n} \\
				\!\! \vdots  & \!\!\! \ddots & \!\!\!\! \vdots  & \!\!\! \ddots & \!\!\!\! \vdots  & \!\!\! \ddots & \!\!\!\! \vdots \\
				\!\! x_{i,1} & \!\!\! \cdots & \!\!\!\! \fbox{$x_{i,m}$} & \!\!\! \cdots & \!\!\!\! x_{i,\ell} & \!\!\! \cdots & \!\!\!\! x_{i,n} \\
				\!\! \vdots  & \!\!\! \ddots & \!\!\!\! \vdots  & \!\!\! \ddots & \!\!\!\! \vdots  & \!\!\! \ddots & \!\!\!\! \vdots \\
				\!\! x_{n,1} & \!\!\! \cdots & \!\!\!\! x_{n,m} & \!\!\! \cdots & \!\!\!\! x_{n,\ell} & \!\!\! \cdots & \!\!\!\! x_{n,n}
			\end{array}
			\right|
			\left| 
			\begin{array}{ccccccc}
				\!\! x_{1,1} & \!\!\! \cdots & \!\!\!\! x_{1,m} & \!\!\! \cdots & \!\!\!\! x_{1,\ell} & \!\!\! \cdots & \!\!\!\! x_{1,n} \\
				\!\! \vdots  & \!\!\! \ddots & \!\!\!\! \vdots  & \!\!\! \ddots & \!\!\!\! \vdots  & \!\!\! \ddots & \!\!\!\! \vdots \\
				\!\! 0       & \!\!\! \cdots & \!\!\!\! 1       & \!\!\! \cdots & \!\!\!\! \fbox{$0$}& \!\!\! \cdots & \!\!\!\! 0 \\
				\!\! \vdots  & \!\!\! \ddots & \!\!\!\! \vdots  & \!\!\! \ddots & \!\!\!\! \vdots  & \!\!\! \ddots & \!\!\!\! \vdots \\
				\!\! x_{n,1} & \!\!\! \cdots & \!\!\!\! x_{n,m} & \!\!\! \cdots & \!\!\!\! x_{n,\ell} & \!\!\! \cdots & \!\!\!\! x_{n,n}
			\end{array}
			\right|.~~~~~~~
		\end{eqnarray}
\end{prop_2.5}
Finally, we present a useful formula for computing the derivatives of Grammian-type quasideterminants (quasi-Grammians).
\newtheorem{Lem_2.6}[prop_2.1]{Lemma}
\begin{Lem_2.6}
	{\bf (Derivative Formula of Quasi-Grammian \cite{GN-2007})}
	\medskip \\
	Let $A$ be an $n \times n$ matrix, $B$ be an $n \times 1$ column matrix, $C$ be a $1 \times n$ row matrix, and $d$ be a $1 \times 1$ matrix.
	If the derivative of matrix A can be expressed as
	\begin{eqnarray}
		\partial A=\sum_{\ell=1}^{k}E_{\ell}F_{\ell},
	\end{eqnarray}
	where $E_{\ell}$ and $F_{\ell}$ stand for certain square matrices,
	then we have the following derivative formula of quasideterminant:
	\begin{eqnarray}
		\label{Derivative formula_grammian}
		&&\partial\left|
		\begin{array}{cc}
			A & B \\
			C & \fbox{$d$}
		\end{array}
		\right|
		= 
		\partial d
		+
		\left|
		\begin{array}{cc}
			A & B \\
			\partial C & \fbox{$0$}
		\end{array}
		\right|
		+
		\left|
		\begin{array}{cc}
			A & \partial B \\
			C & \fbox{$0$}
		\end{array}
		\right|
		+
		\sum_{\ell=1}^{k}
		\left|
		\begin{array}{cc}
			A & E_{\ell} \\
			C & \fbox{$0$}
		\end{array}
		\right|
		\left|
		\begin{array}{cc}
			A & B \\
			F_{\ell} & \fbox{$0$}
		\end{array}
		\right|.
	\end{eqnarray}
\end{Lem_2.6}

\section{Quasi-Grammian Solutions and Unitarity}

In this section, we slightly modify the linear systems first introduced in \cite{NGO-2000} so that they are compatible with the four real spaces. Based on these linear systems, we present, in Theorem 4.1, 
 construction of quasi-Grammian solutions to the ASDYM equation. 
The Cauchy matrix approach developed in \cite{LQYZ-2022, LQZ-2023} is naturally incorporated into the present construction.
Finally, we discuss the conditions under which the resulting solutions are unitary.  
\newtheorem{Thm_4.1}{Theorem}[section]
\begin{Thm_4.1} {\bf (General Solutions on $\mathbb{CM}$, $\mbox{G}=\mbox{GL}(n, \mathbb{C})$, (Cf:\cite{NGO-2000, LQYZ-2022, LQZ-2023, Ohta-2024, LHHZ-2025}))}
\medskip \\
	Let $(z,\widetilde{z},w,\widetilde{w})$ denote the coordinates on $\mathbb{CM}$.
	Let $J_{0}, K_{0}\in\mathbb{C}_{n\times n}[z,\widetilde{z},w,\widetilde{w}], \theta,\rho\in\mathbb{C}_{n\times N}[z,\widetilde{z},w,\widetilde{w}]$, and $\Lambda,\Xi\in\mathbb{C}_{N\times N}$.
	Suppose that there exist eigenfunction--eigenvalue pairs
	$(\theta,\Lambda)$ and $(\rho,\Xi)$ satisfying the following linear systems:
	\begin{eqnarray}
		\label{Linear systems of ASDYM}
		\left\{
		\begin{array}{l}
			\left[\partial_{z}-(\partial_{z}J_{0})J_{0}^{-1}\right]\theta
			-(\partial_{\widetilde{w}}\theta)\Lambda=0
			\smallskip\\
			\left[\partial_{w}-(\partial_{w}J_{0})J_{0}^{-1}\right]\theta
			-(\partial_{\widetilde{z}}\theta)\Lambda=0
		\end{array},
		\right.
		~~
		\left\{
		\begin{array}{l}
			\left[\partial_{z}-(\partial_{z}J_{0}^{-T})J_{0}^{T}\right]\rho
			-(\partial_{\widetilde{w}}\rho)\Xi=0
			\smallskip\\
			\left[\partial_{w}-(\partial_{w}J_{0}^{-T})J_{0}^{T}\right]\rho
			-(\partial_{\widetilde{z}}\rho)\Xi=0
		\end{array}.
		\right.
	\end{eqnarray}
	Assume furthermore that there exists
	$\Omega(\theta,\rho)\in\mathbb{C}_{N\times N}[z,\widetilde{z},w,\widetilde{w}]$
	satisfying the Sylvester equation
	\begin{eqnarray}
		\label{Sylvester eq}
		\Xi^{T}\Omega(\theta,\rho)-\Omega(\theta,\rho)\Lambda
		=\rho^{T}\theta,
	\end{eqnarray}
	where $\Lambda$ and $\Xi$ are assumed to have no common eigenvalues.
	
	Define
	\begin{eqnarray}
		\label{J_CM}
		J
		:=
		\left|
		\begin{array}{cc}
			\Xi^{T}\Omega(\theta,\rho) & \rho^{T}
			\smallskip\\
			\theta & \fbox{$I$}
		\end{array}
		\right|
		J_{0},
		\qquad
		K
		:=
		\left|
		\begin{array}{cc}
			\Omega(\theta,\rho) & \rho^{T}
			\smallskip\\
			\theta & \fbox{$K_{0}$}
		\end{array}
		\right|
		\in\mathbb{C}_{n\times n}[z,\widetilde{z},w,\widetilde{w}].
	\end{eqnarray}
	If $J_{0}$ and $K_{0}$ satisfy
	\begin{eqnarray}
		\partial_{\widetilde{z}}K_{0}
		=(\partial_{w}J_{0})J_{0}^{-1},
		\qquad
		\partial_{\widetilde{w}}K_{0}
		=(\partial_{z}J_{0})J_{0}^{-1},
	\end{eqnarray}
	then $J$ and $K$ satisfy
	\begin{eqnarray}
		\partial_{\widetilde{z}}K
		=(\partial_{w}J)J^{-1},
		\qquad
		\partial_{\widetilde{w}}K
		=(\partial_{z}J)J^{-1}.
	\end{eqnarray}
	That is, $J$ is a solution of the Yang equation
	\begin{eqnarray}
		\label{Yang eq_CM}
		\partial_{\widetilde{z}}\left[ (\partial_{z}J)J^{-1} \right]
		-
		\partial_{\widetilde{w}}\left[ (\partial_{w}J)J^{-1} \right]=0.
	\end{eqnarray}
\end{Thm_4.1}
\textit{Proof:} \medskip \\
For convenience, we use $\Omega$ simply as the abbreviation of $\Omega(\theta, \rho)$ in the following proof.
First of all, one can show that \cite{LHHZ-2025}
\begin{eqnarray}
	\label{Omega_derivative}
	\left\{
	\begin{array}{l}
		\partial_{z}\Omega
		=\Xi^{T}(\partial_{\widetilde{w}}\Omega)-\rho^{T}(\partial_{\widetilde{w}}\theta)
		=(\partial_{\widetilde{w}}\Omega)\Lambda + (\partial_{\widetilde{w}}\rho^{T})\theta
		\smallskip \\
		\partial_{w}\Omega
		=\Xi^{T}(\partial_{\widetilde{z}}\Omega)-\rho^{T}(\partial_{\widetilde{z}}\theta)
		=(\partial_{\widetilde{z}}\Omega)\Lambda + (\partial_{\widetilde{z}}\rho^{T})\theta
	\end{array}.
	\right.
\end{eqnarray}
The detailed verification is given in Appendix A.
By the derivative formula of quasi-Grammian (Lemma 3.6) and \eqref{Omega_derivative}, we have
\begin{eqnarray}
	\label{partial J}
	\partial_{z}J
	&\!\!\!\!=\!\!\!\!&
	(\partial_{z}
	\left|
	\begin{array}{cc}
		\Xi^{T}\Omega & \rho^{T} 
		\smallskip \\
		\theta  & \fbox{$I$}
	\end{array}
	\right|)J_{0} +
	\left|
	\begin{array}{cc}
		\Xi^{T}\Omega & \rho^{T} 
		\smallskip \\
		\theta  & \fbox{$I$}
	\end{array}
	\right|(\partial_{z}J_{0})
	\\
	&\!\!\!\!=\!\!\!\!&
	\begin{array}{l}
		\left[
		\begin{array}{l}
			\partial_{z}I +
			\left|
			\begin{array}{cc}
				\Xi^{T}\Omega & \rho^{T} 
				\smallskip \\
				\partial_{z}\theta  & \fbox{$0$}
			\end{array}
			\right|
			+
			\left|
			\begin{array}{cc}
				\Xi^{T}\Omega & \partial_{z}\rho^{T} 
				\smallskip \\
				\theta  & \fbox{$0$}
			\end{array}
			\right|
			\smallskip \\
			+
			\left|
			\begin{array}{cc}
				\Omega & \partial_{\widetilde{w}}\rho^{T} 
				\smallskip \\
				\theta  & \fbox{$0$}
			\end{array}
			\right|
			\left|
			\begin{array}{cc}
				\Xi^{T}\Omega & \rho^{T} 
				\smallskip \\
				\theta  & \fbox{$0$}
			\end{array}
			\right|
			+
			\left|
			\begin{array}{cc}
				\Omega & \partial_{\widetilde{w}}\Omega
				\smallskip \\
				\theta  & \fbox{$0$}
			\end{array}
			\right|
			\left|
			\begin{array}{cc}
				\Xi^{T}\Omega & \rho^{T} 
				\smallskip \\
				\Lambda  & \fbox{$0$}
			\end{array}
			\right|
		\end{array}
		\right]J_{0}
		\smallskip \\
		+\left[
		\left|
		\begin{array}{cc}
			\Xi^{T}\Omega & \rho^{T} 
			\smallskip \\
			\theta  & \fbox{$I$}
		\end{array}
		\right|(\partial_{z}J_{0})J_{0}^{-1}
		\right]J_{0}
	\end{array}
\end{eqnarray}
By \eqref{Linear systems of ASDYM}, we have
\begin{eqnarray}
	\left\{
	\begin{array}{l}
		\partial_{z}\theta = (\partial_{\widetilde{w}}\theta)\Lambda + (\partial_{z}J_{0})J_{0}^{-1}\theta
		\smallskip \\
		\partial_{z}\rho^{T}
		= \Xi^{T}\partial_{\widetilde{w}}\rho^{T}-\rho^{T}(\partial_{z}J_{0})J_{0}^{-1}
	\end{array}.
	\right.
\end{eqnarray}
Therefore,
\begin{eqnarray}
	&&\!\!\!\!\left|
	\begin{array}{cc}
		\Xi^{T}\Omega & \rho^{T} 
		\smallskip \\
		\partial_{z}\theta  & \fbox{$0$}
	\end{array}
	\right|
	+
	\left|
	\begin{array}{cc}
		\Xi^{T}\Omega & \partial_{z}\rho^{T} 
		\smallskip \\
		\theta  & \fbox{$0$}
	\end{array}
	\right|
	\\
	&\!\!\!\!=\!\!\!\!&
	\left|
	\begin{array}{cc}
		\Xi^{T}\Omega & \rho^{T} 
		\smallskip \\
		(\partial_{\widetilde{w}}\theta)\Lambda  & \fbox{$0$}
	\end{array}
	\right|
	+
	\left|
	\begin{array}{cc}
		\Xi^{T}\Omega & \rho^{T} 
		\smallskip \\
		(\partial_{z}J_{0})J_{0}^{-1}\theta  & \fbox{$0$}
	\end{array}
	\right|
	+
	\left|
	\begin{array}{cc}
		\Xi^{T}\Omega & \Xi^{T}\partial_{\widetilde{w}}\rho^{T} 
		\smallskip \\
		\theta  & \fbox{$0$}
	\end{array}
	\right|
	-
		\left|
	\begin{array}{cc}
		\Xi^{T}\Omega & \rho^{T}(\partial_{z}J_{0})J_{0}^{-1} 
		\smallskip \\
		\theta  & \fbox{$0$}
	\end{array}
	\right| ~~~~~~~
	\\
	&\!\!\!\!=\!\!\!\!&
	\left|
	\begin{array}{cc}
		\Xi^{T}\Omega & \rho^{T} 
		\smallskip \\
		(\partial_{\widetilde{w}}\theta)\Lambda  & \fbox{$0$}
	\end{array}
	\right|
	+
	\left|
	\begin{array}{cc}
	\Omega & \partial_{\widetilde{w}}\rho^{T} 
		\smallskip \\
		\theta  & \fbox{$0$}
	\end{array}
	\right|
	+
	\left[
	(\partial_{z}J_{0})J_{0}^{-1}, ~
	\left|
	\begin{array}{cc}
		\Xi^{T}\Omega & \rho^{T} 
		\smallskip \\
		\theta  & \fbox{$I$}
	\end{array}
	\right|
	\right].
\end{eqnarray}
Comparing with \eqref{partial J}, we obtain
\begin{eqnarray}
	\partial_{z}J
	&\!\!\!\!=\!\!\!\!&
	\left[
	\begin{array}{l}
		\left|
		\begin{array}{cc}
			\Xi^{T}\Omega & \rho^{T} 
			\smallskip \\
			(\partial_{\widetilde{w}}\theta)\Lambda  & \fbox{$0$}
		\end{array}
		\right|
		+
		\left|
		\begin{array}{cc}
			\Omega & \partial_{\widetilde{w}}\rho^{T} 
			\smallskip \\
			\theta  & \fbox{$0$}
		\end{array}
		\right|
		\smallskip \\
		+
		\left|
		\begin{array}{cc}
			\Omega & \partial_{\widetilde{w}}\rho^{T} 
			\smallskip \\
			\theta  & \fbox{$0$}
		\end{array}
		\right|
		\left|
		\begin{array}{cc}
			\Xi^{T}\Omega & \rho^{T} 
			\smallskip \\
			\theta  & \fbox{$0$}
		\end{array}
		\right|
		+
		\left|
		\begin{array}{cc}
			\Omega & \partial_{\widetilde{w}}\Omega
			\smallskip \\
			\theta  & \fbox{$0$}
		\end{array}
		\right|
		\left|
		\begin{array}{cc}
			\Xi^{T}\Omega & \rho^{T} 
			\smallskip \\
			\Lambda  & \fbox{$0$}
		\end{array}
		\right|
		\smallskip \\
		+(\partial_{z}J_{0})J_{0}^{-1}
		\left|
		\begin{array}{cc}
			\Xi^{T}\Omega & \rho^{T} 
			\smallskip \\
			\theta  & \fbox{$I$}
		\end{array}
		\right|
	\end{array}
	\right]J_{0}.
\end{eqnarray}
On the other hand, by \eqref{Sylvester eq} we have
\begin{eqnarray}
	\left|
	\begin{array}{cc}
		\Xi^{T}\Omega & \rho^{T} 
		\smallskip \\
		(\partial_{\widetilde{w}}\theta)\Lambda  & \fbox{$0$}
	\end{array}
	\right|
	&\!\!\!\!=\!\!\!\!&
	\left|
	\begin{array}{cc}
		\Omega & (\Omega\Lambda)\Omega^{-1}\Xi^{-T}\rho^{T} 
		\smallskip \\
		\partial_{\widetilde{w}}\theta  & \fbox{$0$}
	\end{array}
	\right|
	=
	\left|
	\begin{array}{cc}
		\Omega & (\Xi^{T}\Omega-\rho^{T}\theta)\Omega^{-1}\Xi^{-T}\rho^{T} 
		\smallskip \\
		\partial_{\widetilde{w}}\theta  & \fbox{$0$}
	\end{array}
	\right| ~~~~~~~~~~~
	\\
	&\!\!\!\!=\!\!\!\!&
	\left|
	\begin{array}{cc}
		\Omega & \rho^{T} 
		\smallskip \\
		\partial_{\widetilde{w}}\theta  & \fbox{$0$}
	\end{array}
	\right|
	+
	\left|
	\begin{array}{cc}
		\Omega & \rho^{T} 
		\smallskip \\
		\partial_{\widetilde{w}}\theta  & \fbox{$0$}
	\end{array}
	\right|
	\left|
	\begin{array}{cc}
		 \Xi^{T}\Omega & \rho^{T} 
		\smallskip \\
		\theta  & \fbox{$0$}
	\end{array}
	\right|
	\\
	&\!\!\!\!=\!\!\!\!&
	\left|
	\begin{array}{cc}
		\Omega & \rho^{T} 
		\smallskip \\
		\partial_{\widetilde{w}}\theta  & \fbox{$0$}
	\end{array}
	\right|
	\left|
	\begin{array}{cc}
		\Xi^{T}\Omega & \rho^{T} 
		\smallskip \\
		\theta  & \fbox{$I$}
	\end{array}
	\right|.
\end{eqnarray}
Similarly, we have
\begin{eqnarray}
	\left|
	\begin{array}{cc}
		\Xi^{T}\Omega & \rho^{T} 
		\smallskip \\
		\Lambda  & \fbox{$0$}
	\end{array}
	\right|
	=
	\left|
	\begin{array}{cc}
		\Omega & \rho^{T} 
		\smallskip \\
		1  & \fbox{$0$}
	\end{array}
	\right|
	\left|
	\begin{array}{cc}
		\Xi^{T}\Omega & \rho^{T} 
		\smallskip \\
		\theta  & \fbox{$I$}
	\end{array}
	\right|.
\end{eqnarray}
Now we can rewrite $\partial_{z}J$ as
\begin{eqnarray}
	\partial_{z}J
	&\!\!\!\!=\!\!\!\!&
	\left[
	\begin{array}{l}
		\left|
		\begin{array}{cc}
			\Omega & \rho^{T} 
			\smallskip \\
			\partial_{\widetilde{w}}\theta  & \fbox{$0$}
		\end{array}
		\right|
		+
		\left|
		\begin{array}{cc}
			\Omega & \partial_{\widetilde{w}}\rho^{T} 
			\smallskip \\
			\theta  & \fbox{$0$}
		\end{array}
		\right|
		+
		\left|
		\begin{array}{cc}
			\Omega & \partial_{\widetilde{w}}\Omega 
			\smallskip \\
			\theta  & \fbox{$0$}
		\end{array}
		\right|
		\left|
		\begin{array}{cc}
			\Omega & \rho^{T} 
			\smallskip \\
			1  & \fbox{$0$}
		\end{array}
		\right|
		\medskip \\
		+ (\partial_{z}J_{0})J_{0}^{-1}
	\end{array}
	\right]
	\left|
	\begin{array}{cc}
		\Xi^{T}\Omega & \rho^{T} 
		\smallskip \\
		\theta  & \fbox{$I$}
	\end{array}
	\right|J_{0} ~~~~~
	\\
	&\!\!\!\!=\!\!\!\!&
	\partial_{\widetilde{w}}
	\left|
	\begin{array}{cc}
		\Omega & \rho^{T} 
		\smallskip \\
		\theta  & \fbox{$K_{0}$}
	\end{array}
	\right|J
	=
	(\partial_{\widetilde{w}}K)J
\end{eqnarray}
Here we use Lemma 3.6 with respect to the derivative of $K$.
Similarly, we also have
\begin{eqnarray}
	\partial_{w}J=(\partial_{\widetilde{z}}K)J.
\end{eqnarray}
By the compatibility  $\partial_{\widetilde{z}}\partial_{\widetilde{w}}K=\partial_{\widetilde{w}}\partial_{\widetilde{z}}K$, we find that $J$ satisfies
the Yang equation \eqref{Yang eq_CM}.
$\hfill\Box$ \\
Note that the determinant of $J$ depends only on the spectral parameters $\Lambda$ and $\Xi$, and is given by
\begin{eqnarray}
\label{det(J)}	
	|J|=|\Xi^{-T}\Lambda|.
\end{eqnarray}
This identity follows directly from Sylvester's determinant identity,
\begin{eqnarray}
\label{Sylvester det}	
	\det(I_{m}+A_{m\times n}B_{n\times m})
	=
	\det(I_{n}+B_{n\times m}A_{m\times n}).
\end{eqnarray}

\newtheorem{Prop_4.2}[Thm_4.1]{Proposition}
\begin{Prop_4.2} 
{\bf (Unitary Solutions on $\mathbb{CM}$, $\mbox{G}=\mbox{U}(n)$)}
\medskip \\
Let $J_{0}$ be unitary in \eqref{Linear systems of ASDYM}. Suppose that $(\theta, \Lambda)$ and $(\overline{\theta}, \Xi)$ satisfy the first and the second linear system of \eqref{Linear systems of ASDYM} respectively. If there exists $\Omega(\theta, \theta) \in \mathbb{C}_{N \times N}[z, \widetilde{z}, w, \widetilde{w}]$  such that 
\begin{eqnarray}
	\label{Sylvester eq_Unitary}
	\Xi^{T}\Omega(\theta, \theta) - \Omega(\theta, \theta)\Lambda = \theta^{\dagger}\theta.
\end{eqnarray}
 Then 
\begin{eqnarray}
\label{J_unitary}
J
:=
\left|
\begin{array}{cc}
	\Xi^{T}\Omega(\theta, \theta) & \theta^{\dagger} 
	\smallskip \\
	\theta  & \fbox{$I$}
\end{array}
\right|J_{0}
\end{eqnarray}
is a unitary solution of the Yang equation \eqref{Yang eq_CM} if and only if $\Omega(\theta, \theta)$  satisfies the relation
\begin{eqnarray}
\label{Unitary condition_CM}
\Omega^{\dagger}(\theta, \theta)
=
-\Omega(\theta, \theta)\Lambda\overline{\Xi}^{-1}.		
\end{eqnarray}
\end{Prop_4.2}
\textit{Proof:}
First of all, the condition $\overline{\rho}=\theta$ implies \eqref{Sylvester eq} is Hermitian, that is,
\begin{eqnarray}
\label{Sylvester eq_Hermitian}
\Xi^{T}\Omega - \Omega\Lambda = \theta^{\dagger}\theta 
=(\theta^{\dagger}\theta)^{\dagger}
=-\Lambda^{\dagger}\Omega^{\dagger} + \Omega^{\dagger}\overline{\Xi}.
\end{eqnarray}
Therefore, \eqref{Unitary condition_CM} makes sense because such $\Omega$ satisfies \eqref{Sylvester eq_Hermitian}. By direct calculation, we obtain
\begin{eqnarray}
JJ^{\dagger}
&\!\!\!\!=\!\!\!\!&
I-\theta\Omega^{-1}\Xi^{-T}(I+\Omega\Lambda\overline{\Xi}^{-1}\Omega^{-\dagger})\theta^{\dagger},
\\
J^{\dagger}J
&\!\!\!\!=\!\!\!\!&
I-\theta(I+\overline{\Xi}^{-1}\Omega^{-\dagger}\Omega\Lambda)\Omega^{-1}\Xi^{-T}\theta^{\dagger}.
\end{eqnarray}
Now we find that $J^{\dagger}J=JJ^{\dagger}=I$ if and only if 
\eqref{Unitary condition_CM} holds.
$\hfill\Box$ \\

Note that if we choose $\Lambda, \Xi$ to be diagonal in Proposition 4.2, say
\begin{eqnarray}
	\label{Lambda_diagnoal}
	\Lambda=\mbox{diag}(\lambda_{1}, \cdots, \lambda_{N}), ~ \Xi=\mbox{diag}(\mu_{1}, \cdots, \mu_{N}),
\end{eqnarray}
then $\Omega(\theta, \theta)$ can be expressed simply as a Grammian-like matrix whose $(j, k)$-component is
\begin{eqnarray}
	\label{Omega}
	\Omega(\theta_{\cdot ~\!\!k}, \theta_{\cdot ~\!\!j})
	=\frac{(\theta^{\dagger})_{j~\!\!\cdot }\theta_{\cdot ~\!\! k}}{\mu_{j} - \lambda_{k}}.
\end{eqnarray}
Here we introduce 
$\theta_{j ~\!\! \cdot}$ and $\theta_{\cdot ~\!\! k}$ to denote the $j$-th row and $k$-th column of $\theta$ respectively.
By \eqref{Unitary condition_CM} and \eqref{Omega}, one can verify the following result directly.

\newtheorem{Cor_4.3}[Thm_4.1]{Corollary}
\begin{Cor_4.3} 
	{\bf (A Special Class of Unitary Solutions on $\mathbb{CM}$, $\mbox{G}=\mbox{U}(n)$)}
	\medskip \\	
	If, in Proposition 4.2, we take $\Lambda, \Xi$ to be given by \eqref{Lambda_diagnoal}, then
	$J$ is unitary if and only if 
	\begin{eqnarray}
		\label{Unitary condition_CM_diagonal}
		\frac{\overline{\mu}_{k}}{\overline{\lambda}_{j}-\overline{\mu}_{k}}
		=
		\frac{\lambda_{k}}{\mu_{j}-\lambda_{k}},~1 \leq j, k \leq N.
	\end{eqnarray} 
\end{Cor_4.3}


\section{Quasi-Grammian $N$-Soliton for $\mbox{G}=\mbox{GL}(2, \mathbb{C})$}
In this section, we first present a quasi-Grammian ansatz for the $N$-soliton solutions of the ASDYM equation with $\mbox{G}=\mbox{GL}(2, \mathbb{C})$ and discuss the conditions for unitarity on each real space. Next, we verify that the ansatz satisfies the defining criteria for an $N$-soliton solution. Finally, we establish its asymptotic equivalence to the quasi-Wronskian representation of the ASDYM $N$-soliton solutions in the WZW$_4$ model.

\subsection{$N$-Soliton Ansatz and Unitarity}
By the definition of quasideterminants (cf.~\eqref{Quasideterminant_defn_concrete form}, \eqref{Quasideterminant_defn_expansion}) for $\mbox{G}=\mbox{GL}(2, \mathbb{C})$, the $J$-matrix defined in \eqref{J_CM} can be written explicitly as a $2 \times 2$ matrix with commutative entries. For a given $N$, the four entries can be expanded in terms of $(N+1)$-th order quasideterminants according to the following rule:
\begin{eqnarray}
	\label{J_CM_components}
	J_{_{[N]}}=
	\left|
	\begin{array}{cc}
		\Xi^{T}\Omega &  \rho^{T}
		\smallskip \\
		\theta & \fbox{$I$}
	\end{array}
	\right|J_{0} 
	=
	\left(
	\begin{array}{cc}
		\!\left|
		\begin{array}{cc}
			\Xi^{T}\Omega &  (\rho^{T})_{\cdot~\!\! 1}
			\smallskip \\
			\theta_{1~\!\! \cdot} & \fbox{1}
		\end{array}
		\right|
		&
		\!\left|
		\begin{array}{cc}
			\Xi^{T}\Omega & (\rho^{T})_{\cdot~\!\! 2}
			\smallskip \\
			\theta_{1~\!\! \cdot} & \fbox{0}
		\end{array}
		\right|
		\medskip \\
		\!\left|
		\begin{array}{cc}
			\Xi^{T}\Omega & (\rho^{T})_{\cdot~\!\! 1}
			\smallskip \\
			\theta_{2~\!\! \cdot} & \fbox{0}
		\end{array}
		\right| 
		&
		\!\left|
		\begin{array}{cc}
			\Xi^{T}\Omega & (\rho^{T})_{\cdot~\!\! 2}
			\smallskip \\
			\theta_{2~\!\! \cdot} & \fbox{1}
		\end{array}
		\right|
	\end{array}
	\!\!\right)J_{0} ,
\end{eqnarray}
where $\Omega, \Xi \in \mathbb{C}_{N \times N}$, $\theta, \rho \in \mathbb{C}_{2 \times N}$, and $I=I_2$. For convenience, we introduce the notations $\theta_{j ~\!\!\cdot}$ and $(\rho^{T})_{\cdot~\!\! k}$ to denote the $j$-th row of $\theta$ and the $k$-th column of $\rho^{T}$, respectively, for $j,k = 1,2$.
Note that the superscript $[N]$ of the $J$-matrix denotes the soliton number rather than the matrix size, the latter being $2 \times 2$.
An ansatz for the quasi-Grammian $N$-soliton (cf. \cite{LHHZ-2025, LQYZ-2022, LQZ-2023, Ohta-2024}) can be obtained by taking the seed solution $J_{0}=I$ together with the following input data in \eqref{J_CM_components}:
\begin{eqnarray}
\label{Soliton ansatz_theta}	
\theta
&\!\!\!\!=\!\!\!\!&
\left(
\begin{array}{c}
	\theta_{1~\!\!\cdot}
	\\
	\theta_{2~\!\!\cdot}
\end{array}
\right)
:=
\left(
\begin{array}{cccc}
E_{1}^{\scriptscriptstyle(+)} & E_{2}^{\scriptscriptstyle(+)} & \cdots & E_{_{N}}^{\scriptscriptstyle(+)}
\smallskip \\
E_{1}^{\scriptscriptstyle(-)} & E_{2}^{\scriptscriptstyle(-)} & \cdots & E_{_{N}}^{\scriptscriptstyle(-)}
\end{array}
\right)
, ~
\Lambda
=
\mbox{diag}\left(
\lambda_{1}, \lambda_{2}, \cdots, \lambda_{_N}
\right), ~
\smallskip \\
\label{Soliton ansatz_rho}
\rho
&\!\!\!\!=\!\!\!\!&
\left(
\begin{array}{c}
	\rho_{1~\!\!\cdot}
	\\
	\rho_{2~\!\!\cdot}
\end{array}
\right)
:=
\left(
\begin{array}{cccc}
	\widetilde{E}_{1}^{\scriptscriptstyle(+)} & \widetilde{E}_{2}^{\scriptscriptstyle(+)} & \cdots & \widetilde{E}_{_{N}}^{\scriptscriptstyle(+)}
	\smallskip \\
	\widetilde{E}_{1}^{\scriptscriptstyle(-)} & \widetilde{E}_{2}^{\scriptscriptstyle(-)} & \cdots & \widetilde{E}_{_{N}}^{\scriptscriptstyle(-)}
\end{array}
\right),~
\Xi
	=
	\mbox{diag}\left(
	\mu_{1}, \mu_{2}, \cdots, \mu_{_N}
	\right), 
\end{eqnarray}
 where the entries of $\theta$ and $\rho$ are defined by
\begin{eqnarray}
\label{E_j-pm}
E_{j}^{(\pm)}&\!\!\!\!:=\!\!\!\!&
(a_{j}^{(\pm)})^2e^{\pm L_{j}},~
L_{j}:=\lambda_{j}\alpha_{j}z + \beta_{j}\widetilde{z} + \lambda_{j}\beta_{j}w + \alpha_{j}\widetilde{w}
:=\ell_{jm}z^{m},~ 
\\
\label{E-tilde_j-pm}
\widetilde{E}_{j}^{(\pm)}&\!\!\!\!:=\!\!\!\!&
(\widetilde{a}_{j}^{(\pm)})^2e^{\pm \widetilde{L}_{j}},~
\widetilde{L}_{j}:=\mu_{j}\widetilde{\alpha}_{j}z + \widetilde{\beta}_{j}\widetilde{z} + \mu_{j}\widetilde{\beta}_{j}w + \widetilde{\alpha}_{j}\widetilde{w}
:=\widetilde\ell_{jm}z^{m},~
\end{eqnarray}
$a_{j}^{(\pm)}, \widetilde{a}_{j}^{(\pm)},\alpha_{j}, \beta_{j}, \widetilde{\alpha}_{j}, \widetilde{\beta}_{j}, \lambda_{j}, \mu_{j} \in \mathbb{C}, 1\leq j \leq N, (z^{1}, z^{2}, z^{3}, z^{4}):=(z, \widetilde{z}, w, \widetilde{w})$.

To discuss the unitarity of the $N$-soliton ansatz, we make use of Proposition 4.2 and impose the following additional symmetry:
\begin{eqnarray}
	\label{Reality condition}
	\widetilde{E}_{j}^{(\pm)}=\overline{E}_{j}^{(\pm)},~ 1 \leq j \leq N
\end{eqnarray}
which implies $\widetilde{a}_{j}^{(\pm)} = \overline{a}_{j}^{(\pm)}$ and
\begin{eqnarray}
	\label{Reality conditions}
	\left\{
	\begin{array}{l}
		(1) ~ \mu_{j}=\overline{\lambda}_{j}, ~ \widetilde{\alpha}_{j}=\overline{\alpha}_{j}, ~
		\widetilde{\beta}_{j}=\overline{\beta}_{j}
		~~~~~~~~~~~~~~~~~~~\mbox{on} ~\mathbb{U}^{2,2}_{1} ~\!(z, w, \widetilde{z}, \widetilde{w} \in \mathbb{R})
		\smallskip \\
		(2)~ \mu_{j}=\overline{\lambda}_{j}, ~ \widetilde{\alpha}_{j}=\overline{\alpha}_{j}, ~
		\widetilde{\beta}_{j}=\overline{\beta}_{j}, ~ \alpha_{j}=\lambda_{j}\beta_{j} ~~~~\! \mbox{on} ~ \mathbb{M}^{3,1} ~\!(z, \widetilde{z}, \in \mathbb{R}, \widetilde{w}=\overline{w})
		\smallskip \\
		(3)~  \mu_{j}=-(1/\overline{\lambda}_{j}),~
		\widetilde{\alpha}_{j}=-\overline{\lambda}_{j}\overline{\beta}_{j}, ~\widetilde{\beta}_{j}=\overline{\lambda}_{j}\overline{\alpha}_{j}
		 ~~~\mbox{on} ~ \mathbb{E}^{4} ~\!(\widetilde{z}=\overline{z}, \widetilde{w}=-\overline{w})
		\smallskip \\
		(4)~ \mu_{j}=1/\overline{\lambda}_{j}, ~
		\widetilde{\alpha}_{j}=\overline{\lambda}_{j}\overline{\beta}_{j}, ~\widetilde{\beta}_{j}=\overline{\lambda}_{j}\overline{\alpha}_{j}
		 ~~~~~~~~~~~\! \mbox{on} ~\mathbb{U}^{2,2}_{2} ~\!(\widetilde{z}=\overline{z}, \widetilde{w} = \overline{w})
	\end{array}
	\right.
	, 1 \leq j \leq N.  
\end{eqnarray}
Note that (1) and (2) satisfy the requirement of \eqref{Unitary condition_CM_diagonal} for all $1 \leq j, k \leq N$. Therefore, unitary solutions can be realized on the real spaces $\mathbb{U}^{2,2}_1$ and $\mathbb{M}^{3,1}$.
On the other hand, substituting (3) and (4) into \eqref{Unitary condition_CM_diagonal} yields $\lambda_{k}\overline{\lambda}_{j}=1$ and $\lambda_{k}\overline{\lambda}_{j}=-1$, respectively. Since these relations must hold for all $1 \leq j, k \leq N$, it follows that $|\lambda_{1}|=|\lambda_{j}|=1$ for all $j$ on $\mathbb{E}^{4}$, and $|\lambda_{1}|=|\lambda_{j}|=-1$ for all $j$ on $\mathbb{U}^{2,2}_{2}$.
The latter is not admissible, while the former suggests that unitary $N$-soliton solutions are not available for $N \geq 2$ on $\mathbb{E}$.
In fact, we will show later that, even for nonunitary solitons, the condition \eqref{Reality condition} ensures that the action densities of the WZW$_4$ model are real-valued on each real space. For physical purposes, we will therefore impose the condition \eqref{Reality condition} on the $N$-soliton ansatz throughout what follows. Accordingly, the choice of the spectral parameter $\Xi$ is determined by the reality condition \eqref{Reality conditions} associated with each real space.

From \eqref{Omega}, the $(j, k)$-entry of $\Xi^{T}\Omega$ is given explicitly by
\begin{eqnarray}
	\label{gamma_jk}
	(\Xi^{T}\Omega)_{jk}
	=
	\gamma_{jk}(\overline{E}_{j}^{\scriptscriptstyle(+)}\!E_{k}^{\scriptscriptstyle(+)} \!+ \overline{E}_{j}^{\scriptscriptstyle(-)}\!E_{k}^{\scriptscriptstyle(-)}
	),~ \gamma_{jk} = \frac{\mu_{j}}{\mu_{j} - \lambda_{k}}  ,~ 1 \leq j, k \leq N. ~~
\end{eqnarray}
For convenience in the subsequent discussion, we further define
\begin{eqnarray}
\label{gamma-tilde_jk}	
\widetilde{\gamma}_{jk}:=  
\frac{\lambda_{j}}{\mu_{k} - \lambda_{j}}.
\end{eqnarray}
For $N=1$, the $1$-soliton ansatz can be readily obtained as
\begin{eqnarray}
\label{J_1-soliton}	
J_{[1]}
=
\frac{1}{\gamma_{_{11}}(|E_{1}^{\scriptscriptstyle(+)}|^2 +|E_{1}^{\scriptscriptstyle(-)}|^2)}
\left(
\begin{array}{cc}
\widetilde{\gamma}_{_{11}}|E_{1}^{\scriptscriptstyle(+)}|^2 + \gamma_{_{11}}|E_{1}^{\scriptscriptstyle(-)}|^2
	&
	E_{1}^{\scriptscriptstyle(+)}\overline{E}_{1}^{\scriptscriptstyle(-)}
	\smallskip \\
	\overline{E}_{1}^{\scriptscriptstyle(+)}E_{1}^{\scriptscriptstyle(-)}
	& 
   \gamma_{_{11}}|E_{1}^{\scriptscriptstyle(+)}|^2 + \widetilde{\gamma}_{_{11}}|E_{1}^{\scriptscriptstyle(-)}|^2
\end{array}
\right),
\end{eqnarray}
where $\gamma_{11}, \widetilde{\gamma}_{11}$ are determined by the reality condition \eqref{Reality conditions} for each real space.
By using \eqref{NLSM action} and \eqref{Tr(A_m A_n)}, 
the NL$\sigma$M action density associated with $J_{[1]}$ is (cf. \cite{HHK-2023})
\begin{eqnarray}
\label{NLSM_1-soliton}	
	{\cal{L}}_{\sigma}[J_{[1]}]
	=
	\frac{-1}{16\pi}\mbox{Tr}\left[
	(\partial_{m}J_{[1]})J_{[1]}^{-1}(\partial^{m}J_{[1]})J_{[1]}^{-1}
	\right]
	=\frac{1}{8\pi}d_{11}\mbox{sech}^2(X_{1} + \log\left|\frac{a_{1}^{\scriptscriptstyle(+)}}{b_{1}^{\scriptscriptstyle(-)}}\right|),
\end{eqnarray}
where $\partial^m := g^{mn}\partial_n$ with the metric defined in \eqref{com_metric}, $X_{1}:= L_{1} + \overline{L}_{1}$, and
\begin{eqnarray}
	\label{d_11}
	d_{11}
	:=
\left\{
	\begin{array}{l}	
	(1) ~
	(\alpha_{1}\overline{\beta}_{1} - \overline{\alpha}_{1}\beta_{1})(\lambda_{1} - \overline{\lambda}_{1})^{3}/|\lambda_{1}|^2  ~~~~~~\mbox{on}~~\mathbb{U}^{2,2}_{1}
	\smallskip \\
	(2) ~|\beta_{11}|^2(\lambda_{1}-\overline{\lambda}_{1})^4/|\lambda_{1}|^2 ~~~~~~~~~~~~~~~~~~\!\mbox{on}~~\mathbb{M}^{3,1}
	\smallskip \\
	(3) ~
	-(|\alpha_{1}|^2 + |\beta_{1}|^2)(|\lambda_{1}|^2 + 1)^3/|\lambda_{1}|^2 ~~\!\mbox{on}~~\mathbb{E}^{4}
	\smallskip \\
	(4) ~
	(|\alpha_{1}|^2 - |\beta_{1}|^2)(|\lambda_{1}|^2 - 1)^3/|\lambda_{1}|^2 ~~~~~~\!\mbox{on}~~\mathbb{U}^{2,2}_{2}
	
\end{array}.
\right.
\end{eqnarray}
It is straightforward to verify that ${\cal{L}}_{\sigma}[J_{[1]}]$ is real-valued on each real space.
By using \eqref{Tr(A_m A_n A_p)}, one can verified that (cf. \cite{HHK-2023})
\begin{eqnarray}
	\mbox{Tr}\left[
	(\partial_{m}J_{[1]})J_{[1]}^{-1}(\partial_{n}J_{[1]})J_{[1]}^{-1}(\partial_{p}J_{[1]})J_{[1]}^{-1}
	\right]=0,
\end{eqnarray}
which implies that the Wess--Zumino action density  ${\cal{L}}_{_{WZ}}[J_{[1]}]$ vanishes. Therefore, the total contribution of the WZW$_{4}$ action associated with $J_{[1]}$ arises solely from the NL$\sigma$M action density, which exhibits behavior analogous to the standard 1-soliton in lower-dimensional integrable systems, such as KdV/KP solitons \cite{Kodama-2017-18, Hirota-2004}, and provides a visualization of a codimension-one soliton for the ASDYM equation. 
Thus, for $N=1$, the $N$-soliton ansatz $J_{[N]}$ correctly describes the characteristic behavior of the standard 1-soliton. 
Next, we proceed to establish the validity of the $N$-soliton ansatz for arbitrary natural numbers $N$ via an asymptotic analysis.

\subsection{Asymptotic Analysis for the $N$-Soliton Ansatz}

Motivated by the particle-like nature of solitons, namely that they preserve their identities under multi-soliton collisions up to phase shifts, it suffices to verify the validity of the $N$-soliton ansatz by establishing that, asymptotically, it reduces to a 1-soliton-like object at the matrix level. To analyze the asymptotics of the $N$-soliton ansatz $J_{_{[N]}}$, without loss of generality, we consider a comoving frame associated with the $K$-th 1-soliton, where $K$ is a fixed positive integer belonging to $\{1,2,\ldots,N\}$, and focus the evolution of $J_{_{[N]}}$ on the following asymptotic region $\mathscr{R}^{\scriptscriptstyle (K)}(M)$ associated with a large real number $M$.
More specifically,
for a given $M>0$ and fixed $l_{_{K}}$, we choose an asymptotic region
\begin{eqnarray}
\label{asymptotic region_epsilon}	
	\mathscr{R}^{\scriptscriptstyle (K)}(M):=
	\left\{
	(x_{1},x_{2},x_{3},x_{4})\in\mathbb{M}_4
	\;\Bigg|\;
\begin{array}{l}	
	L_{_{K}}=l_{_{K}},
	\ \
	r:=\sqrt{x_{1}^{2}+x_{2}^{2}+x_{3}^{2}+x_{4}^{2}}>R
\\
\text{for some sufficiently large $R>0$ such that}
\\
|\operatorname{Re}(L_j)|>M,  j\neq K.
\end{array}	
	\right\}
\end{eqnarray}
In fact, $\mathscr{R}^{\scriptscriptstyle (K)}(M)$ can be decomposed into $2^{N-1}$ asymptotic sub-regions according to the signs of $\operatorname{Re}(L_j)$, $j\neq K$. More precisely, for each choice of symbols
\begin{eqnarray}
\label{epsilon symbols}	
\varepsilon_j=\pm, ~j\neq K,
\end{eqnarray}
we define
\begin{eqnarray}
	&&\!\!\!\!\mathscr{R}^{\scriptscriptstyle (K)}_{(\varepsilon_{1},\cdots, \varepsilon_{_{K-1}}, ~\bullet~,
		\varepsilon_{_{K+1}}, \cdots, \varepsilon_{_{N}})}(M)
	\nonumber\\
	&\!\!\!\!:=\!\!\!\!&
	\left\{
	(x_{1},x_{2},x_{3},x_{4})
	\in
	\mathscr{R}^{\scriptscriptstyle (K)}(M)
	\;\Bigg|\;
	\operatorname{Re}(L_j)
	\begin{cases}
		>M, &  \mbox{if} ~~\varepsilon_j=+,\\
		<-M, & \mbox{if} ~~\varepsilon_j=-,
	\end{cases}
	\quad j\neq K
	~~\right\}.
\end{eqnarray}
Therefore, we can decompose $\theta$ into the following $2^{N-1}$ cases: 
\begin{eqnarray}
\label{theta decomposition}
\theta=\theta^{\scriptscriptstyle(K)}{D}^{\scriptscriptstyle(K)},
\end{eqnarray}
where the $j$-th column of $\theta^{\scriptscriptstyle(K)}$ are defined by
\begin{eqnarray}
(\theta^{\scriptscriptstyle(K)})_{\cdot~\!\! j}
= \left\{
\begin{array}{l}
\left(
\begin{array}{c}
1 \\ E_{j}^{\scriptscriptstyle(-)}/E_{j}^{\scriptscriptstyle(+)} 
\end{array}
\right)
~~\text{if} ~~ \mbox{Re}(L_j)_{j \neq K} > M
\smallskip \\
\left(
\begin{array}{c}
 E_{j}^{\scriptscriptstyle(+)}/E_{j}^{\scriptscriptstyle(-)} \\ 1 
\end{array}
\right)~~\text{if} ~~ \mbox{Re}(L_j)_{j \neq K} < -M
\smallskip \\
\left(
\begin{array}{c}
	E_{_{K}}^{\scriptscriptstyle(+)} \\ E_{_{K}}^{\scriptscriptstyle(-)} 
\end{array} 
\right), ~j=K
\end{array}
\right.,
\end{eqnarray}
$D^{\scriptscriptstyle(K)}$ is a diagonal matrix associated with a given $\theta^{\scriptscriptstyle(K)}$.
Substituting \eqref{theta decomposition} into \eqref{Sylvester eq_Unitary} and  \eqref{J_unitary}, one obtains
\begin{eqnarray}
\label{Sylvester eq-pm}
\Xi^{T}\Omega^{\scriptscriptstyle(K)} - \Omega^{\scriptscriptstyle(K)}\Lambda = \theta^{{\scriptscriptstyle(K)}\dagger}\theta^{\scriptscriptstyle(K)}, ~~
\Omega^{\scriptscriptstyle(K)}:=(D^{\scriptscriptstyle(K)})^{-\dagger}\Omega (D^{\scriptscriptstyle(K)})^{-1},
\end{eqnarray}
and
\begin{eqnarray}
	\label{J=J_pm(k)}
	J_{_{[N]}}
	=
	\left|
	\begin{array}{cc}
		\!\!\Xi^{T}D^{{\scriptscriptstyle(K)}\dagger}\Omega^{\scriptscriptstyle(K)}D^{\scriptscriptstyle(K)} & \!\!D^{{\scriptscriptstyle(K)}\dagger} \theta^{{\scriptscriptstyle(K)}\dagger}
		\smallskip\\
		\!\!\theta^{\scriptscriptstyle(K)}D^{\scriptscriptstyle(K)} & \!\!\fbox{$I$}
	\end{array}
	\!\!\right|
	=
	\left|
	\begin{array}{cc}
		\!\!\Xi^{T}\Omega^{\scriptscriptstyle(K)} & \!\!\theta^{{\scriptscriptstyle(K)}\dagger}
		\smallskip\\
		\!\!\theta^{\scriptscriptstyle(K)} & \!\!\fbox{$I$}
	\end{array}
	\!\!\right|.
\end{eqnarray}
Here, we use the fact that $[\Xi^{T}, D^{{\scriptscriptstyle(K)}\dagger}]=0$ together with the multiplication rule for quasideterminants (Proposition 3.2) to factor out and remove the common row factor $D^{{\scriptscriptstyle(K)}\dagger}$ and the common column factor $D_{_{K}}$.
It follows that $J_{_{[N]}}$ admits an equivalent expression in terms of $\Omega^{\scriptscriptstyle(K)}$ and $\theta^{\scriptscriptstyle(K)}$.

Next, we analyze the asymptotic behavior of $J_{_{[N]}}$. 
By the permutation rule for quasideterminants (Proposition 3.1), without loss of generality, we can take the sub-region characterized by
\begin{eqnarray}
\varepsilon_{j}
=
\begin{cases}
	+, &  1 < j < K\\
	-, &  K < j < N 
\end{cases}
\end{eqnarray}
as a representative of the asymptotic regions,
that is, 
\begin{eqnarray}
\label{Asymptotic region_standard}
\mathscr{R}^{\scriptscriptstyle (K)}_{(+, \cdots, +, ~\!\bullet~\!,
		-, \cdots, -)}(M)
=
	\left\{
	(x_{1},x_{2},x_{3},x_{4})
	\in
	\mathscr{R}^{\scriptscriptstyle (K)}(M)
	\;\Bigg|\;
	\operatorname{Re}(L_j)
	\begin{cases}
		>M, & 1 < j < K\\
		<-M, & K < j < N
	\end{cases}
	~\right\}.~
\end{eqnarray}
In this case, we have
\begin{eqnarray}
	\label{theta_K-pm}
	\begin{array}{l}
		\theta^{\scriptscriptstyle(K)}
		:=\left( 
		\begin{array}{ccccccc}
			\!\!1 & \!\!\!\cdots  &  \!\!\!1  & \!\!\!E_{_{K}}^{\scriptscriptstyle (+)} & \!\!\!E_{_{K+1}}^{\scriptscriptstyle(+)}/E_{_{K+1}}^{\scriptscriptstyle(-)} & \!\!\!\cdots & \!\!\!E_{_{N}}^{\scriptscriptstyle(+)}/E_{_{N}}^{\scriptscriptstyle(-)}
			\smallskip \\
			\!\! E_{1}^{\scriptscriptstyle(-)}/E_{1}^{\scriptscriptstyle(+)} & \!\!\!\cdots & \!\!\! E_{_{K-1}}^{\scriptscriptstyle(-)}/E_{_{K-1}}^{\scriptscriptstyle(+)} & E_{_{K}}^{(-)} & \!\!\!1 & \!\!\!\cdots & \!\!\!1
		\end{array}
		\!\!\right)
		\medskip \\
		D^{\scriptscriptstyle(K)}:=\mbox{diag}\left(
		E_{1}^{\scriptscriptstyle(+)}, \cdots, E_{_{K-1}}^{\scriptscriptstyle(+)}, ~1~, E_{_{K+1}}^{\scriptscriptstyle(-)}, \cdots, E_{_{N}}^{\scriptscriptstyle(-)}
		\right)
	\end{array}.
\end{eqnarray}
Before our discussion, we introduce the variant characters 
$ \vartheta^{\scriptscriptstyle(K)}$ and $\mathcal{J}_{_{[N]}}^{\scriptscriptstyle(K)}$ to denote the asymptotic forms of $\theta^{\scriptscriptstyle(K)}$ and $J_{_{[N]}}$, respectively.
Note that for $j \neq K$, we have
\begin{eqnarray}	
\left|(E_{j}^{(\pm)}/E_{j}^{(\mp)}) \right|
=\left|(a_{j}^{(\pm)}/a_{j}^{(\mp)})^2\right|e^{\pm 2 \mbox{Re}(L_{j})}
\longrightarrow  0 
~~\mbox{as}~~\mbox{Re}(L_{j}) \longrightarrow \mp \infty. 
\end{eqnarray}
Comparing with \eqref{theta_K-pm}, we have
\begin{eqnarray}
\label{theta_asymptotic}	
\theta^{\scriptscriptstyle(K)} \xrightarrow[]{r \rightarrow \infty}
\vartheta^{\scriptscriptstyle(K)}
:=\left(
\begin{array}{ccccccc}
\!\!1 & \!\!\cdots & \!\!1 & \!\!E_{_{K}}^{\scriptscriptstyle(+)} & \!\!0 & \!\!\cdots & \!\!0
\smallskip \\
\!\!0 & \!\!\cdots & \!\!0 & \!\!E_{_{K}}^{\scriptscriptstyle(-)} & \!\!1 & \!\!\cdots & \!\!1
\end{array}
\!\!\right) .
\end{eqnarray}
By using the first equation of \eqref{Sylvester eq-pm}, \eqref{theta_asymptotic} and \eqref{J=J_pm(k)}, one yields
\begin{eqnarray}
\label{J_asymptotic}
&&\!\!\!\! J_{_{[N]}} 
\xrightarrow[]{r \rightarrow \infty}
\mathcal{J}_{_{[N]}}^{\scriptscriptstyle(K)} \nonumber \\
&\!\!\!\!:=\!\!\!\!&
\!\!\left|
\begin{array}{c:c}
\begin{array}{ccccccc}
\!\!\!\!\gamma_{_{11}} & \!\!\!\!\!\!\cdots & \!\!\!\!\!\!\gamma_{_{1,K-1}} & \!\!\!\!\!\!\gamma_{_{1,K}} E_{_{K}}^{\scriptscriptstyle(+)} & \!\!\!\!\!\!0 & \!\!\!\!\!\!\cdots & \!\!\!\!\!\!0 
\\
\!\!\!\!\vdots & \!\!\!\!\!\!\ddots & \!\!\!\!\!\!\vdots & \!\!\!\!\!\!\vdots & \!\!\!\!\!\!\vdots & \!\!\!\!\!\!\ddots & \!\!\!\!\!\! \vdots 
\\
\!\!\!\!\gamma_{_{K-1,1}} & \!\!\!\!\!\!\cdots & \!\!\!\!\!\!\gamma_{_{K-1,K-1}} & \!\!\!\!\!\!\gamma_{_{K-1,K}}E_{_{K}}^{\scriptscriptstyle(+)} & \!\!\!\!\!\!0 & \!\!\!\!\!\!\cdots & \!\!\!\!\!\!0  
\smallskip \\
\!\!\!\!\gamma_{_{K,1}}\overline{E}_{_{K}}^{\scriptscriptstyle(+)} & \!\!\!\!\!\!\cdots & \!\!\!\!\!\!\gamma_{_{K,K-1}}\overline{E}_{_{K}}^{\scriptscriptstyle(+)} & \!\!\!\!\gamma_{_{K, K}}\!\left( |E_{_{K}}^{\scriptscriptstyle(+)}|^2 + |E_{_{K}}^{\scriptscriptstyle(-)}|^2 \right)
& \!\!\!\!\gamma_{_{K,K+1}}\overline{E}_{_{K}}^{\scriptscriptstyle(-)} & \!\!\!\!\cdots & \!\!\!\!\gamma_{_{K,N}}\overline{E}_{_{K}}^{\scriptscriptstyle(-)}
\smallskip \\
\!\!\!\!0 & \!\!\!\!\!\!\cdots & \!\!\!\!\!\!0 & \!\!\!\!\!\!\gamma_{_{K+1,K}}E_{_{K}}^{\scriptscriptstyle(-)} & \!\!\!\!\!\!\gamma_{_{K+1,K+1}} & \!\!\!\!\!\!\!\!\cdots & \!\!\!\!\!\!\gamma_{_{K+1,N}} 
\\
\!\!\!\!\vdots & \!\!\!\!\!\!\ddots & \!\!\!\!\!\!\vdots & \!\!\!\!\!\!\vdots & \!\!\!\!\!\!\vdots & \!\!\!\!\!\!\ddots & \!\!\!\!\!\!\vdots 
\\
\!\!\!\!0 & \!\!\!\!\!\!\cdots & \!\!\!\!\!\!0 & \!\!\!\!\!\!\gamma_{_{N,K}}E_{_{K}}^{\scriptscriptstyle(-)} & \!\!\!\!\!\!\gamma_{_{N,K+1}} & \!\!\!\!\!\!\cdots & \!\!\!\!\!\!\gamma_{_{N,N}}  
\end{array}	 
& 
\begin{array}{cc}
\!\!\!1 & \!\!\!\!\!0 
\\
\!\!\!\vdots & \!\!\!\!\vdots
\\
\!\!\!1 & \!\!\!\!0
\smallskip \\
\!\!\!\overline{E}_{_{K}}^{\scriptscriptstyle(+)}  &  \!\!\!\!\overline{E}_{_{K}}^{\scriptscriptstyle(-)}
\smallskip \\
\!\!\!0 & \!\!\!\!1
\\
\!\!\!\vdots & \!\!\!\!\!\vdots
\\
\!\!\!0 & \!\!\!\!1 
\end{array}
\smallskip \\
\hdashline 
\begin{array}{ccccccc}
\!\!\!\!\!1 ~~~ & \!\!\!\cdots ~~~~ & \!1 ~~~~~~~~~~~~~~~~~~~ & E_{_{K}}^{\scriptscriptstyle(+)} ~~~~~~~~~~~~~~~ & 
\!\!\!0 ~~~  & \cdots ~~~~ & 0  
\\
\!\!\!\!\!0 ~~~ & \!\!\!\cdots~~~~ & \!0 ~~~~~~~~~~~~~~~~~~~ & E_{_{K}}^{\scriptscriptstyle(-)} ~~~~~~~~~~~~~~~ & \!\!\!1~~~  & \cdots  ~~~~& 1  
\end{array} 
& 
\!\!\fbox{
		$\begin{array}{cc}
		\!\!1 & ~~0 \\
		\!\!0 & ~~1
		\end{array}$
	}
\end{array}
	\!\!\!\right|, ~~~~~~~  
\end{eqnarray}
where $\gamma_{jk}$ are defined in \eqref{gamma_jk}.
Note that the four components of $2 \times 2$ matrix $\mathcal{J}_{_{[N]}}^{\scriptscriptstyle(K)}$ can be computed exactly by using \eqref{Quasideterminant_defn_expansion} and the Noncomutative Jacobi identity (Proposition 3.4).
To deal with the quasideterminants involving block matrix with zero blocks, we introduce the following simple result as Lemma 5.1, which can be verified directly from the definition of quasideterminants.
The proof is given in Appendix B.

\newtheorem{Lem_5.1}{Lemma}[section]\label{lemma-5-1}
\begin{Lem_5.1}
Let	$X$ be an $n \times n$ matrix consisting of four block matrices $A$, $O_{1}$, $O_{2}$ and $D$ as
\begin{eqnarray}
X=(x_{ij})_{n \times n}
=
\left(
\begin{array}{cc}
A & O_{1} \\
O_{2} & D 
\end{array}
\right),
\end{eqnarray}
and $|X|_{ij}$ be the $(i,j)$-th quasideterminant of $X$.
If $A$ and $D$ are $m \times m$ and $(n-m) \times (n-m)$ square matrices respectively, and $O_{1}$ and $O_{2}$ are zero matrices with suitable sizes, then 
\begin{eqnarray}
|X|_{ij} =
\left\{
\begin{array}{l}
|A|_{ij}, ~~ 1 \leq i, j \leq m
\\
|D|_{ij}, ~~ n-m \leq i, j \leq n
 ~~\end{array}
\right..
\end{eqnarray}
\end{Lem_5.1}
 To simplified the computation of $\mathcal{J}_{_{[N]}}^{\scriptscriptstyle(K)}$, 
we derive several highly nontrivial quasideterminant identities with  Vandermonde-like structures, presented in Lemmas 5.2--5.4. Their proofs are also given in Appendix B.
\newtheorem{Lem_5.2}[Lem_5.1]{Lemma}
\begin{Lem_5.2} 
	\begin{eqnarray}
		\left|
		\begin{array}{cccccc}
			\displaystyle{\frac{\mu_{1}}{\mu_{1} - \lambda_{1}}} & \cdots & \displaystyle{\frac{\mu_{1}}{\mu_{1} - \lambda_{m}}} & \cdots & \displaystyle{\frac{\mu_{1}}{\mu_{1} - \lambda_{n}}} \\
			\vdots &  \ddots & \vdots  & \ddots & \vdots    \\
			\displaystyle{\frac{\mu_{m}}{\mu_{m} - \lambda_{1}}} & \cdots & \fbox{$\displaystyle{\frac{\mu_{m}}{\mu_{m} - \lambda_{m}}}$} & \cdots & \displaystyle{\frac{\mu_{m}}{\mu_{m} - \lambda_{n}}}  \\ 
			\vdots & \ddots & \vdots & \ddots & \vdots \\
			\displaystyle{\frac{\mu_{n}}{\mu_{n} - \lambda_{1}}} & \cdots & \displaystyle{\frac{\mu_{n}}{\mu_{n} - \lambda_{m}}} & \cdots & \displaystyle{\frac{\mu_{n}}{\mu_{n} - \lambda_{n}}}
		\end{array}
		\right|
		=
		\frac{\mu_{m}}{\mu_{m} - \lambda_{m}}
		\prod_{j=1, j \neq m}^{n}(\frac{\lambda_{m} - \lambda_{j}}{\lambda_{m} - \mu_{j}} )
		(\frac{\mu_{m} - \mu_{j}}{\mu_{m} - \lambda_{j}} ). ~~	
	\end{eqnarray}
\end{Lem_5.2}

\newtheorem{Lem_5.3}[Lem_5.1]{Lemma}
\begin{Lem_5.3} 
	\begin{eqnarray}
		\left|
		\begin{array}{cccccc}
			\displaystyle{\frac{\mu_{1}}{\mu_{1} - \lambda_{1}}} & \cdots & \displaystyle{\frac{\mu_{1}}{\mu_{1} - \lambda_{m}}} & \cdots & \displaystyle{\frac{\mu_{1}}{\mu_{1} - \lambda_{n}}} \\
			\vdots &  \ddots & \vdots  & \ddots & \vdots    
			\\
			\displaystyle{\frac{\mu_{m-1}}{\mu_{m-1} - \lambda_{1}}} & \cdots & \displaystyle{\frac{\mu_{m-1}}{\mu_{m-1} - \lambda_{m}}} & \cdots & \displaystyle{\frac{\mu_{m-1}}{\mu_{m-1} - \lambda_{n}}}  
			\medskip \\ 
			1 & \cdots & \fbox{$1$} & \cdots & 1 
			\medskip \\
			\displaystyle{\frac{\mu_{m+1}}{\mu_{m+1} - \lambda_{1}}} & \cdots & \displaystyle{\frac{\mu_{m+1}}{\mu_{m+1} - \lambda_{m}}} & \cdots & \displaystyle{\frac{\mu_{m+1}}{\mu_{m+1} - \lambda_{n}}}  
			\\  
			\vdots & \ddots & \vdots & \ddots & \vdots 
			\\
			\displaystyle{\frac{\mu_{n}}{\mu_{n} - \lambda_{1}}} & \cdots & \displaystyle{\frac{\mu_{n}}{\mu_{n} - \lambda_{m}}} & \cdots & \displaystyle{\frac{\mu_{n}}{\mu_{n} - \lambda_{n}}}
		\end{array}
		\right|
		=\prod_{j=1, j \neq m}^{n}(\frac{\lambda_{m} - \lambda_{j}}{\lambda_{m} - \mu_{j}} ).
	\end{eqnarray}
	\begin{eqnarray}
		\left|
		\begin{array}{ccccccc}
			\displaystyle{\frac{\mu_{1}}{\mu_{1} - \lambda_{1}}} & \!\!\cdots & \!\!\displaystyle{\frac{\mu_{1}}{\mu_{1} - \lambda_{m-1}}} & 1  & \!\!\displaystyle{\frac{\mu_{1}}{\mu_{1} - \lambda_{m+1}}} & \cdots & \!\!\displaystyle{\frac{\mu_{1}}{\mu_{1} - \lambda_{n}}} \\
			\vdots &  \!\!\ddots & \!\!\vdots &  \vdots  & \vdots & \!\!\ddots & \!\!\vdots    \\
			\displaystyle{\frac{\mu_{m}}{\mu_{m} - \lambda_{1}}} & \cdots & \!\!\displaystyle{\frac{\mu_{m}}{\mu_{m} - \lambda_{m-1}}} & \fbox{$1$} & \!\!\displaystyle{\frac{\mu_{m}}{\mu_{m} - \lambda_{m+1}}} & \!\!\cdots & \!\!\displaystyle{\frac{\mu_{m}}{\mu_{m} - \lambda_{n}}}  \\ 
			\vdots & \!\!\ddots & \!\!\vdots & \vdots & \vdots & \!\!\ddots & \!\!\vdots \\
			\displaystyle{\frac{\mu_{n}}{\mu_{n} - \lambda_{1}}} & \!\!\cdots & \!\!\displaystyle{\frac{\mu_{n}}{\mu_{n} - \lambda_{m-1}}} & 1 & \displaystyle{\frac{\mu_{n}}{\mu_{n} - \lambda_{m+1}}} & \!\!\cdots & \!\!\displaystyle{\frac{\mu_{n}}{\mu_{n} - \lambda_{n}}}
		\end{array}
		\!\!\right|
		=
		\!\!\prod_{j=1, j \neq m}^{n}\!\!(\frac{\lambda_{j}}{\mu_{j}})(\frac{\mu_{m} - \mu_{j}}{\mu_{m} - \lambda_{j}} ).  ~
	\end{eqnarray}		
\end{Lem_5.3}
\newtheorem{Lem_5.4}[Lem_5.1]{Lemma}
\begin{Lem_5.4} 
	\begin{eqnarray}
		\left|
		\begin{array}{ccccccc}
			\displaystyle{\frac{\mu_{1}}{\mu_{1} - \lambda_{1}}} & \!\!\cdots & \!\!\displaystyle{\frac{\mu_{1}}{\mu_{1} - \lambda_{m-1}}} & 1 & \displaystyle{\frac{\mu_{1}}{\mu_{1} - \lambda_{m+1}}} & \!\!\cdots & \!\!\displaystyle{\frac{\mu_{1}}{\mu_{1} - \lambda_{n}}} \\
			\vdots &  \!\!\ddots & \!\!\vdots & \vdots & \vdots  & \!\!\ddots & \!\!\vdots    \\
			\displaystyle{\frac{\mu_{m-1}}{\mu_{m-1} - \lambda_{1}}} & \!\!\cdots & \!\!\displaystyle{\frac{\mu_{m-1}}{\mu_{m-1} - \lambda_{m-1}}} & 1 & \displaystyle{\frac{\mu_{m-1}}{\mu_{m-1} - \lambda_{m+1}}} & \!\!\cdots & \!\!\displaystyle{\frac{\mu_{m-1}}{\mu_{m-1} - \lambda_{n}}}  
			\medskip \\ 
			1 & \!\!\cdots & \!\!1 & \fbox{$1$} & 1 & \!\!\cdots & \!\!1  
			\medskip \\ 
			\displaystyle{\frac{\mu_{m+1}}{\mu_{m+1} - \lambda_{1}}} & \!\!\cdots & \!\!\displaystyle{\frac{\mu_{m+1}}{\mu_{m+1} - \lambda_{m-1}}} & 1 & \displaystyle{\frac{\mu_{m+1}}{\mu_{m+1} - \lambda_{m+1}}} & \!\!\cdots & \!\!\displaystyle{\frac{\mu_{m+1}}{\mu_{m+1} - \lambda_{n}}}  \\ 
			\vdots & \!\!\ddots & \!\!\vdots & \vdots & \vdots & \!\!\ddots & \!\!\vdots \\
			\displaystyle{\frac{\mu_{n}}{\mu_{n} - \lambda_{1}}} & \!\!\cdots & \!\!\displaystyle{\frac{\mu_{n}}{\mu_{n} - \lambda_{m-1}}} & 1 &  \displaystyle{\frac{\mu_{n}}{\mu_{n} - \lambda_{m+1}}}  & \!\!\cdots & \!\!\displaystyle{\frac{\mu_{n}}{\mu_{n} - \lambda_{n}}}
		\end{array}
		\right|
		=\!\!\prod_{j=1, j \neq m}^{n} \!\!(\frac{\lambda_{j}}{\mu_{j}}). ~~
	\end{eqnarray}	
\end{Lem_5.4}
In fact, the above identities play an essential role in analyzing the structure of \eqref{J_asymptotic}. First of all, the Noncommutative Jacobi Identity (Proposition 3.4) implies that  
\begin{eqnarray}
\mathcal{J}_{_{[N]}}^{\scriptscriptstyle(K)}
&\!\!\!\!:=\!\!\!\!&
\left(
\begin{array}{cc}
\left|
A
\right|_{_{N+1, N+1}} 
& 
\left|
B
\right|_{_{N+1, N+1}}
\\
\left|
C
\right|_{_{N+1, N+1}}
&
\left|
D
\right|_{_{N+1, N+1}}
\end{array}
\right)
\\
&\!\!\!\!=\!\!\!\!&
\label{J_asymptotic_Jacobi identity}
\left(
\begin{array}{cc}
\!\left|
\begin{array}{cc}
\!\!\left|
	A_{_{\widehat{N+1},~\!\widehat{N+1}}}
	\right|_{_{K,K}}  
	&
	\!\!\left|
	A_{_{\widehat{K},~\!\widehat{N+1}}}
	\right|_{_{N+1,K}}
    \smallskip \\
	\!\!\left|
	A_{_{\widehat{N+1},~\!\widehat{K}}}
	\right|_{_{K, N+1}}
	&
	\!\!\fbox{
		$\left|
		A_{_{\widehat{K},~\!\widehat{K}}}
		\right|_{_{N+1, N+1}}$
	}
\end{array}
\!\right|
& 
\!\left|
\begin{array}{cc}
	\!\!\left|
	B_{_{\widehat{N+1},~\!\widehat{N+1}}}
	\right|_{_{K,K}}  
	&
	\!\!\left|
	B_{_{\widehat{K},~\!\widehat{N+1}}}
	\right|_{_{N+1,K}}
	\smallskip \\
	\!\!\left|
	B_{_{\widehat{N+1},~\!\widehat{K}}}
	\right|_{_{K, N+1}}
	&
	\!\!\fbox{
		$\left|
		B_{_{\widehat{K},~\!\widehat{K}}}
		\right|_{_{N+1, N+1}}$
	}
\end{array}
\!\right|
\medskip \\
\!\left|
\begin{array}{cc}
	\!\!\left|
	C_{_{\widehat{N+1},~\!\widehat{N+1}}}
	\right|_{_{K,K}}  
	&
	\!\!\left|
	C_{_{\widehat{K},~\!\widehat{N+1}}}
	\right|_{_{N+1,K}}
	\smallskip \\
	\!\!\left|
	C_{_{\widehat{N+1},~\!\widehat{K}}}
	\right|_{_{K, N+1}}
	&
	\!\!\fbox{
		$\left|
		C_{_{\widehat{K},~\!\widehat{K}}}
		\right|_{_{N+1, N+1}}$
	}
\end{array}
\!\right|
&
\!\left|
\begin{array}{cc}
	\!\!\left|
	D_{_{\widehat{N+1},~\!\widehat{N+1}}}
	\right|_{_{K,K}}  
	&
	\!\!\left|
	D_{_{\widehat{K},~\!\widehat{N+1}}}
	\right|_{_{N+1,K}}
	\smallskip \\
	\!\!\left|
	D_{_{\widehat{N+1},~\!\widehat{K}}}
	\right|_{_{K, N+1}}
	&
	\!\!\fbox{
		$\left|
		D_{_{\widehat{K},~\!\widehat{K}}}
		\right|_{_{N+1, N+1}}$
	}
\end{array}
\!\right|
\end{array}
\!\right) 
. ~~~~~~~~
\end{eqnarray}
Comparing with \eqref{J_asymptotic}, we have
\begin{eqnarray}
	&& \left|
	A_{_{\widehat{N+1}, \widehat{N+1}}}
	\right|_{_{K, K}} 
	=
	\left|
	B_{_{\widehat{N+1}, \widehat{N+1}}}
	\right|_{_{K, K}}
	=
	\left|
	C_{_{\widehat{N+1}, \widehat{N+1}}}
	\right|_{_{K, K}} 
	=
	\left|
	D_{_{\widehat{N+1}, \widehat{N+1}}}
	\right|_{_{K, K}}      \nonumber \\
	&\!\!\!\!=\!\!\!\!&
	\left|
	\begin{array}{cccccccc}
		\!\!\gamma_{_{11}} & \!\!\!\!\cdots & \!\!\!\!\gamma_{_{1,K-1}} & \!\!\!\!\gamma_{_{1,K}} E_{_{K}}^{\scriptscriptstyle(+)} & \!\!\!\!0 & \!\!\!\!\cdots & \!\!\!\!0  
		\\
		\!\!\vdots & \!\!\!\!\ddots & \!\!\!\!\vdots & \!\!\!\!\vdots & \!\!\!\!\vdots & \!\!\!\!\ddots & \!\!\!\! \vdots  
		\\
		\!\!\gamma_{_{K-1,1}} & \!\!\!\!\cdots & \!\!\!\!\gamma_{_{K-1,K-1}} & \!\!\!\!\gamma_{_{K-1,K}}E_{_{K}}^{\scriptscriptstyle(+)} & \!\!\!\!0 & \!\!\!\!\cdots & \!\!\!\!0  
		\smallskip \\
		\!\!\gamma_{_{K,1}}\overline{E}_{_{K}}^{\scriptscriptstyle(+)} & \!\!\!\!\cdots & \!\!\!\!\gamma_{_{K,K-1}}\overline{E}_{_{K}}^{\scriptscriptstyle(+)} & \!\! \fbox{$\gamma_{_{K, K}}\!\left( |E_{_{K}}^{\scriptscriptstyle(+)}|^2 + |E_{_{K}}^{\scriptscriptstyle(-)}|^2 \right)$}
		& \!\!\gamma_{_{K,K+1}}\overline{E}_{_{K}}^{\scriptscriptstyle(-)} & \!\!\cdots & \!\!\gamma_{_{K,N}}\overline{E}_{_{K}}^{\scriptscriptstyle(-)}
		\smallskip \\
		\!\!0 & \!\!\!\!\cdots & \!\!\!\!0 & \!\!\!\!\gamma_{_{K+1,K}}E_{_{K}}^{\scriptscriptstyle(-)} & \!\!\!\!\gamma_{_{K+1,K+1}} & \!\!\!\!\cdots & \!\!\!\!\gamma_{_{K+1,N}}  
		\\
		\!\!\vdots & \!\!\!\!\ddots & \!\!\!\!\vdots & \!\!\!\!\vdots & \!\!\!\!\vdots & \!\!\!\!\ddots & \!\!\!\!\vdots 
		\\
		\!\!0 & \!\!\!\!\cdots & \!\!\!\!0 & \!\!\!\!\gamma_{_{N,K}}E_{_{K}}^{\scriptscriptstyle(-)} & \!\!\!\!\gamma_{_{N,K+1}} & \!\!\!\!\cdots & \!\!\!\!\gamma_{_{N,N}} 
	\end{array}
	\!\!\right|
\end{eqnarray}
(By definition \eqref{Quasideterminant_defn_concrete form} and the multiplication rule (Proposition 3.2), the above expression can be decomposed into four quasideterminants.)
\begin{eqnarray}	
=
\begin{array}{l}
	\left|
	\begin{array}{cccccccc}
		\!\!\gamma_{_{11}} & \!\!\!\!\!\cdots & 
		\!\!\!\!\!\gamma_{_{1,K}}  & \!\!\!\!\!0 & \!\!\!\!\!\cdots & \!\!\!\!\!0  
		\\
		\!\!\vdots & \!\!\!\!\!\ddots & 
		\!\!\!\!\!\vdots & \!\!\!\!\!\vdots & \!\!\!\!\!\ddots & \!\!\!\!\! \vdots  
		\\
		\!\!\gamma_{_{K,1}} & \!\!\!\!\!\cdots & 
		\!\!\!\!\! \fbox{$\gamma_{_{K, K}}$}
		& \!\!\!\!\!0 & \!\!\!\!\!\cdots & \!\!\!\!\!0
		\smallskip \\
		\!\!0 & \!\!\!\!\!\cdots & 
		\!\!\!\!\!0 & \!\!\!\!\!\gamma_{_{K+1,K+1}} & \!\!\!\!\!\cdots & \!\!\!\!\!\gamma_{_{K+1,N}}  
		\\
		\!\!\vdots & \!\!\!\!\!\ddots & 
		\!\!\!\!\!\vdots & \!\!\!\!\!\vdots & \!\!\!\!\!\ddots & \!\!\!\!\!\vdots 
		\\
		\!\!0 & \!\!\!\!\!\cdots & 
		\!\!\!\!\!0 & \!\!\!\!\!\gamma_{_{N,K+1}} & \!\!\!\!\!\cdots & \!\!\!\!\!\gamma_{_{N,N}} 
	\end{array}
	\!\!\!\right|\!|E_{_{K}}^{\scriptscriptstyle(+)}|^2
	\!+
	\left|
	\begin{array}{cccccccc}
		\!\!\gamma_{_{11}} & \!\!\!\!\!\cdots & \!\!\!\!\!\gamma_{_{1,K-1}} & \!\!\!\!\!0  & 
		\!\!\!\!\!\cdots & \!\!\!\!\!0  
		\\
		\!\!\vdots & \!\!\!\!\!\ddots & \!\!\!\!\!\vdots & \!\!\!\!\!\vdots & 
		\!\!\!\!\!\ddots & \!\!\!\!\! \vdots  
		\\
		\!\!\gamma_{_{K-1,1}} & \!\!\!\!\!\cdots & \!\!\!\!\!\gamma_{_{K-1,K-1}} & \!\!\!\!\!0 & 
		\!\!\!\!\!\cdots & \!\!\!\!\!0 
		\smallskip \\
		\!\!0 & \!\!\!\!\!\cdots & \!\!\!\!\!0 & \!\!\!\!\! \fbox{$\gamma_{_{K, K}}$} & 
		\!\!\!\!\!\cdots & \!\!\!\!\!\gamma_{_{K, N}}
		\\
		\!\!\vdots & \!\!\!\!\!\ddots & \!\!\!\!\!\vdots & \!\!\!\!\!\vdots & 
		\!\!\!\!\!\ddots & \!\!\!\!\!\vdots 
		\\
		\!\!0 & \!\!\!\!\!\cdots & \!\!\!\!\!0 & \!\!\!\!\!\gamma_{_{N, K}} & 
		\!\!\!\!\!\cdots & \!\!\!\!\!\gamma_{_{N,N}} 
	\end{array}
	\!\!\!\right|\!|E_{_{K}}^{\scriptscriptstyle(-)}|^2 
	\smallskip \\
+ \left|
	\begin{array}{cccccccc}
		\!\!\gamma_{_{11}} & \!\!\!\!\!\cdots & 
		\!\!\!\!\!\gamma_{_{1,K}}  & \!\!\!\!\!0 & \!\!\!\!\!\cdots & \!\!\!\!\!0  
		\\
		\!\!\vdots & \!\!\!\!\!\ddots & 
		\!\!\!\!\!\vdots & \!\!\!\!\!\vdots & \!\!\!\!\!\ddots & \!\!\!\!\! \vdots  
		\\
		\!\!\gamma_{_{K-1,1}} & \!\!\!\!\!\cdots & 
		\!\!\!\!\!\gamma_{_{K-1,K}} & \!\!\!\!\!0 & \!\!\!\!\!\cdots & \!\!\!\!\!0 
		\smallskip \\
		\!\!\!0 & \!\!\!\!\!\!\cdots & 
		\!\!\!\! \fbox{$0$}
		& \!\!\!\!\!\gamma_{_{K, K+1}} & \!\!\!\!\!\cdots & \!\!\!\!\!\gamma_{_{K, N}}
		\\
		\!\!\vdots & \!\!\!\!\!\ddots & 
		\!\!\!\!\!\vdots & \!\!\!\!\!\vdots & \!\!\!\!\!\ddots & \!\!\!\!\!\vdots 
		\\
		\!\!0 & \!\!\!\!\!\cdots & 
		\!\!\!\!\!0 & \!\!\!\!\!\gamma_{_{N,K+1}} & \!\!\!\!\!\cdots & \!\!\!\!\!\gamma_{_{N,N}} 
	\end{array}
	\!\!\!\right|\!E_{_{K}}^{\scriptscriptstyle(+)}\overline{E}_{_{K}}^{\scriptscriptstyle(-)}
	\!\!+
	\!\left|
	\begin{array}{cccccccc}
		\!\!\gamma_{_{11}} & \!\!\!\!\!\cdots & \!\!\!\!\!\gamma_{_{1,K-1}} & \!\!\!\!\!0  & 
		\!\!\!\!\!\cdots & \!\!\!\!\!0  
		\\
		\!\!\vdots & \!\!\!\!\!\ddots & \!\!\!\!\!\vdots & \!\!\!\!\!\vdots & 
		\!\!\!\!\!\ddots & \!\!\!\!\! \vdots  
		\\
		\!\!\gamma_{_{K, 1}} & \!\!\!\!\!\cdots & \!\!\!\!\!\gamma_{_{K, K-1}} & \!\!\!\!\! \fbox{$0$} & 
		\!\!\!\!\!\cdots & \!\!\!\!\!0
		\smallskip \\
		\!\!0 & \!\!\!\!\!\cdots & \!\!\!\!\!0 & \!\!\!\!\!\gamma_{_{K+1, K}} & 
		\!\!\!\!\!\cdots & \!\!\!\!\!\gamma_{_{K+1,N}}  
		\\
		\!\!\vdots & \!\!\!\!\!\ddots & \!\!\!\!\!\vdots & \!\!\!\!\!\vdots & 
		\!\!\!\!\!\ddots & \!\!\!\!\!\vdots 
		\\
		\!\!0 & \!\!\!\!\!\cdots & \!\!\!\!\!0 & \!\!\!\!\!\gamma_{_{N, K}} & 
		\!\!\!\!\!\cdots & \!\!\!\!\!\gamma_{_{N,N}} 
	\end{array}
	\!\!\!\right|\!\overline{E}_{_{K}}^{\scriptscriptstyle(+)}E_{_{K}}^{\scriptscriptstyle(-)}   
\end{array}
\end{eqnarray}
(By Lemmas 5.1 and 5.2, the first two terms can be expressed explicitly. The remaining two terms vanish by using the column operation  (Proposition 3.3), since  $(\gamma_{_{1,K}}, \cdots, \gamma_{_{K-1,K}})^T$ lies in the span of the $K-1$ column vectors $(\gamma_{_{11}}, \cdots, \gamma_{_{K-1,1}})^T$, $\cdots$ $(\gamma_{_{1,K-1}}, \cdots, \gamma_{_{K-1,K-1}})^T$, while $(\gamma_{_{K+1,K}}, \cdots, \gamma_{_{N,K}})^T$ lies in the span of the $N-K$ column vectors $(\gamma_{_{K+1, K+1}}, \cdots, \gamma_{_{N, K+1}})^T$, $\cdots$ $(\gamma_{_{K+1, N}}, \cdots, \gamma_{_{N, N}})^T$.)
\begin{eqnarray}
	&\!\!\!\!=\!\!\!\!&
	\left|
	\begin{array}{cccccccc}
		\!\!\gamma_{_{11}} & \!\!\!\!\!\cdots &
		\!\!\!\!\!\gamma_{_{1,K}}  
		\\
		\!\!\vdots & \!\!\!\!\!\ddots & 
		\!\!\!\!\!\vdots 
		\\
		\!\!\gamma_{_{K,1}} & \!\!\!\!\!\cdots &  
		\!\!\!\!\! \fbox{$\gamma_{_{K, K}}$}
	\end{array}
	\!\right|\!|E_{_{K}}^{\scriptscriptstyle(+)}|^2
	+
	\left|
	\begin{array}{cccccccc}
		\!\!\fbox{$\gamma_{_{K, K}}$} &  \!\!\!\!\!\cdots & \!\!\!\!\!\gamma_{_{K, N}}
		\\
		\!\!\vdots & \!\!\!\!\!\ddots & \!\!\!\!\!\vdots 
		\\
		\!\!\gamma_{_{N, K}} & \!\!\!\!\!\cdots & \!\!\!\!\!\gamma_{_{N,N}} 
	\end{array}
	\!\!\!\right|\!|E_{_{K}}^{\scriptscriptstyle(-)}|^2  	  
 \\
	&\!\!\!\!=\!\!\!\!&
	\frac{\mu_{_{K}}}{\mu_{_{K}}-\lambda_{_{K}}}\!\!\left[
	\prod_{j=1}^{\scriptscriptstyle{K-1}}
	\!\!\left(\frac{\lambda_{_{K}}-\lambda_{j}}{\lambda_{_{K}}-\mu_{j}} \right)\!\!\left(\frac{\mu_{_{K}}-\mu_{j}}{\mu_{_{K}}-\lambda_{j}} \right)\!|E_{_{K}}^{\scriptscriptstyle(+)}|^2 
	+\!\!\!\!\prod_{j={_{K+1}}}^{\scriptscriptstyle{N}}
	\!\!\!\!\left(\frac{\lambda_{_{K}}-\lambda_{j}}{\lambda_{_{K}}-\mu_{j}} \right)\!\!\left(\frac{\mu_{_{K}}-\mu_{j}}{\mu_{_{K}}-\lambda_{j}} \right)\!|E_{_{K}}^{\scriptscriptstyle(-)}|^2 
	\right]. \label{A_11=B_11=C_11=D_11}
\end{eqnarray}
The remaining verification of \eqref{J_asymptotic_Jacobi identity} is analogous to the above. We therefore present only the essential formulas without further detailed explanation.
Using Propositions 3.1 -- 3.3 and Lemmas 5.1 and 5.3, one can obtain
\begin{eqnarray}
\label{A_12=C_12}	
	&&\left|
	A_{\widehat{_{K}}, \widehat{_{N+1}}}
	\right|_{_{N+1, K}}  
	=\left|
	C_{\widehat{_{K}}, \widehat{_{N+1}}}
	\right|_{_{N+1, K}}    \nonumber \\ 
	&\!\!\!\!=\!\!\!\!&
	\left|
	\begin{array}{cccccccc}
		\!\!\gamma_{_{11}} & \!\!\!\!\cdots & \!\!\!\!\gamma_{_{1,K-1}} & 
		& \!\!\!\!0 & \!\!\!\!\cdots & \!\!\!\!0 & \!\!\!\!1 
		\\
		\!\!\vdots & \!\!\!\!\ddots & \!\!\!\!\vdots & 
		& \!\!\!\!\vdots & \!\!\!\!\ddots & \!\!\!\! \vdots & \!\!\!\!\vdots
		\\
		\!\!\gamma_{_{K-1,1}} & \!\!\!\!\cdots & \!\!\!\!\gamma_{_{K-1,K-1}} & 
		& \!\!\!\!0 & \!\!\!\!\cdots & \!\!\!\!0  & \!\!\!\!1
		\\
		\!\!\gamma_{_{K,1}}\overline{E}_{_{K}}^{\scriptscriptstyle(+)} & \!\!\!\!\cdots & \!\!\!\!\gamma_{_{K,K-1}}\overline{E}_{_{K}}^{\scriptscriptstyle(+)} & 
		& \!\!\!\!\gamma_{_{K,K+1}}\overline{E}_{_{K}}^{\scriptscriptstyle(-)} & \!\!\cdots & \!\!\gamma_{_{K,N}}\overline{E}_{_{K}}^{\scriptscriptstyle(-)}
		& \!\! \fbox{$\overline{E}_{_{K}}^{\scriptscriptstyle(+)}$}
		\\
		\!\!0 & \!\!\!\!\cdots & \!\!\!\!0 & 
		& \!\!\!\!\gamma_{_{K+1,K+1}} & \!\!\!\!\cdots & \!\!\!\!\gamma_{_{K+1,N}} & \!\!\!\!0
		\\
		\!\!\vdots & \!\!\!\!\ddots & \!\!\!\!\vdots & 
		& \!\!\!\!\vdots & \!\!\!\!\ddots & \!\!\!\!\vdots & \!\!\!\!\vdots
		\\
		\!\!0 & \!\!\!\!\cdots & \!\!\!\!0 & 
		& \!\!\!\!\gamma_{_{N,K+1}} & \!\!\!\!\cdots & \!\!\!\!\gamma_{_{N,N}} & \!\!\!\! 0
	\end{array}
	\!\!\right|  \nonumber \\
	&\!\!\!\!=\!\!\!\!&
	\left|
	\begin{array}{cccccccc}
		\!\!\gamma_{_{11}} & \!\!\!\!\!\!\cdots & \!\!\!\!\!\!\gamma_{_{1,K-1}} & \!\!\!\!1 
		& \!\!\!\!0 & \!\!\!\!\!\!\cdots & \!\!\!\!\!\!0 
		\\
		\!\!\vdots & \!\!\!\!\!\!\ddots & \!\!\!\!\!\!\vdots & \!\!\!\!\vdots
		& \!\!\!\!\vdots & \!\!\!\!\!\!\ddots & \!\!\!\!\!\! \vdots 
		\\
		\!\!\gamma_{_{K,1}} & \!\!\!\!\!\!\cdots & \!\!\!\!\!\!\gamma_{_{K,K-1}} &  \!\!\!\! \fbox{$1$} & \!\!\!\!0 
		& \!\!\!\!\!\!\cdots & \!\!\!\!\!\!0  
		\\
		\!\!0 & \!\!\!\!\!\!\cdots & \!\!\!\!\!\!0 & \!\!\!\!0
		& \!\!\!\!\gamma_{_{K+1,K+1}} & \!\!\!\!\!\!\cdots & \!\!\!\!\!\!\gamma_{_{K+1,N}} 
		\\
		\!\!\vdots & \!\!\!\!\!\!\ddots & \!\!\!\!\!\!\vdots & \!\!\!\!\vdots 
		& \!\!\!\!\vdots & \!\!\!\!\!\!\ddots & \!\!\!\!\!\!\vdots 
		\\
		\!\!0 & \!\!\!\!\!\!\cdots & \!\!\!\!\!\!0 & \!\!\!\!0
		& \!\!\!\!\gamma_{_{N,K+1}} & \!\!\!\!\!\!\cdots & \!\!\!\!\!\!\gamma_{_{N,N}} 
	\end{array}
	\!\!\!\right|\!\overline{E}_{_{K}}^{\scriptscriptstyle(+)}
	\!+\!\left|
	\begin{array}{cccccccc}
		\!\!\gamma_{_{11}} & \!\!\!\!\!\!\cdots & \!\!\!\!\!\!\gamma_{_{1,K-1}} & \!\!\!\!\!\!1 
		& \!\!\!\!0 & \!\!\!\!\!\!\cdots & \!\!\!\!\!\!0 
		\\
		\!\!\vdots & \!\!\!\!\!\!\ddots & \!\!\!\!\!\!\vdots & \!\!\!\!\!\!\vdots
		& \!\!\!\!\vdots & \!\!\!\!\!\!\ddots & \!\!\!\!\!\! \vdots 
		\\
		\!\!\gamma_{_{K-1,1}} & \!\!\!\!\!\!\cdots & \!\!\!\!\!\!\gamma_{_{K-1,K-1}} & \!\!\!\!1 
		& \!\!\!\!0 & \!\!\!\!\!\!\cdots & \!\!\!\!\!\!0  
		\\
		\!\!0 & \!\!\!\!\!\!\cdots & \!\!\!\!\!\!0 &  \!\!\!\! \fbox{$0$} & \!\!\!\!\gamma_{_{K,K+1}} 
		& \!\!\!\!\!\!\cdots & \!\!\!\!\!\!\gamma_{_{K,N}} 
		\\
		\!\!\vdots & \!\!\!\!\!\!\ddots & \!\!\!\!\!\!\vdots & \!\!\!\!\vdots 
		& \!\!\!\!\vdots & \!\!\!\!\!\!\ddots & \!\!\!\!\!\!\vdots 
		\\
		\!\!0 & \!\!\!\!\!\!\cdots & \!\!\!\!\!\!0 & \!\!\!\!0
		& \!\!\!\!\gamma_{_{N,K+1}} & \!\!\!\!\!\!\cdots & \!\!\!\!\!\!\gamma_{_{N,N}} 
	\end{array}
	\!\!\!\right|\!\overline{E}_{_{K}}^{\scriptscriptstyle(-)}
	\nonumber \\
	&\!\!\!\!=\!\!\!\!&
	\left|
	\begin{array}{cccccccc}
		\!\!\gamma_{_{11}} & \!\!\!\!\cdots & \!\!\!\!\gamma_{_{1,K-1}} & \!\!\!\!1 
		\\
		\!\!\vdots & \!\!\!\!\ddots & \!\!\!\!\vdots & \!\!\!\!\vdots 
		\\
		\!\!\gamma_{_{K,1}} & \!\!\!\!\cdots & \!\!\!\!\gamma_{_{K,K-1}} &  \!\!\!\! \fbox{$1$}
	\end{array}
	\!\!\right|\overline{E}_{_{K}}^{\scriptscriptstyle(+)}
	=\prod_{j=1}^{_{K-1}}(\frac{\lambda_{j}}{\mu_{j}})\left(\frac{\mu_{_{K}}-\mu_{j}}{\mu_{_{K}}-\lambda_{j}}\right)\overline{E}_{_{K}}^{\scriptscriptstyle(+)}.   
\end{eqnarray}
Similarly, 
\begin{eqnarray}
	\label{B_12=D_12}	
\left|
	B_{\widehat{_{K}}, \widehat{_{N+1}}}
	\right|_{_{N+1, K}}  
	&\!\!\!\!=\!\!\!\!&
\left|
	D_{\widehat{_{K}}, \widehat{_{N+1}}}
	\right|_{_{N+1, K}}   
	=
	\!\left|
	\begin{array}{cccccccc}
		\!\!\gamma_{_{K, K+1}} & \!\!\!\!\cdots & \!\!\!\!\gamma_{_{K,N}} & \!\!\!\!\fbox{1} 
		\\
		\!\!\vdots & \!\!\!\!\ddots & \!\!\!\!\vdots & \!\!\!\!\vdots 
		\\
		\!\!\gamma_{_{N, K+1}} & \!\!\!\!\cdots & \!\!\!\!\gamma_{_{N, N}} &  \!\!\!\! 1
	\end{array}
	\!\!\right|\!\overline{E}_{_{K}}^{\scriptscriptstyle(-)}
	=
	\!\!\prod_{j=_{K+1}}^{_{N}}\!(\frac{\lambda_{j}}{\mu_{j}})\!\left(\frac{\mu_{_{K}}-\mu_{j}}{\mu_{_{K}}-\lambda_{j}}\right)\!\overline{E}_{_{K}}^{\scriptscriptstyle(-)}, ~~~~~~~~~
\\
\label{A_21=B=21}
\left|
	A_{\widehat{_{N+1}}, \widehat{_{K}}}
	\right|_{_{K, N+1}}
	&\!\!\!\!=\!\!\!\!&
\left|
	B_{\widehat{_{N+1}}, \widehat{_{K}}}
	\right|_{_{K, N+1}} 
=	\left|
	\begin{array}{cccccccc}
		\!\!\gamma_{_{11}} & \!\!\!\!\cdots & 
		\!\!\!\!\gamma_{_{1,K}} 
		\\
		\!\!\vdots & \!\!\!\!\ddots & 
		\!\!\!\!\vdots 
		\\
		\!\!\gamma_{_{K-1,1}} & \!\!\!\!\cdots & 
		\!\!\!\!\gamma_{_{K-1,K}}  
		\smallskip \\
		\!\! 1 & \!\!\!\!\cdots & 
		\!\!\!\! \fbox{$1$} 
	\end{array}
	\!\!\right|E_{_{K}}^{\scriptscriptstyle(+)}
	=
	\prod_{j=1}^{_{K-1}}
	\left(
	\frac{\lambda_{_{K}}-\lambda_{j}}{\lambda_{_{K}}-\mu_{j}}
	\right)\!E_{_{K}}^{\scriptscriptstyle(+)},
\\
	\label{C_21=D=21}
\left|
	C_{\widehat{_{N+1}}, \widehat{_{K}}}
	\right|_{_{K, N+1}}
	&\!\!\!\!=\!\!\!\!&
\left|
	D_{\widehat{_{N+1}}, \widehat{_{K}}}
	\right|_{_{K, N+1}}   
	=
	\left|
\begin{array}{cccccccc}
	\!\!\gamma_{_{K+1, K}} & \!\!\!\!\cdots  & 
	\!\!\!\!\gamma_{_{K+1, N}} 
	\\
	\!\!\vdots & \!\!\!\!\ddots & 
	\!\!\!\!\vdots 
	\\
	\!\!\gamma_{_{N, K}} & \!\!\!\!\cdots &
	\!\!\!\!\gamma_{_{N, N}}  
	\smallskip \\
	\!\! \fbox{1} & \!\!\!\!\cdots & 
	\!\!\!\! 1 
\end{array}
\!\!\right|\!E_{_{K}}^{\scriptscriptstyle(-)}
=
\!\prod_{j=_{K+1}}^{_{N}}
\!\!\left(
\frac{\lambda_{_{K}}-\lambda_{j}}{\lambda_{_{K}}-\mu_{j}}
\right)\!E_{_{K}}^{\scriptscriptstyle(-)}. ~~~~~~	
\end{eqnarray}	
Using Proposition 3.2 and Lemmas 5.1 and 5.4, one obtains
\begin{eqnarray}
\label{A_22}
&&\left|
A_{\widehat{_{K}},~\!\widehat{_{K}}}
\right|_{_{N+1, N+1}}   \nonumber \\
&\!\!\!\!=\!\!\!\!&
\left|
\begin{array}{cccccccc}
	\!\!\gamma_{_{11}} & \!\!\!\!\cdots & \!\!\!\!\gamma_{_{1,K-1}}   
	& \!\!\!\!0 & \!\!\!\!\cdots & \!\!\!\!0 & \!\!\!\!1 
	\\
	\!\!\vdots & \!\!\!\!\ddots & \!\!\!\!\vdots &  \!\!\!\!\vdots & \!\!\!\!\ddots & \!\!\!\! \vdots & \!\!\!\!\vdots
	\\
	\!\!\gamma_{_{K-1,1}} & \!\!\!\!\cdots & \!\!\!\!\gamma_{_{K-1,K-1}} &  \!\!\!\!0 & \!\!\!\!\cdots & \!\!\!\!0  & \!\!\!\!1
	\\
	\!\!0 & \!\!\!\!\cdots & \!\!\!\!0
	& \!\!\!\!\gamma_{_{K+1,K+1}} & \!\!\!\!\cdots & \!\!\!\!\gamma_{_{K+1,N}} & \!\!\!\!0
	\\
	\!\!\vdots & \!\!\!\!\ddots & \!\!\!\!\vdots &  \!\!\!\!\vdots & \!\!\!\!\ddots & \!\!\!\!\vdots & \!\!\!\!\vdots
	\\
	\!\!0 & \!\!\!\!\cdots & \!\!\!\!0 
	& \!\!\!\!\gamma_{_{N,K+1}} & \!\!\!\!\cdots & \!\!\!\!\gamma_{_{N,N}} & \!\!\!\! 0
	\\
	\!\! 1 & \!\!\!\!\cdots & \!\!\!\!1 
	& \!\!\!\!0 & \!\!\!\!\cdots & \!\!\!\!0 & \!\!\!\! \fbox{1} 
\end{array}
\!\right|   
=	\left|
	\begin{array}{cccccccc}
	\!\! \gamma_{_{11}} & \!\!\!\!\cdots & \!\!\!\!\gamma_{_{1, K-1}} & \!\!\!1   
	\\
	\!\!\vdots & \!\!\ddots & \!\!\!\!\vdots & \!\!\!\vdots
	\\
	\!\!\gamma_{_{K-1, 1}} & \!\!\!\!\cdots & \!\!\!\!\gamma_{_{K-1, K-1}} & \!\!\! 1 
	\\
	\!\! 1 & \!\!\!\!\cdots & \!\!\!\! 1  & \!\!\! \fbox{1}
	\end{array}
	\!\right|
	=
	\prod_{j=1}^{_{K-1}}(\frac{\lambda_{j}}{\mu_{j}}).  ~~~~~~~~ 
\end{eqnarray}
Similarly, 
\begin{eqnarray}
	\label{D_22}
	\left|
	D_{\widehat{_{K}},~\!\widehat{_{K}}}
	\right|_{_{N+1, N+1}}   
	=
		\left|
	\begin{array}{cccccccc}
		\!\! \fbox{1} & \!\! 1 & \!\!\!\!\cdots & \!\!\!\! 1 
		\\
		\!\!1 & \!\!\gamma_{_{K+1, K+1}} & \!\!\!\!\cdots & \!\!\!\!\gamma_{_{K+1, N}} 
		\\
		\!\!\vdots & \!\!\ddots & \!\!\!\!\vdots & \!\!\!\!\vdots
		\\
		\!\!1 & \!\!\gamma_{_{N, K+1}} & \!\!\!\!\cdots & \!\!\!\!\gamma_{_{N, N}}
	\end{array}
	\!\right|
	=
	\prod_{j=_{K+1}}^{_{N}}\!(\frac{\lambda_{j}}{\mu_{j}}), ~~ 
\end{eqnarray}	
\begin{eqnarray}
\label{B_22=C_22=0}
\left|
B_{\widehat{_{K}},~\!\widehat{_{K}}}
\right|_{_{N+1, N+1}} 
= \left|
C_{\widehat{_{K}},~\!\widehat{_{K}}}
\right|_{_{N+1, N+1}} 
=0.
\end{eqnarray}
Substituting these expressions into \eqref{J_asymptotic_Jacobi identity}, one obtains 
\begin{eqnarray}
\label{J_asymptotic_special region}
\mathcal{J}_{_{[N]}}^{\scriptscriptstyle(K)}	
=
	\frac{
		\left(
		\begin{array}{cc}
			\begin{array}{l}
			\widetilde{\gamma}_{_{KK}}
			p_{_K}\widetilde{p}_{_K}
			|E_{_{K}}^{\scriptscriptstyle(+)}|^2  
			\smallskip \\
			\!\!\!\! + \gamma_{_{KK}}q_{_K}\widetilde{q}_{_K}|E_{_{K}}^{\scriptscriptstyle(-)}|^2
			\end{array}
			&
			\!\!-p_{_K}\widetilde{q}_{_K}
			E_{_{K}}^{\scriptscriptstyle(+)}\overline{E}_{_{K}}^{\scriptscriptstyle(-)}
			\smallskip \\
			\!\!-\widetilde{p}_{_K}q_{_K}
			\overline{E}_{_{K}}^{\scriptscriptstyle(+)}E_{_{K}}^{\scriptscriptstyle(-)} 
			&
			\begin{array}{l}
			\gamma_{_{KK}}
			p_{_K}\widetilde{p}_{_K}
			|E_{_{K}}^{\scriptscriptstyle(+)}|^2
			\smallskip \\
			\!\!\!\! + \widetilde{\gamma}_{_{KK}}
			q_{_K}\widetilde{q}_{_K}
			|E_{_{K}}^{\scriptscriptstyle(-)}|^2			
	\end{array}	
	\end{array}	
		\!\!\right)
\!\!\left(
\begin{array}{cc}
	\!\!\displaystyle{\prod_{j=1}^{_{K-1}}}
	\!(\frac{\lambda_{j}}{\mu_{j}})   & \!\!\!\!\!\!0 
	\\
	\!\!\!\!\!\!\!\!0 & \!\!\!\!\!\!\!\!\!\!\!\displaystyle{\prod_{j=_{K+1}}^{_{N}}}\!\!(\frac{\lambda_{j}}{\mu_{j}})
\end{array}
\!\!\right)		
	}
	{
		\gamma_{_{KK}}(p_{_K}\widetilde{p}_{_K}|E_{_{K}}^{\scriptscriptstyle(+)}|^2
		+
		q_{_K}\widetilde{q}_{_K}|E_{_{K}}^{\scriptscriptstyle(-)}|^2)
	},  ~~~~~~~~~  
\end{eqnarray}
where
\begin{eqnarray}
	\begin{array}{l}
		\label{pq_K}	
		p_{_K}:=
		\displaystyle{\prod_{j=1}^{_{K-1}}}
		(\lambda_{_{K}}\!-\!\lambda_{j})(\mu_{_{K}}\!-\!\mu_{j}), ~~
		\widetilde{p}_{_K}:=
		\displaystyle{\prod_{j=_{K+1}}^{_{N}}}
		\!\!\!(\lambda_{_{K}}\!-\!\mu_{j})(\mu_{_{K}}\!-\!\lambda_{j})
		\medskip \\
		q_{_K}:=
		\displaystyle{\prod_{j=1}^{_{K-1}}}
		(\lambda_{_{K}}\!-\!\mu_{j})(\mu_{_{K}}\!-\!\lambda_{j}), ~~
		\widetilde{q}_{_K}:=
		\displaystyle{\prod_{j=_{K+1}}^{_{N}}}
		\!\!\!(\lambda_{_{K}}\!-\!\lambda_{j})(\mu_{_{K}}\!-\!\mu_{j})
	\end{array},
\end{eqnarray}
$\gamma_{jk}$ and $\widetilde{\gamma}_{jk}$ are defined in \eqref{gamma_jk} and \eqref{gamma-tilde_jk}, respectively.
Therefore, the asymptotic form 
$\mathcal{J}_{_{[N]}}^{\scriptscriptstyle (K)}$ 
of the $N$-soliton ansatz $J_{_{[N]}}$ 
in the asymptotic region \eqref{Asymptotic region_standard} 
is obtained. 
To describe the general case concisely, we extend the notation introduced in \eqref{epsilon symbols} to
\begin{eqnarray}
	\label{epsilon_pm}
	(\varepsilon_{j}^{+}, \varepsilon_{j}^{-}):=
	\left\{
	\begin{array}{l}
		(+, -), ~~\mbox{Re}(L_{j}) > M
		\\
		(-, +), ~~\mbox{Re}(L_{j}) < -M
	\end{array}, j \neq K
	\right.
\end{eqnarray}
to label the asymptotic regions and redefine the spectral parameters by
\begin{eqnarray}
	\label{lambda_pm}
	(\lambda_{j}, \mu_{j}):=(\lambda_{j}^{\scriptscriptstyle(+)}, \lambda_{j}^{\scriptscriptstyle(-)}).
\end{eqnarray}
To avoid confusion in the notation, we consider, as an example, the case where
$\operatorname{Re}(L_{j_{0}})<M$ for some fixed $j_{0}$. Then,
$
(\lambda_{j_{0}}^{(\varepsilon_{j_{0}}^{+})},
\lambda_{j_{0}}^{(\varepsilon_{j_{0}}^{-})})
=
(\mu_{j_{0}},\lambda_{j_{0}}).
$
\newtheorem{Thm}[Lem_5.1]{Theorem}
\begin{Thm}
Let $E_{_{K}}^{(\pm)}$ be defined by \eqref{E_j-pm}, with the spectral parameters redefined according to \eqref{lambda_pm}. Then, for the asymptotic region determined by the choice of $(\varepsilon_{j}^{+}, \varepsilon_{j}^{-})$ in \eqref{epsilon_pm}, the $N$-soliton ansatz $J_{_{[N]}}$ takes the following asymptotic form in the comoving frame associated with the $K$-th 1-soliton:
\begin{eqnarray}
\label{J_asymptotic_exact}
&&\!\!\!\!\mathcal{J}_{_{[N]}}^{\scriptscriptstyle(K)}	
\nonumber \\
&\!\!\!\!\!\!=\!\!\!\!\!\!&
\frac{
\!\!\!\!\!\!\left(
\begin{array}{cc}
\widetilde{\gamma}_{_{KK}}
p_{_K}\widetilde{p}_{_K}
|E_{_{K}}^{\scriptscriptstyle(+)}|^2
+
\gamma_{_{KK}}q_{_K}\widetilde{q}_{_K}|E_{_{K}}^{\scriptscriptstyle(-)}|^2
&
-p_{_K}\widetilde{q}_{_K}
E_{_{K}}^{\scriptscriptstyle(+)}\overline{E}_{_{K}}^{\scriptscriptstyle(-)}
\smallskip \\
-\widetilde{p}_{_K}q_{_K}
\overline{E}_{_{K}}^{\scriptscriptstyle(+)}E_{_{K}}^{\scriptscriptstyle(-)} 
&
\gamma_{_{KK}}
p_{_K}\widetilde{p}_{_K}
|E_{_{K}}^{\scriptscriptstyle(+)}|^2
+\widetilde{\gamma}_{_{KK}}
q_{_K}\widetilde{q}_{_K}
|E_{_{K}}^{\scriptscriptstyle(-)}|^2			
\end{array}	
\right)
}
{
\gamma_{_{KK}}(p_{_K}\widetilde{p}_{_K}|E_{_{K}}^{\scriptscriptstyle(+)}|^2
+
q_{_K}\widetilde{q}_{_K}|E_{_{K}}^{\scriptscriptstyle(-)}|^2)
}C^{\scriptscriptstyle(K)},  ~~~~~~~  
\end{eqnarray}
where $\gamma_{_{KK}}, \widetilde{\gamma}_{_{KK}}$ are defined in \eqref{gamma_jk} and \eqref{gamma-tilde_jk},
\begin{eqnarray}
\begin{array}{l}
\label{pq_K}	
p_{_K}:=
\displaystyle{\prod_{j=1, j \neq K}^{_{N}}}
\!\!\!\!(\lambda_{_{K}}^{\scriptscriptstyle(+)}\!-\!\lambda_{j}^{(\varepsilon_{j}^{+})}), ~~
\widetilde{p}_{_K}:=
\displaystyle{\prod_{j=1, j \neq K}^{_{N}}}
\!\!\!\!(\lambda_{_{K}}^{\scriptscriptstyle(-)}\!-\!\lambda_{j}^{(\varepsilon_{j}^{-})})
\medskip \\
q_{_K}:=
\displaystyle{\prod_{j=1, j \neq K}^{_{N}}}
\!\!\!\!(\lambda_{_{K}}^{\scriptscriptstyle(+)}\!-\!\lambda_{j}^{(\varepsilon_{j}^{-})}), ~~
\widetilde{q}_{_K}:=
\displaystyle{\prod_{j=1, j \neq K}^{_{N}}}
\!\!\!\!(\lambda_{_{K}}^{\scriptscriptstyle(-)}\!-\!\lambda_{j}^{(\varepsilon_{j}^{+})})
\end{array},
\end{eqnarray}
and
\begin{eqnarray}
C^{\scriptscriptstyle(K)}
:=
\mbox{diag}
\left( 
\displaystyle{\prod_{j=1, j \neq K}^{N}}
c_{j}(\varepsilon_{j}^{+}), ~
\displaystyle{\prod_{j=1, j \neq K}^{N}}
c_{j}(\varepsilon_{j}^{-})
\right)
\end{eqnarray}
is a constant matrix defined by 
\begin{eqnarray}
c_{j}(+):=\lambda_{j}^{(+)}/\lambda_{j}^{(-)}, ~~ 
c_{j}(-):=1
\end{eqnarray}
\end{Thm}
Comparing \eqref{J_asymptotic_exact} with \eqref{J_1-soliton}, one finds that $\mathcal{J}_{_{[N]}}^{\scriptscriptstyle(K)}$ can be obtain from $J_{[1]}$ by taking the following replacements  
\begin{eqnarray}
\begin{array}{l}
(E_{1}^{\scriptscriptstyle(+)}, ~\! \overline{E}_{1}^{\scriptscriptstyle(+)}) 
	\longrightarrow  
(p_{_{K}}E_{_{K}}^{\scriptscriptstyle(+)}, ~\! \widetilde{p}_{_{K}}\overline{E}_{_{K}}^{\scriptscriptstyle(+)}),
~~
(E_{1}^{\scriptscriptstyle(-)}, ~\! \overline{E}_{1}^{\scriptscriptstyle(-)}) 
\longrightarrow  
(q_{_{K}}E_{_{K}}^{\scriptscriptstyle(-)}, ~\! \widetilde{q}_{_{K}}\overline{E}_{_{K}}^{\scriptscriptstyle(-)}),
\medskip \\
(\gamma_{_{11}}, \widetilde{\gamma}_{_{11}})
\longrightarrow 
(\gamma_{_{KK}}, \widetilde{\gamma}_{_{KK}}), 
~~
I_{2\times2}  \longrightarrow  C^{\scriptscriptstyle(K)}. 
\end{array}
\end{eqnarray}
Therefore, the asymptotic form $\mathcal{J}_{_{[N]}}^{\scriptscriptstyle(K)}$ of the $N$-soliton ansatz is, in essence, a 1-soliton like object at the matrix level. To visualize the soliton behavior and determine the phase shift factors arising from soliton collisions, one must extract a real-valued physical quantity from $\mathcal{J}_{_{[N]}}^{\scriptscriptstyle(K)}$. The most natural choice is to consider the action density of the WZW$_4$ model, which serves as an analogue of the energy density for solitons.


\subsection{Quasi-Grammian $N$-Soliton in WZW$_{4}$ Model}

In this subsection, we aim to show that, asymptotically, the contribution of the $N$-soliton ansatz $J_{_{[N]}}$ to the WZW$_4$ action arises entirely from the NL$\sigma$M action density, which coincides with that of an exact 1-soliton solution $J_{[1]}$ up to a phase shift factor. Consequently, $J_{_{[N]}}$ satisfies the defining criteria for an $N$-soliton solution.
Before turning to the main discussion, it is necessary to verify the commutativity between taking asymptotic limits and partial differentiation when acting on  $J_{_{[N]}}$. This follows from the derivative formula for quasi-Grammian \eqref{Derivative formula_grammian} together with the commutativity property of
$\displaystyle{
	\lim_{r \rightarrow \infty}(\partial_{m}E_{j}^{\scriptscriptstyle(\pm)})
=
\partial_{m}(\lim_{r \rightarrow \infty}E_{j}^{\scriptscriptstyle(\pm)})
}$. 
The details are provided in Appendix C for brevity.
\newtheorem{Lem_5.6}[Lem_5.1]{Lemma}
\begin{Lem_5.6}
In the asymptotic regions \eqref{asymptotic region_epsilon}, we have
\begin{eqnarray}
\lim_{r \rightarrow \infty}(\partial_{m}J_{_{[N]}})	  
=
\partial_{m}\mathcal{J}_{_{[N]}}^{\scriptscriptstyle (K)}.
\end{eqnarray}
\end{Lem_5.6}
By Lemma 5.6 and \eqref{J_asymptotic_exact}, we have
\begin{eqnarray}
\lim_{r \rightarrow \infty}\mbox{Tr}\left[		(\partial_{m}J_{_{[N]}})J_{_{[N]}}^{-1}(\partial^{m}J_{_{[N]}})J_{_{[N]}}^{-1}\right]   
&\!\!\!\!=\!\!\!\!&\mbox{Tr}\left[
(\partial_{m}\mathcal{J}_{_{[N]}}^{\scriptscriptstyle (K)})(\mathcal{J}_{_{[N]}}^{\scriptscriptstyle (K)})^{-1}(\partial^{m}\mathcal{J}_{_{[N]}}^{\scriptscriptstyle (K)})(\mathcal{J}_{_{[N]}}^{\scriptscriptstyle (K)})^{-1}\right]
\\
&\!\!\!\!=\!\!\!\!&
\label{Trace formula_C}
\mbox{Tr}\left[
(\partial_{m}\mathcal{J})\mathcal{J}^{-1}(\partial^{m}\mathcal{J})\mathcal{J}^{-1}\right],~~\mathcal{J}=\mathcal{J}_{_{[N]}}^{\scriptscriptstyle (K)}(C^{\scriptscriptstyle(K)})^{-1}.  ~~~~~~
\end{eqnarray}
Substituting \eqref{J_asymptotic_exact} into the above expression and performing direct computations, we obtain the following result. The details are referred to Appendix D.
\newtheorem{Thm5.7}[Lem_5.1]{Theorem}
\begin{Thm5.7}
In the asymptotic regions \eqref{asymptotic region_epsilon}, we have	
\begin{eqnarray}
\label{WZ action_asymptotic}
(1)~\lim_{r \rightarrow \infty}\mbox{Tr}\left[
(\partial_{m}J_{_{[N]}})J_{_{[N]}}^{-1}(\partial_{n}J_{_{[N]}})J_{_{[N]}}^{-1}(\partial_{p}J_{_{[N]}})J_{_{[N]}}^{-1}
\right] =0
\end{eqnarray}
on each real space and
\begin{eqnarray}
\label{NLSM action_asymptotic}	
(2) ~\lim_{r \rightarrow \infty}\mbox{Tr}\left[		(\partial_{m}J_{_{[N]}})J_{_{[N]}}^{-1}(\partial^{m}J_{_{[N]}})J_{_{[N]}}^{-1}\right]  
=-2d_{_{KK}}\mbox{sech}^2\left[
X_{_K} + \log(\frac{|a_{_K}^{(+)}|}{|a_{_K}^{(-)}|}) + \delta_{_K}
\right]	~~~~~~
\end{eqnarray}
is real-valued on $\mathbb{M}^{3,1}$, $\mathbb{U}_{1}^{2,2}$, $\mathbb{U}_{2}^{2,2}$ for soliton number $N \in \mathbb{N}$, and on $\mathbb{E}^{4}$ for $N=2n-1, n \in \mathbb{N}$,
where $X_{_{K}}:= L_{_{K}} + \overline{L}_{_{K}}$, 
$d_{_{KK}}
:=\displaystyle{\frac{(\lambda_{_{K}}^{\scriptscriptstyle(+)} - \lambda_{_{K}}^{\scriptscriptstyle(-)})^3(\alpha_{_{K}}\widetilde{\beta}_{_{K}}-\widetilde{\alpha}_{_{K}}\beta_{_{K}}) }{\lambda_{_{K}}^{\scriptscriptstyle(+)}\lambda_{_{K}}^{\scriptscriptstyle(-)}}}$, 
and the phase shift factors are
\begin{eqnarray}
\label{Phase shift factors_explicit form}	
\delta_{_K}:=
\frac{1}{2}\log\!\!\!\!\!\prod_{j=1, j \neq K}^{N}\!\frac{
(\lambda_{_{K}}^{\scriptscriptstyle(+)}-\lambda_{j}^{(\varepsilon_{j}^{+})})(\lambda_{_{K}}^{\scriptscriptstyle(-)}-\lambda_{j}^{(\varepsilon_{j}^{-})})
}
{
(\lambda_{_{K}}^{\scriptscriptstyle(+)}-\lambda_{j}^{(\varepsilon_{j}^{-})})(\lambda_{j}^{(-)}-\lambda_{j}^{(\varepsilon_{j}^{+})})
},
\end{eqnarray}
here $\varepsilon_{j}^{\pm}$ is defined in \eqref{epsilon_pm} depending on the choice of asymptotic regions, and $\lambda_{j}^{(\pm)}$ is define in \eqref{lambda_pm} depending on the reality condition of each real space given in \eqref{Reality conditions}.
\end{Thm5.7}
Comparing \eqref{WZ action_asymptotic} and \eqref{NLSM action_asymptotic} with \eqref{WZ action_2} and \eqref{NLSM action}, we find that, asymptotically, the contribution to the WZW$_4$ model associated with the $N$-soliton ansatz $J_{_{[N]}}$ arises entirely from the NL$\sigma$M term. 
On the other hand, since solitons preserve their profiles during propagation, it follows that the WZ term does not contribute to the characterization of solitonic behavior.
Therefore, the solitonic behavior in the WZW$_4$ model is dominated asymptotically by \eqref{NLSM action_asymptotic} which is real-valued and coincides with that of an exact 1-soliton solution $J_{[1]}$ (cf. \eqref{NLSM_1-soliton}) up to a explicit phase shift factor $\delta_{_{K}}$. 
Note also that the real-valuedness of the asymptotic NL$\sigma$M action density $\mathcal{L}_{\sigma}\left[\mathcal{J}_{_{[N]}}^{\scriptscriptstyle(K)}\right]$ follows directly from the amplitude $d_{_{KK}}$ of \eqref{NLSM action_asymptotic} and the phase shift factor $\delta_{_{K}}$.
Indeed,  $d_{_{KK}}$ is real-valued on each real space (cf. $j=K$ in \eqref{Reality conditions}). Furthermore,
\begin{eqnarray}
\label{Phase shifts_each space}
\frac{
(\lambda_{_{K}}^{\scriptscriptstyle(+)} - \lambda_{j}^{\scriptscriptstyle(+)})
(\lambda_{_{K}}^{\scriptscriptstyle(-)} - \lambda_{j}^{\scriptscriptstyle(-)})
}
{(\lambda_{_{K}}^{\scriptscriptstyle(+)} - \lambda_{j}^{\scriptscriptstyle(-)})
(\lambda_{_{K}}^{\scriptscriptstyle(-)} - \lambda_{j}^{\scriptscriptstyle(+)})
}
=
\left\{
\begin{array}{l}
(1) ~\displaystyle{
\left|
\frac{
\lambda_{_{K}} - \lambda_{j}}{\lambda_{_{K}} - \overline{\lambda}_{j}}
\right|^2
} \in \mathbb{R}^{+} ~~~~~~\mbox{on}~~\mathbb{M}^{3,1}, \mathbb{U}^{2,2}_{1}
\smallskip \\
(2)~\displaystyle{
\left|
\frac{
\lambda_{_{K}} - \lambda_{j}}{\lambda_{_{K}}\overline{\lambda}_{j}-1}
\right|^2
} \in \mathbb{R}^{+} ~~~~~\!\mbox{on}~~\mathbb{U}^{2,2}_{2}
\smallskip \\
(3)~\displaystyle{
	-\left|
	\frac{
		\lambda_{_{K}} - \lambda_{j}}{\lambda_{_{K}}\overline{\lambda}_{j}+1}
	\right|^2
} \in \mathbb{R}^{-} ~~\mbox{on}~~\mathbb{E}^{4}
\end{array}
\right.
\end{eqnarray}
implies that the phase shift factor $\delta_{_{K}}$ is also real-valued on each real space, except in the cases of even soliton numbers $N=2n$, $n\in\mathbb{N}$, on $\mathbb{E}^{4}$, where $\delta_{_{K}}$ becomes ill-defined. 
The failure of $\delta_{_{K}}$ to be well-defined on $\mathbb{E}^{4}$ can be solved by adjusting the $N$-soliton ansatz slightly in \eqref{Soliton ansatz_rho} as
\begin{eqnarray}
\label{Soliton ansatz_rho_redefine}
\rho
:=
\left(
\begin{array}{cccc}
	\widetilde{E}_{1}^{\scriptscriptstyle(+)} & \widetilde{E}_{2}^{\scriptscriptstyle(+)} & \cdots & \widetilde{E}_{_{N}}^{\scriptscriptstyle(+)}
	\smallskip \\
	\xi_{_{[N]}}\widetilde{E}_{1}^{\scriptscriptstyle(-)} & 	\xi_{_{[N]}}\widetilde{E}_{2}^{\scriptscriptstyle(-)} & \cdots & 	\xi_{_{[N]}}\widetilde{E}_{_{N}}^{\scriptscriptstyle(-)}
\end{array}
\right), 
\end{eqnarray}
where
\begin{eqnarray}
\label{xi}
\xi_{_{[N]}}:=
\left\{
\begin{array}{l}
-1, ~~\mbox{if $N$ is even on $\mathbb{E}^{4}$}
\\
1, ~~\mbox{otherwise}
\end{array}.	
\right.	
\end{eqnarray}
Then, we impose the condition $\widetilde{E}^{(\pm)}=\overline{E}^{(\pm)}$, which guarantees that the action density is real-valued.

Through the WZW$_4$ model, we demonstrate that the asymptotic NL$\sigma$M action density associated with $J_{_{[N]}}$ retains the characteristic features of a 1-soliton solution $J_{[1]}$ up to a explicit phase shift factor, indicating that the individual solitons preserve their identities after collisions. This is a hallmark of the particle-like behavior of solitons. Therefore, $J_{_{[N]}}$ indeed satisfies the defining criteria for a standard $N$-soliton solution, in the same sense as KP solitons.

\subsection{Comparison with Quasi-Wronskian $N$-Soliton Solutions}
The quasi-Wronskian solution \cite{NGO-2000, GHHN-2020} to the ASDYM equation for $\mathrm{G}=\mathrm{GL}(n, \mathbb{C})$, denoted here by $\widetilde{J}_{_{[N]}}$ to distinguish it from the quasi-Grammian solution $J_{_{[N]}}$, is given by
\begin{eqnarray}
	\widetilde{J}_{_{[N]}}
	=
	\left|
	\begin{array}{cccc}
		\theta_{1} & \cdots & \theta_{_{N}} &  I_{n \times n}
		\\
		\theta_{1}\Lambda_{1} & \cdots & \theta_{_{N}}\Lambda_{_{N}} & O_{n \times n}
		\\
		\vdots & \ddots & \vdots & \vdots
		\\
		\theta_{1}\Lambda_{1}^{_{N}} & \cdots & \theta_{_{N}}\Lambda_{_{N}}^{_{N}} &  \fbox{$O_{n \times n}$}
	\end{array}
	\right|
\end{eqnarray}
provided that $(\theta_{j}, \Lambda_{j}), 1 \leq j \leq N$ satisfy the first equation of the linear system \eqref{Linear systems of ASDYM}, where $\theta_{j} \in \mathbb{C}_{n \times n}(z, \widetilde{z}, w, \widetilde{w})$ and $\Lambda_{j} \in \mathbb{C}_{n \times n}$. Note that the matrix sizes of $\theta_{j}$ and $\Lambda_{j}$ here are quite different from those appearing in the quasi-Grammian solutions in Theorem 4.1. 

One can also observe that the quasi-Grammian forms provide a more efficient representation of solutions to the ASDYM equation than the quasi-Wronskian forms. More precisely, for $\mbox{G}=\mbox{GL}(n, \mathbb{C})$, each entry of the $N$-th order quasi-Wronskian involves an $(nN+1)\times(nN+1)$ quasideterminant, whereas each entry of the $N$-th order quasi-Grammian involves only an $(N+1)\times(N+1)$ quasideterminant (cf. \eqref{Quasideterminant_defn_expansion}). 
This naturally raises the important question of whether the two representations are intrinsically equivalent for the $N$-soliton solutions of the ASDYM equation.
In the following, we will show that, at the asymptotic level, the two representations indeed describe the same class of $N$-soliton solutions.

For physical applications, we restrict ourselves to the case $\mbox{G}=\mbox{U}(2)$. To facilitate a clearer comparison with the quasi-Grammian $N$-soliton solutions, we consider the quasi-Wronskian soliton ansatz first introduced in \cite{GHHN-2020, HH-2020}, with the notation adjusted as follows:
\begin{eqnarray}
\theta_{j}
:=
\left(
\begin{array}{cc}
E_{j}^{\scriptscriptstyle(+)}  &  -\xi_{_{[N]}}\overline{E}_{j}^{\scriptscriptstyle(-)}
\smallskip \\
E_{j}^{\scriptscriptstyle(-)}  & 
\overline{E}_{j}^{\scriptscriptstyle(+)}
\end{array}
\right)  
, ~~ 
\Lambda_{j}
	:=
	\left(
	\begin{array}{cc}
		\lambda_{j}^{\scriptscriptstyle (+)}  &  0 
		\\
		0 &  \lambda_{j}^{\scriptscriptstyle (-)}
	\end{array}
	\right),
\end{eqnarray}
where $E_{j}^{(\pm)}$, $\lambda_{j}^{(\pm)}$ are defined in \eqref{E_j-pm} and \eqref{lambda_pm} respectively, $\xi_{_{[N]}}$ is defined in \eqref{xi}. 
According to the discussions in \cite{Huang-2021, HH-2022}, together with some direct computations, the following asymptotic form of the quasi-Wronskian $N$-soliton solution associated with the asymptotic regions determined by \eqref{epsilon_pm} is obtained:
\begin{eqnarray}
&&\!\!\!\!\widetilde{J}_{_{[N]}} \xrightarrow[]{r \rightarrow \infty} \mathcal{\widetilde{J}}_{_{[N]}}^{\scriptscriptstyle (K)}   \nonumber \\
&\!\!\!\!\!\!=\!\!\!\!\!\!&
\frac{
	\!\!\!\!\!\!\left(
	\begin{array}{cc}
		\widetilde{\gamma}_{_{KK}}
		p_{_K}\widetilde{p}_{_K}
		|E_{_{K}}^{\scriptscriptstyle(+)}|^2
		+
		\gamma_{_{KK}}q_{_K}\widetilde{q}_{_K}|E_{_{K}}^{\scriptscriptstyle(-)}|^2
		&
		-p_{_K}\widetilde{q}_{_K}
		E_{_{K}}^{\scriptscriptstyle(+)}\overline{E}_{_{K}}^{\scriptscriptstyle(-)}
		\smallskip \\
		-\widetilde{p}_{_K}q_{_K}
		\overline{E}_{_{K}}^{\scriptscriptstyle(+)}E_{_{K}}^{\scriptscriptstyle(-)} 
		&
		\gamma_{_{KK}}
		p_{_K}\widetilde{p}_{_K}
		|E_{_{K}}^{\scriptscriptstyle(+)}|^2
		+\widetilde{\gamma}_{_{KK}}
		q_{_K}\widetilde{q}_{_K}
		|E_{_{K}}^{\scriptscriptstyle(-)}|^2			
	\end{array}	
	\right)
}
{
	\gamma_{_{KK}}(p_{_K}\widetilde{p}_{_K}|E_{_{K}}^{\scriptscriptstyle(+)}|^2
	+
	q_{_K}\widetilde{q}_{_K}|E_{_{K}}^{\scriptscriptstyle(-)}|^2
)}\widetilde{C}^{\scriptscriptstyle(K)},  ~~~~~~~   
\end{eqnarray}
where $\gamma_{_{KK}}, \widetilde{\gamma}_{_{KK}}, p_{_{K}}, \widetilde{p}_{_{K}}, q_{_{K}}, \widetilde{q}_{_{K}}$ are defined exactly as in \eqref{J_asymptotic_exact}, and
\begin{eqnarray}
	\widetilde{C}^{\scriptscriptstyle(K)}
	:=
	-\lambda_{_{K}}^{\scriptscriptstyle(-)}\mbox{diag}
	\left( 
	\displaystyle{\prod_{j=1, j \neq K}^{N}}\!\!\lambda_{j}^{(+)}, 
	\displaystyle{\prod_{j=1, j \neq K}^{N}}\!\!\lambda_{j}^{(-)}
	\right)
\end{eqnarray}
is a constant matrix.
We therefore conclude that the $N$-soliton asymptotic forms
$\mathcal{J}_{_{[N]}}^{\scriptscriptstyle (K)}$
and
$\widetilde{\mathcal{J}}_{_{[N]}}^{\scriptscriptstyle (K)}$
coincide up to a slight discrepancy in the constant matrix factor, which does not affect the solitonic behavior (cf. \eqref{Trace formula_C}). 
The particle-like nature of solitons shows that the quasi-Wronskian and quasi-Grammian representations describe the same N-soliton solution.

\subsection{Quasi-Grammian Exact Multi-Soliton Solutions}
Although the quasi-Wronskian and quasi-Grammian constructions characterize the same class of $N$-soliton solutions, the latter requires substantially less computation for a fixed soliton number. Therefore, we present the exact multi-soliton solutions obtained from the quasi-Grammian construction in this subsection.
The $N$-soliton solution $J_{_{[N]}}$ here is redefined by \eqref{Soliton ansatz_theta} and \eqref{Soliton ansatz_rho_redefine}.

Using \eqref{J_CM_components} and \eqref{Quasideterminant_defn_expansion}, $J_{_{[N]}}$ can be expressed as a quotient of determinants, namely,
\begin{eqnarray}
J_{_{[N]}}:=
\frac{1}{\Delta_{_{[N]}}}
\left(
\begin{array}{cc}
\Delta_{_{[N]}}^{(11)} & \Delta_{_{[N]}}^{(12)}
\smallskip \\
\Delta_{_{[N]}}^{(21)} & \Delta_{_{[N]}}^{(22)}
\end{array}
\right).
\end{eqnarray}
Note that one can verify the following symmetric relations
\begin{eqnarray}
(\Delta_{_{[N]}}^{(21)}, \Delta_{_{[N]}}^{(22)})
:=
(\Delta_{_{[N]}}^{(12)}, \Delta_{_{[N]}}^{(11)})
\Big|_{
	(E_{j}^{(\pm)}, \widetilde{E}_{j}^{(\pm)}) \longrightarrow 
	(E_{j}^{(\mp)}, ~\!\xi_{_{[N]}}^{\pm 1}\widetilde{E}_{j}^{(\mp)})}
\end{eqnarray}
and
\begin{eqnarray}
\Delta_{_{[N]}}
=\Delta_{_{[N]}}^{(11)}\Big|_{\widetilde{\gamma}_{jk} \rightarrow \gamma_{jk}}
=\Delta_{_{[N]}}^{(22)}\Big|_{\widetilde{\gamma}_{jk} \rightarrow \gamma_{jk}},
\end{eqnarray}
where $\gamma_{jk}$ and $\widetilde{\gamma}_{jk}$ are defined in \eqref{gamma_jk} and \eqref{gamma-tilde_jk} respectively.
Therefore, all information of the $N$-soliton solution $J_{_{[N]}}$ is completely determined by $\Delta_{_{[N]}}^{(11)}$ and $\Delta_{_{[N]}}^{(12)}$.
For simplicity, we further define 
\begin{eqnarray}
\Phi_{j}^{(\pm)}
:=
E_{j}^{(\pm)}\widetilde{E}_{j}^{(\mp)}
\end{eqnarray}
Under the condition $\widetilde{E}_{j}^{(\pm)}=\overline{E}_{j}^{(\pm)}$, for $N=2$, we have
\begin{eqnarray}
\Delta_{[2]}^{(11)}
&\!\!\!\!=\!\!\!\!&
\begin{array}{l}
~~~~~\!\left|
\begin{array}{cc}
\widetilde{\gamma}_{11}	 &  \widetilde{\gamma}_{12}
\\
\widetilde{\gamma}_{21}  &  \widetilde{\gamma}_{22}
\end{array}
\right|
|E_{1}^{\scriptscriptstyle(+)}E_{2}^{\scriptscriptstyle(+)}|^2 
+
\xi_{[2]}^2\!\left|
\begin{array}{cc}
\gamma_{11}	 &  \gamma_{12}
\\
\gamma_{21}  &  \gamma_{22}
\end{array}
\right|
|E_{1}^{\scriptscriptstyle(-)}E_{2}^{\scriptscriptstyle(-)}|^2 
\medskip \\
\!\!+
\xi_{[2]}\!\left|
\begin{array}{cc}
\widetilde{\gamma}_{11} & 0 \\ 
0 &\gamma_{22}		
\end{array}
\right|
|E_{1}^{\scriptscriptstyle(+)}E_{2}^{\scriptscriptstyle(-)}|^2 
+
\xi_{[2]}\!\left|
\begin{array}{cc}
\gamma_{11} & 0 \\ 
0  & \widetilde{\gamma}_{22}		
\end{array}
\right|
|E_{1}^{\scriptscriptstyle(-)}E_{2}^{\scriptscriptstyle(+)}|^2 
\medskip \\
\!\!+ \xi_{[2]}\!\left|
\begin{array}{cc}
0 & \gamma_{12}
\\
\widetilde{\gamma}_{21} & 0
\end{array}
\right|
\Phi_{1}^{\scriptscriptstyle(+)}\Phi_{2}^{\scriptscriptstyle(-)}
+
\xi_{[2]}\!\left|
\begin{array}{cc}
0 & \widetilde{\gamma}_{12}
\\
\gamma_{21} & 0
\end{array}
\right|
\Phi_{1}^{\scriptscriptstyle(-)}\Phi_{2}^{\scriptscriptstyle(+)}
\end{array},
\end{eqnarray}
and
\begin{eqnarray}
\Delta_{[2]}^{(12)}
&\!\!\!\!=\!\!\!\!&
\begin{array}{l}
\!\!-\left[
\xi_{[2]}\left|
\begin{array}{cc}
\gamma_{11} & \gamma_{12} \\1 & 1
\end{array}	
\right|
\left| E_{1}^{\scriptscriptstyle(+)} \right|^2
+
\xi_{[2]}^2\left|
\begin{array}{cc}
\gamma_{11} & 1  \\ \gamma_{21} & 1
\end{array}	
\right|
\left| E_{1}^{\scriptscriptstyle(-)} \right|^2
\right]\Phi_{2}^{\scriptscriptstyle(+)}
\medskip \\
\!\!-\left[
\xi_{[2]}\left|
\begin{array}{cc}
1 & 1 \\ \gamma_{21} & \gamma_{22} 
\end{array}	
\right|
\left| E_{2}^{\scriptscriptstyle(+)} \right|^2
+
\xi_{[2]}^2\left|
\begin{array}{cc}
1 &	\gamma_{12} \\ 1 & \gamma_{22}
\end{array}	
\right|
\left| E_{2}^{\scriptscriptstyle(-)} \right|^2
\right]\Phi_{1}^{\scriptscriptstyle(+)}
\end{array}.
\end{eqnarray}
For $N=3$, we have
\begin{eqnarray}
\Delta_{[3]}^{(11)}
=
\begin{array}{l}
	~~~~~~~\left|
	\begin{array}{ccc}
		\widetilde{\gamma}_{11} & \widetilde{\gamma}_{12} & \widetilde{\gamma}_{13} 
		\\
		\widetilde{\gamma}_{21} & \widetilde{\gamma}_{22} & \widetilde{\gamma}_{23} 
		\\
		\widetilde{\gamma}_{31} & \widetilde{\gamma}_{32} & \widetilde{\gamma}_{33}
	\end{array}
	\right|
	|E_{1}^{\scriptscriptstyle(+)}E_{2}^{\scriptscriptstyle(+)}E_{3}^{\scriptscriptstyle(+)}|^2
	+
\left|
\begin{array}{ccc}
	\gamma_{11} & \gamma_{12} & \gamma_{13} 
	\\
	\gamma_{21} & \gamma_{22} & \gamma_{23} 
	\\
	\gamma_{31} & \gamma_{32} & \gamma_{33}
\end{array}
\right|	
	|E_{1}^{\scriptscriptstyle(-)}E_{2}^{\scriptscriptstyle(-)}E_{3}^{\scriptscriptstyle(-)}|^2
\medskip \\
+\!\left[
\begin{array}{l}	
	~~\!\left|
	\begin{array}{ccc}
		\widetilde{\gamma}_{11} & \widetilde{\gamma}_{12} & 0 
		\\
		\widetilde{\gamma}_{21} & \widetilde{\gamma}_{22} & 0 
		\\
		0 & 0 & \gamma_{33}
	\end{array}
	\right|
	|E_{1}^{\scriptscriptstyle(+)}E_{2}^{\scriptscriptstyle(+)}E_{3}^{\scriptscriptstyle(-)}|^2
+\left|
\begin{array}{ccc}
	\gamma_{11} & \gamma_{12} & 0 
	\\
	\gamma_{21} & \gamma_{22} & 0 
	\\
	0 & 0 & \widetilde{\gamma}_{33}
\end{array}
\right|	
	|E_{1}^{\scriptscriptstyle(-)}E_{2}^{\scriptscriptstyle(-)}E_{3}^{\scriptscriptstyle(+)}|^2
	\medskip \\
	\!\!+\left|
	\begin{array}{ccc}
		\widetilde{\gamma}_{11} & 0 & \widetilde{\gamma}_{13} 
		\\
		0 & \gamma_{22} & 0 
		\\
		\widetilde{\gamma}_{31} & 0 & \widetilde{\gamma}_{33}
	\end{array}
\right|
	|E_{1}^{\scriptscriptstyle(+)}E_{2}^{\scriptscriptstyle(-)}E_{3}^{\scriptscriptstyle(+)}|^2
+\left|
\begin{array}{ccc}
	\gamma_{11} & 0 & \gamma_{13} 
	\\
	0 & \widetilde{\gamma}_{22} & 0 
	\\
	\gamma_{31} & 0 & \gamma_{33}
\end{array}
\right|
	|E_{1}^{\scriptscriptstyle(-)}E_{2}^{\scriptscriptstyle(+)}E_{3}^{\scriptscriptstyle(-)}|^2
	\medskip \\
	\!\!+\left|
	\begin{array}{ccc}
		\gamma_{11} & 0 & 0
		\\
		0 & \widetilde{\gamma}_{22} & \widetilde{\gamma}_{23} 
		\\
		0 & \widetilde{\gamma}_{32} & \widetilde{\gamma}_{33}
	\end{array}
	\right|
	|E_{1}^{\scriptscriptstyle(-)}E_{2}^{\scriptscriptstyle(+)}E_{3}^{\scriptscriptstyle(+)}|^2
+\left|
\begin{array}{ccc}
	\widetilde{\gamma}_{11} & 0 & 0
	\\
	0 & \gamma_{22} & \gamma_{23} 
	\\
	0 & \gamma_{32} & \gamma_{33}
\end{array}
\right|
	|E_{1}^{\scriptscriptstyle(+)}E_{2}^{\scriptscriptstyle(-)}E_{3}^{\scriptscriptstyle(-)}|^2
\end{array}	
\!\!\!\right] 
\medskip \\
+\!\left[
\begin{array}{l}	
	~~\!\left|
	\begin{array}{ccc}
		\widetilde{\gamma}_{11} & \widetilde{\gamma}_{12} & 0 
		\\
		0 & 0 & \gamma_{23} 
		\\
		\widetilde{\gamma}_{31} & \widetilde{\gamma}_{32} & 0
	\end{array}
	\right|
	|E_{1}^{\scriptscriptstyle(+)}|^2\Phi_{2}^{\scriptscriptstyle(+)}\Phi_{3}^{\scriptscriptstyle(-)}
	+\left|
	\begin{array}{ccc}
		\gamma_{11} & \gamma_{12} & 0 
		\\
		0 & 0 & \widetilde{\gamma}_{23} 
		\\
		\gamma_{31} & \gamma_{32} & 0
	\end{array}
	\right|
	|E_{1}^{\scriptscriptstyle(-)}|^2\Phi_{2}^{\scriptscriptstyle(-)}\Phi_{3}^{\scriptscriptstyle(+)}
	\medskip \\
	\!\!+\left|
	\begin{array}{ccc}
		\widetilde{\gamma}_{11} & 0 & \widetilde{\gamma}_{13} 
		\\
		\widetilde{\gamma}_{21} & 0 & \widetilde{\gamma}_{23} 
		\\
		0 & \gamma_{32} & 0
	\end{array}
	\right|
	|E_{1}^{\scriptscriptstyle(+)}|^2\Phi_{2}^{\scriptscriptstyle(-)}\Phi_{3}^{\scriptscriptstyle(+)}
	+\left|
	\begin{array}{ccc}
		\gamma_{11} & 0 & \gamma_{13} 
		\\
		\gamma_{21} & 0 & \gamma_{23} 
		\\
		0 & \widetilde{\gamma}_{32} & 0
	\end{array}
	\right|
	|E_{1}^{\scriptscriptstyle(-)}|^2\Phi_{2}^{\scriptscriptstyle(+)}\Phi_{3}^{\scriptscriptstyle(-)}
\end{array}
\!\!\!\right]	
\medskip \\
+\!\left[
\begin{array}{l}	
	~~\!\left|
	\begin{array}{ccc}
		0 & 0 & \gamma_{13} 
		\\
		\widetilde{\gamma}_{21} & \widetilde{\gamma}_{22} & 0 
		\\
		\widetilde{\gamma}_{31} & \widetilde{\gamma}_{32} & 0
	\end{array}
	\right|
	|E_{2}^{\scriptscriptstyle(+)}|^2\Phi_{1}^{\scriptscriptstyle(+)}\Phi_{3}^{\scriptscriptstyle(-)}
	+\left|
	\begin{array}{ccc}
		0 & 0 & \widetilde{\gamma}_{13} 
		\\
		\gamma_{21} & \gamma_{22} & 0 
		\\
		\gamma_{31} & \gamma_{32} & 0
	\end{array}
	\right|
	|E_{2}^{\scriptscriptstyle(-)}|^2\Phi_{1}^{\scriptscriptstyle(-)}\Phi_{3}^{\scriptscriptstyle(+)}
	\medskip\\
	\!\!+\left|
	\begin{array}{ccc}
		0 & \widetilde{\gamma}_{12} & \widetilde{\gamma}_{13} 
		\\
		0 & \widetilde{\gamma}_{22} & \widetilde{\gamma}_{23} 
		\\
		\gamma_{31} & 0 & 0
	\end{array}
	\right|
	|E_{2}^{\scriptscriptstyle(+)}|^2\Phi_{1}^{\scriptscriptstyle(-)}\Phi_{3}^{\scriptscriptstyle(+)}
	+\left|
	\begin{array}{ccc}
		0 & \gamma_{12} & \gamma_{13} 
		\\
		0 & \gamma_{22} & \gamma_{23} 
		\\
		\widetilde{\gamma}_{31} & 0 & 0
	\end{array}
	\right|
	|E_{2}^{\scriptscriptstyle(-)}|^2\Phi_{1}^{\scriptscriptstyle(+)}\Phi_{3}^{\scriptscriptstyle(-)}
\end{array}
\!\!\!\right]	
\medskip \\ 	
+\left[
\begin{array}{l}	
	~~\!\left|
	\begin{array}{ccc}
		0 & \gamma_{12} & 0 
		\\
		\widetilde{\gamma}_{21} & 0 & \widetilde{\gamma}_{23} 
		\\
		\widetilde{\gamma}_{31} & 0 & \widetilde{\gamma}_{33}
	\end{array}
	\right|
	|E_{3}^{\scriptscriptstyle(+)}|^2\Phi_{1}^{\scriptscriptstyle(+)}\Phi_{2}^{\scriptscriptstyle(-)}
+\left|
\begin{array}{ccc}
	0 & \widetilde{\gamma}_{12} & 0 
	\\
	\gamma_{21} & 0 & \gamma_{23} 
	\\
	\gamma_{31} & 0 & \gamma_{33}
\end{array}
\right|
	|E_{3}^{\scriptscriptstyle(-)}|^2\Phi_{1}^{\scriptscriptstyle(-)}\Phi_{2}^{\scriptscriptstyle(+)}
	\medskip \\
	\!\!+\left|
	\begin{array}{ccc}
		0 & \widetilde{\gamma}_{12} & \widetilde{\gamma}_{13} 
		\\
		\gamma_{21} & 0 & 0 
		\\
		0 & \widetilde{\gamma}_{32} & \widetilde{\gamma}_{33}
	\end{array}
	\right|
	|E_{3}^{\scriptscriptstyle(+)}|^2\Phi_{1}^{\scriptscriptstyle(-)}\Phi_{2}^{\scriptscriptstyle(+)}
+\left|
\begin{array}{ccc}
	0 & \gamma_{12} & \gamma_{13} 
	\\
	\widetilde{\gamma}_{21} & 0 & 0 
	\\
	0 & \gamma_{32} & \gamma_{33}
\end{array}
\right|
	|E_{3}^{\scriptscriptstyle(-)}|^2\Phi_{1}^{\scriptscriptstyle(+)}\Phi_{2}^{\scriptscriptstyle(-)}
\end{array}
\!\!\!\right]
\end{array},
\end{eqnarray}
and
\begin{eqnarray}
\Delta_{[3]}^{(12)}
=\begin{array}{l}
-\left[
\left| 
\begin{array}{ccc}
\gamma_{11} & \gamma_{12} & \gamma_{13} \\
\gamma_{21} & \gamma_{22} & \gamma_{23} \\
1 & 1 & 1
\end{array}
\right||E_{1}^{\scriptscriptstyle(+)}E_{2}^{\scriptscriptstyle(+)}|^2
+
\left| 
\begin{array}{ccc}
	\gamma_{11} & \gamma_{12} & 1 \\
	\gamma_{21} & \gamma_{22} & 1 \\
	\gamma_{31} & \gamma_{32} & 1
\end{array}
\right||E_{1}^{\scriptscriptstyle(-)}E_{2}^{\scriptscriptstyle(-)}|^2
\right]\!\!\Phi_{3}^{\scriptscriptstyle(+)}
\medskip \\
-\left[
\left| 
\begin{array}{ccc}
	\gamma_{11} & \gamma_{12} & \gamma_{13} \\
	1 & 1 & 1  \\
	\gamma_{31} & \gamma_{32} & \gamma_{33} \\
\end{array}
\right||E_{1}^{\scriptscriptstyle(+)}E_{3}^{\scriptscriptstyle(+)}|^2
+\left|
\begin{array}{ccc}
	\gamma_{11} & 1 & \gamma_{13}  \\
	\gamma_{21} & 1 & \gamma_{23}  \\
	\gamma_{31} & 1 & \gamma_{33} 
\end{array}
\right||E_{1}^{\scriptscriptstyle(-)}E_{3}^{\scriptscriptstyle(-)}|^2
\right]\!\!\Phi_{2}^{\scriptscriptstyle(+)}
\medskip \\
-\left[
\left|
\begin{array}{ccc}
	1 & 1 & 1  \\
	\gamma_{21} & \gamma_{22} & \gamma_{23} \\
	\gamma_{31} & \gamma_{32} & \gamma_{33} \\
\end{array}
\right||E_{2}^{\scriptscriptstyle(+)}E_{3}^{\scriptscriptstyle(+)}|^2
+
\left| 
\begin{array}{ccc}
	1 & \gamma_{12} & \gamma_{13}  \\
	1 & \gamma_{22} & \gamma_{23}  \\
	1 & \gamma_{32} & \gamma_{33} 
\end{array}
\right||E_{2}^{\scriptscriptstyle(-)}E_{3}^{\scriptscriptstyle(-)}|^2
\right]\!\!\Phi_{1}^{\scriptscriptstyle(+)} 
\end{array}   \nonumber 
\end{eqnarray}
\begin{eqnarray}
\begin{array}{l}
-\left[
\left| 
\begin{array}{ccc}
	\gamma_{11} & 1 & \gamma_{13} \\
	1 & \gamma_{22} & 1 \\
	1 & \gamma_{32} & 1
\end{array}
\right||E_{1}^{\scriptscriptstyle(+)}E_{2}^{\scriptscriptstyle(-)}|^2
+
\left| 
\begin{array}{ccc}
	\gamma_{11} & 1 & 1 \\
	1 & \gamma_{22} & \gamma_{23} \\
	\gamma_{31} & 1 & 1
\end{array}
\right||E_{1}^{\scriptscriptstyle(-)}E_{2}^{\scriptscriptstyle(+)}|^2
\right]\!\!\Phi_{3}^{\scriptscriptstyle(+)}
\medskip \\
-\left[
\left| 
\begin{array}{ccc}
	\gamma_{11} & \gamma_{12} & 1 \\
	1 & 1 & \gamma_{23} \\
	1 & 1 & \gamma_{33}
\end{array}
\right||E_{1}^{\scriptscriptstyle(+)}E_{3}^{\scriptscriptstyle(-)}|^2
+
\left| 
\begin{array}{ccc}
	\gamma_{11} & 1 & 1 \\
	\gamma_{21} & 1 & 1 \\
	1 & \gamma_{32} & \gamma_{33}
\end{array}
\right||E_{1}^{\scriptscriptstyle(-)}E_{3}^{\scriptscriptstyle(+)}|^2
\right]\!\!\Phi_{2}^{\scriptscriptstyle(+)}
\medskip \\
-\left[
\left| 
\begin{array}{ccc}
	1 & 1 & \gamma_{13} \\
	\gamma_{21} & \gamma_{22} & 1 \\
	1 & 1 & \gamma_{33}
\end{array}
\right||E_{2}^{\scriptscriptstyle(+)}E_{3}^{\scriptscriptstyle(-)}|^2
+
\left| 
\begin{array}{ccc}
	1 & \gamma_{12} & 1 \\
	1 & \gamma_{22} & 1 \\
	\gamma_{31} & 1 & \gamma_{33}
\end{array}
\right||E_{2}^{\scriptscriptstyle(-)}E_{3}^{\scriptscriptstyle(+)}|^2
\right]\!\!\Phi_{1}^{\scriptscriptstyle(+)}
\medskip \\
-\!\left[
\left| 
\begin{array}{ccc}
	\!\!1 & \!\!1 & \!\!\gamma_{13} \\
	\!\!1 & \!\!1 & \!\!\gamma_{23} \\
	\!\!\gamma_{31} & \!\!\gamma_{32} & \!\!1
\end{array}
\right|\!\Phi_{1}^{\scriptscriptstyle(+)}\Phi_{2}^{\scriptscriptstyle(+)}\Phi_{3}^{\scriptscriptstyle(-)} 
\!\!+
\left| 
\begin{array}{ccc}
	\!\!1 & \!\!\gamma_{12} & \!\!1 \\
	\!\!\gamma_{21} & \!\!1 & \!\!\gamma_{23}  \\
	\!\!1 & \!\!\gamma_{32} & \!\!1
\end{array}
\!\!\right|\!\Phi_{1}^{\scriptscriptstyle(+)}\Phi_{2}^{\scriptscriptstyle(-)}\Phi_{3}^{\scriptscriptstyle(+)}
\!\!+
\left| 
\begin{array}{ccc}
	\!\!1 & \!\!\gamma_{12} & \!\!\gamma_{13} \\
    \!\!\gamma_{21} & \!\!1 & \!\!1 \\
	\!\!\gamma_{31} & \!\!1 & \!\!1
\end{array}
\!\!\right|\!\Phi_{1}^{\scriptscriptstyle(-)}\Phi_{2}^{\scriptscriptstyle(+)}\Phi_{3}^{\scriptscriptstyle(+)}
\right].
\end{array}
\end{eqnarray}

For $N=4$, the solution still exhibits striking symmetry and elegance, despite its lengthy expression. Its full form is presented in Appendix E. The quasi-Grammian representation offers a significant computational advantage in constructing exact soliton solutions. For a 4-soliton solution with $\mathrm{G}=\mathrm{GL}(2, \mathbb{C})$, the quasi-Wronskian representation requires the expansion of ninth-order determinants, whereas the quasi-Grammian representation involves only fifth-order determinants, thereby substantially reducing the computational complexity. On the other hand, the coefficients of the $N$-soliton solutions admit a highly symmetric matrix representation in terms of $\gamma_{jk}$ and $\widetilde{\gamma}_{jk}$, reflecting the elegant structure of the Cauchy matrix $\Omega$. This rich combinatorial symmetry provides an underlying regularity, thereby rendering the otherwise intricate computations systematic.

\section{Concluding Remarks}

In view of the widely accepted equivalence between the Wronskian and Grammian representations of soliton solutions for a broad class of integrable equations, particularly those related to the KP equation \cite{CLM-2010}, 
it is natural to conjecture that the quasi-Grammian and quasi-Wronskian constructions describe the same class of ASDYM multi-soliton solutions. While the result itself may be anticipated, establishing it requires a deeper understanding of the intricate structures of the quasideterminants. Taking advantage of the particle-like nature of solitons, whereby they preserve their identities under multi-soliton interactions, we adopted an indirect approach to compare the asymptotic behavior of the two representations. 
We proved that, in the WZW$_4$ model, the quasi-Grammian $N$-soliton solution $J_{_{[N]}}$ and the quasi-Wronskian solution $\widetilde{J}_{_{[N]}}$ are asymptotically equivalent and therefore represent the same class of ASDYM $N$-soliton solutions. 

Although the Wronskian technique is well established \cite{FN-1983, Nimmo-1989, Zhang-2019}, and a standard approach to the asymptotic analysis was developed in \cite{MS-1991} and subsequently extended to the quasi-Wronskian setting in \cite{HH-2022}, a corresponding asymptotic framework for Grammian solutions has, to the best of our knowledge, not yet been developed. The asymptotic analysis presented here establishes a systematic framework for quasi-Grammian solutions and is expected to be applicable, in the commutative limit, to the asymptotic analysis of Grammian solutions as well.
Furthermore, our proofs can also be carried out by working directly with ratios of determinants. However, such an approach would be technically much more involved. By exploiting the algebraic properties of quasideterminants, the proofs are significantly simplified, and the resulting framework could be applied to other integrable equations \cite{GN-2007, GNS-2008, DM-2008, DM-2013, NY-2015, Yilmaz-2022, LY-2024} and further extended to $\mathrm{U}(n)$ solutions for $n\ge 3$ \cite{Huang-2022}.
On the other hand, for quasideterminants of the same order, the quasi-Wronskian representation involves matrix-valued entries of substantially larger dimensions than the quasi-Grammian representation. Therefore, once their equivalence is established, the more compact quasi-Grammian formulation can be employed to represent the same class of solutions more efficiently.

The core ingredient of Sato theory is the so-called $\tau$-function \cite{Sato-1981,CLM-2010}, which admits two equivalent representations for soliton solutions, namely the Wronskian and Grammian representations. Although the ASDYM equation does not fall within the framework of Sato theory, our results further suggest the existence of higher-dimensional $\tau$-function-like objects underlying the quasi-Wronskian \cite{NGO-2000,GHHN-2020,HH-2022} and quasi-Grammian \cite{NGO-2000,LQYZ-2022,LQZ-2023,Ohta-2024,LHHZ-2025} representations, analogous to the $\tau$-functions in Sato theory and potentially providing a framework for higher-dimensional soliton classification \cite{KW-2014,Kodama-2017-18}. This observation raises the possibility of a more fundamental higher-dimensional integrable framework underlying both Sato theory and the ASDYM equation in the quasideterminant setting.

From the viewpoint of twistor theory, the ASDYM equation admits a holomorphic gauge-theoretic description \cite{MW-1996, Witten-2004}. Recent developments on holomorphic and four-dimensional Chern--Simons theories \cite{Witten-2004, CWY-2018,CY-2019,Bittleston-2022,BS-2023,HTC-2022,CCHLT-2024}, together with related studies on open $N=2$ string theory \cite{OV-1991, ZhangX-2025}, have provided geometric and gauge-theoretic frameworks in which ASDYM theory and lower-dimensional integrable systems arise from higher-dimensional holomorphic structures. These developments therefore provide further support for the existence of a common higher-dimensional integrable structure, pointing toward a unified understanding of integrable systems, higher-dimensional soliton classification, and the origin of integrability.

\subsection*{Acknowledgements} 
The author would like to thank Dr. Shangshuai Li for his valuable comments. This work was supported by the Iwanami Fujukai Foundation.


\begin{appendices}

\section{Verification of \eqref{Omega_derivative}}
Taking partial derivative on \eqref{Sylvester eq} and using \eqref{Linear systems of ASDYM}, one obtains
\begin{eqnarray}
	\Xi^{T}(\partial_{z}\Omega) - (\partial_{z}\Omega)\Lambda
	&\!\!\!\!=\!\!\!\!&(\partial_{z}\rho^{T})\theta + \rho^{T}(\partial_{z}\theta)
	\\
	&\!\!\!\!=\!\!\!\!&
	[\rho^{T}J_{0}(\partial_{z}J_{0}^{-1}) + \Xi^{T}(\partial_{\widetilde{w}}\rho^{T})]\theta
	+
	\rho^{T}[(\partial_{z}J_{0})J_{0}^{-1}\theta + (\partial_{\widetilde{w}}\theta)\Lambda]
	\\
	&\!\!\!\!=\!\!\!\!& 
	\label{Omega_z}
	\Xi^{T}(\partial_{\widetilde{w}}\rho^{T})\theta + \rho^{T}(\partial_{\widetilde{w}}\theta)\Lambda.
\end{eqnarray}
Similarly,
\begin{eqnarray}
	\label{Omega_w_tilde}
	\Xi^{T}(\partial_{\widetilde{w}}\Omega) - (\partial_{
		\widetilde{w}}\Omega)\Lambda
	=(\partial_{\widetilde{w}}\rho^{T})\theta + \rho^{T}(\partial_{\widetilde{w}}\theta).
\end{eqnarray}
Considering $\Xi^{T} \times \eqref{Omega_w_tilde} + \eqref{Omega_w_tilde} \times \Lambda$ and using \eqref{Omega_z}, one obtains
\begin{eqnarray}
	(\Xi^{T})^2(\partial_{\widetilde{w}}\Omega) - (\partial_{
		\widetilde{w}}\Omega)\Lambda^2
	&\!\!\!\!=\!\!\!\!&
	\Xi^{T}(\partial_{\widetilde{w}}\rho^{T})\theta + \Xi^{T}\rho^{T}(\partial_{\widetilde{w}}\theta)+
	(\partial_{\widetilde{w}}\rho^{T})\theta\Lambda + \rho^{T}(\partial_{\widetilde{w}}\theta)\Lambda
	\\
	&\!\!\!\!=\!\!\!\!&
	\Xi^{T}(\partial_{z}\Omega) - (\partial_{z}\Omega)\Lambda
	+
	\Xi^{T}\rho^{T}(\partial_{\widetilde{w}}\theta)+
	(\partial_{\widetilde{w}}\rho^{T})\theta\Lambda
\end{eqnarray}
which implies
\begin{eqnarray}
	\Xi^{T}[\partial_{z}\Omega-\Xi^{T}(\partial_{\widetilde{w}}\Omega) + \rho^{T}(\partial_{\widetilde{w}}\theta)]
	-
	[\partial_{z}\Omega -(\partial_{\widetilde{w}}\Omega)\Lambda -(\partial_{\widetilde{w}}\rho^{T})\theta]\Lambda =0.
\end{eqnarray}
Comparing with \eqref{Omega_w_tilde}, one can define
\begin{eqnarray}
C:=	
\partial_{z}\Omega-\Xi^{T}(\partial_{\widetilde{w}}\Omega) + \rho^{T}(\partial_{\widetilde{w}}\theta)
	=
	\partial_{z}\Omega -(\partial_{\widetilde{w}}\Omega)\Lambda -(\partial_{\widetilde{w}}\rho^{T})\theta 
\end{eqnarray} 
Note that $\Xi^{T}C - C\Lambda = 0$ admits only the trivial solution $C=0$, provided that $\Xi^{T}$ and $\Lambda$ have no common eigenvalues. 
Therefore, 
\begin{eqnarray}
	\partial_{z}\Omega
	=
	\Xi^{T}(\partial_{\widetilde{w}}\Omega)- \rho^{T}(\partial_{\widetilde{w}}\theta)
	=
	(\partial_{\widetilde{w}}\Omega)\Lambda +(\partial_{\widetilde{w}}\rho^{T})\theta.
\end{eqnarray}
Similarly,
\begin{eqnarray}
	\partial_{w}\Omega
	=
	\Xi^{T}(\partial_{\widetilde{z}}\Omega)- \rho^{T}(\partial_{\widetilde{z}}\theta)
	=
	(\partial_{\widetilde{z}}\Omega)\Lambda +(\partial_{\widetilde{z}}\rho^{T})\theta.
\end{eqnarray}
$\hfill\Box$\medskip \\

\section{Proof of Lemma 5.1 -- 5.4}
\subsection{Proof of Lemma 5.1}
By the permutation rule for quasideterminants (Proposition~3.1), it suffices to consider the case where the boxed entry is the $(n,n)$ entry. That is, it suffices to prove
\begin{eqnarray}
\label{Lem_5.1}	
\left|
\begin{array}{cccccc}
a_{11} & \cdots & a_{1,n} & 0 & \cdots & 0 
\\
\vdots & \ddots & \vdots  & \vdots & \ddots & \vdots
\\
a_{n, 1} & \cdots & \fbox{$a_{n,n}$} & 0 & \cdots & 0 
\smallskip \\
0 & \cdots & 0 & b_{11} & \cdots & b_{1,m}
\\
\vdots & \ddots & \vdots  & \vdots & \ddots & \vdots
\\
0 & \cdots & 0 & b_{m1} & \cdots & b_{m,m}
\end{array}	
\right|
=
\left|
\begin{array}{ccc}
	a_{11} & \cdots & a_{1,n} 
	\\
	\vdots & \ddots & \vdots 
	\\
	a_{n1} & \cdots & \fbox{$a_{n,n}$} 
\end{array}	
\right|.
\end{eqnarray}
For a fixed $m \in \mathbb{N}$ and $n=1$, 
\begin{eqnarray} 
  \left|
  \begin{array}{cccc}
  	\fbox{$a_{11}$} & 0 & \cdots & 0 
  	\smallskip \\
  	0  & b_{11} & \cdots & b_{1,m}
  	\\
  	\vdots  & \vdots & \ddots & \vdots
  	\\
  	0  & b_{m1} & \cdots & b_{m,m}
  \end{array}	
  \right|
 = a_{11} = |A|_{11}
\end{eqnarray}  
holds directly from  the definition \eqref{Quasideterminant_defn_concrete form}.
Suppose that the statement \eqref{Lem_5.1} still holds for $n=k$, 
then by the definition \eqref{Quasideterminant_defn_concrete form}, we have
\begin{eqnarray}
&&\!\!\!\!\left|
\begin{array}{ccccccc}
    a_{11} & \cdots & a_{1,k} & a_{1, k+1} & 0 & \cdots & 0 
	\\
	\vdots & \ddots & \vdots  & \vdots & \vdots & \ddots & \vdots
	\\
	a_{k,1} & \cdots & a_{k,k} & a_{k, k+1} & 0 & \cdots & 0  
	\smallskip \\
	a_{k+1,1} & \cdots & a_{k+1,k} & \fbox{$a_{k+1, k+1}$} & 0 & \cdots & 0 
	\smallskip \\
	0 & \cdots & 0 & 0 & b_{11} & \cdots & b_{1,m}
	\\
	\vdots & \ddots & \vdots & \vdots  & \vdots & \ddots & \vdots
	\\
	0 & \cdots & 0 & 0 & b_{m1} & \cdots & b_{m,m}
\end{array}	
\right|   \nonumber\\
&\!\!\!\!=\!\!\!\!&
a_{k+1, k+1}
-\left(
a_{k+1, 1}, \cdots, a_{k+1, k}, 0, \cdots, 0 
\right) X_{(k+m)\times (k+m)}
(a_{1, k+1}, \cdots, a_{k, k+1}, 0, \cdots, 0)^T ~~~~~~~~
\\
 &\!\!\!\!=\!\!\!\!& 
a_{k+1, k+1}
-\left(
a_{k+1, 1}, \cdots, a_{k+1, k} \right) 
\left(
\begin{array}{ccc}
x_{11} & \cdots & x_{1, k}
\\
\vdots & \ddots & \vdots 
\\
x_{k,1} & \cdots & x_{k, k}
\end{array}	
\right)
\left(
\begin{array}{c}
	a_{1, k+1} \\ \vdots \\ a_{k, k+1} 
\end{array}
\right). 
\end{eqnarray}
It now suffices to show that
\begin{eqnarray}
\label{X=A-inverse}	
	\left(
	\begin{array}{ccc}
		x_{11} & \cdots & x_{1, k}
		\\
		\vdots & \ddots & \vdots 
		\\
		x_{k,1} & \cdots & x_{k, k}
	\end{array}	
	\right)
	=
	\left(
	\begin{array}{ccc}
		a_{11} & \cdots & a_{1, k}
		\\
		\vdots & \ddots & \vdots 
		\\
		a_{k,1} & \cdots & a_{k, k}
	\end{array}	
	\right)^{-1}.
\end{eqnarray}
Once this is established, the proof is complete.
For $1 \leq i, j \leq k$, by the definition \eqref{Quasideterminat_inverse matrix} and \eqref{Lem_5.1} with $n=k$, we obtain
\begin{eqnarray}
x_{ij}
=
\left|
\begin{array}{cccccccc}
	a_{11} & \cdots & a_{1,i} & \cdots & a_{1,k} & 0 & \cdots & 0 
	\\
	\vdots & \ddots & \vdots & \ddots  & \vdots & \vdots & \ddots & \vdots
	\smallskip \\
	a_{j, 1} & \cdots & \fbox{$a_{j,i}$} & \cdots & a_{j, k} & 0 & \cdots & 0
	\smallskip \\
	\vdots & \ddots & \vdots & \ddots  & \vdots & \vdots & \ddots & \vdots
	\\
	a_{k,1} & \cdots & a_{k, i} & \cdots & a_{k,k} & 0 & \cdots & 0 
	\smallskip \\
	0 & \cdots & 0 & \cdots & 0  & b_{11} & \cdots & b_{1,m}
	\\
\vdots & \ddots & \vdots & \ddots  & \vdots & \vdots & \ddots & \vdots
	\\
	0 & \cdots & 0 & \cdots & 0 & b_{m1} & \cdots & b_{m,m}
\end{array}	
\right|^{-1}
=
\left|
\begin{array}{ccccc}
	a_{11} & \cdots & a_{1,i} & \cdots & a_{1,k} 
	\\
	\vdots & \ddots & \vdots & \ddots  & \vdots 
	\smallskip \\
	a_{j, 1} & \cdots & \fbox{$a_{j,i}$} & \cdots & a_{j, k} 
	\smallskip \\
	\vdots & \ddots & \vdots & \ddots  & \vdots 
	\\
	a_{k,1} & \cdots & a_{k, i} & \cdots & a_{k,k} 
\end{array}	
\right|^{-1}
\end{eqnarray}
which is equivalent to \eqref{X=A-inverse}.
$\hfill\Box$ \\

\subsection{Proof of Lemma 5.2}
By Proposition 3.1, it suffices for us to show $m=1$ case, that is,  
\begin{eqnarray}
	\left|
	\begin{array}{cccccc}
		\fbox{$\displaystyle{\frac{\mu_{1}}{\mu_{1} - \lambda_{1}}}$} & \cdots & \displaystyle{\frac{\mu_{1}}{\mu_{1} - \lambda_{n}}} \\
		\vdots &  \ddots & \vdots    \\
		\displaystyle{\frac{\mu_{n}}{\mu_{n} - \lambda_{1}}} & \cdots & \displaystyle{\frac{\mu_{n}}{\mu_{n} - \lambda_{n}}}
	\end{array}
\right|
		=
	\frac{\mu_{1}}{\mu_{1} - \lambda_{1}}
	\prod_{j=2}^{n}(\frac{\lambda_{1} - \lambda_{j}}{\lambda_{1} - \mu_{j}} )
	(\frac{\mu_{1} - \mu_{j}}{\mu_{1} - \lambda_{j}}).	
\end{eqnarray}
For $n=2$, it is not hard to check that
\begin{eqnarray}
\left|
\begin{array}{ccc}
\fbox{$\displaystyle{\frac{\mu_{1}}{\mu_{1} - \lambda_{1}}}$} & 
\displaystyle{\frac{\mu_{1}}{\mu_{1} - \lambda_{2}}}
\medskip \\
\displaystyle{\frac{\mu_{2}}{\mu_{2} - \lambda_{1}}} &
\displaystyle{\frac{\mu_{2}}{\mu_{2} - \lambda_{2}}}
\end{array}
\right|
=
\frac{\mu_{1}}{\mu_{1} - \lambda_{1}}
(\frac{\lambda_{1} - \lambda_{2}}{\lambda_{1} - \mu_{2}} )
(\frac{\mu_{1} - \mu_{2}}{\mu_{1} - \lambda_{2}})
\end{eqnarray}
holds.
Suppose that the statement still holds for $n=k$.
Then for $n=k+1$, we have
\begin{eqnarray}
&&\!\!\!\!\left|
	\begin{array}{cccccc}
		\fbox{$\displaystyle{\frac{\mu_{1}}{\mu_{1} - \lambda_{1}}}$} & \cdots & \displaystyle{\frac{\mu_{1}}{\mu_{1} - \lambda_{k}}} &  \displaystyle{\frac{\mu_{1}}{\mu_{1} - \lambda_{k+1}}} \\
		\vdots &  \ddots & \vdots  & \vdots    \\
		\displaystyle{\frac{\mu_{k}}{\mu_{k} - \lambda_{1}}} & \cdots & \displaystyle{\frac{\mu_{k}}{\mu_{k} - \lambda_{k}}}  & \displaystyle{\frac{\mu_{k}}{\mu_{k} - \lambda_{{k+1}}}}  
		\medskip \\ 
		\displaystyle{\frac{\mu_{k+1}}{\mu_{k+1} - \lambda_{1}}} & \cdots & \displaystyle{\frac{\mu_{{k+1}}}{\mu_{{k+1}} - \lambda_{k+1}}} & \displaystyle{\frac{\mu_{k+1}}{\mu_{k+1} - \lambda_{k+1}}}
	\end{array}
	\right|  \nonumber \\
&\!\!\!\!=\!\!\!\!&	
=\mu_{1}(\frac{\lambda_{1} - \lambda_{k+1}}{\mu_{1} - \lambda_{k+1}})
\left|
\begin{array}{cccccc}
	\fbox{$\displaystyle{\frac{1}{\mu_{1} - \lambda_{1}}}$} & \cdots & \displaystyle{\frac{1}{\mu_{1} - \lambda_{k}}} &  1 \\
	\vdots &  \ddots & \vdots  & \vdots    \\
	\displaystyle{\frac{1}{\mu_{k} - \lambda_{1}}} & \cdots & \displaystyle{\frac{1}{\mu_{k} - \lambda_{k}}}  & 1  
	\medskip \\ 
	\displaystyle{\frac{1}{\mu_{k+1} - \lambda_{1}}} & \cdots & \displaystyle{\frac{1}{\mu_{{k+1}} - \lambda_{k}}} & 1
\end{array}
\right|
\\
&&\!\!\!\!\!\!\!\!(\mbox{Here we use the column operations: $j$-th column $\longrightarrow$ $j$-th column $-$ $(k+1)$-th column for}
\nonumber \\
&&\!\!\!\!\!\!\!\!\mbox{ $j=1 \cdots k$, and factor out the row and column common factors $w. r. t$ the box entry.)}
\nonumber 
\end{eqnarray}
\begin{eqnarray}
&\!\!\!\!=\!\!\!\!&	
\mu_{1}(\frac{\lambda_{1} - \lambda_{k+1}}{\mu_{1} - \lambda_{k+1}})
(\frac{\mu_{1} - \mu_{k+1}}{\lambda_{1} -\mu_{k+1}})
\left|
\begin{array}{cccccc}
	\fbox{$\displaystyle{\frac{1}{\mu_{1} - \lambda_{1}}}$} & \cdots & \displaystyle{\frac{1}{\mu_{1} - \lambda_{k}}} &  0 \\
	\vdots &  \ddots & \vdots  & \vdots    \\
	\displaystyle{\frac{1}{\mu_{k} - \lambda_{1}}} & \cdots & \displaystyle{\frac{1}{\mu_{k} - \lambda_{k}}}  & 0  
	\medskip \\ 
	0 & \cdots & 0 & 1
\end{array}
\right|
\\
&&(\mbox{Here we use the row operations: $j$-th row $\longrightarrow$ $j$-th row $-$ $(k+1)$-th row for $j=1 \cdots k$,}
\nonumber \\
&&\mbox{and factor out the row and column common factors $w. r. t$ the box entry.)}
\nonumber \\
&\!\!\!\!=\!\!\!\!&
(\frac{\lambda_{1} - \lambda_{k+1}}{\mu_{1} - \lambda_{k+1}})
(\frac{\mu_{1} - \mu_{k+1}}{\lambda_{1} -\mu_{k+1}})
\left|
\begin{array}{cccccc}
	\fbox{$\displaystyle{\frac{\mu_{1}}{\mu_{1} - \lambda_{1}}}$} & \cdots & \displaystyle{\frac{\mu_{1}}{\mu_{1} - \lambda_{k}}} \\
	\vdots &  \ddots & \vdots  \\
	\displaystyle{\frac{\mu_{k}}{\mu_{k} - \lambda_{1}}} & \cdots & \displaystyle{\frac{\mu_{k}}{\mu_{k} - \lambda_{k}}}   
\end{array}
\right|
\\
&&\mbox{(Here we use the row multiplication rule.)}
\nonumber \\
&\!\!\!\!=\!\!\!\!&
\frac{\mu_{1}}{\mu_{1} - \lambda_{1}}
\prod_{j=2}^{k+1}(\frac{\lambda_{1} - \lambda_{j}}{\lambda_{1} - \mu_{j}} )
(\frac{\mu_{1} - \mu_{j}}{\mu_{1} - \lambda_{j}}).
\end{eqnarray}	
$\hfill\Box$ \\

\subsection{Proof of Lemma 5.3}
The two identities of Lemma 5.3 can be proved by mathematical induction directly and the proofs are quite similar. Here we just show the proof of the second one. For $n=2$, it is not hard to check the following two cases.
For $m=1$ and $m=2$, we have
\begin{eqnarray}
\left|
\begin{array}{cc}
\fbox{$1$} & \displaystyle{\frac{\mu_{1}}{\mu_{1}-\lambda_{2}}}
\medskip\\
1 & \displaystyle{\frac{\mu_{2}}{\mu_{2}-\lambda_{2}}}
\end{array}	
\right|
=\frac{\lambda_{2}}{\mu_{2}}\left(\frac{\mu_{1}-\mu_{2}}{\mu_{1}-\lambda_{2}}\right),~~
\left|
\begin{array}{cc}
	\displaystyle{\frac{\mu_{1}}{\mu_{1}-\lambda_{1}}} & 1
	\medskip\\
	\displaystyle{\frac{\mu_{2}}{\mu_{2}-\lambda_{1}}} & \fbox{$1$}
\end{array}	
\right|
=\frac{\lambda_{1}}{\mu_{1}}\left(\frac{\mu_{2}-\mu_{1}}{\mu_{2}-\lambda_{1}}\right),
\end{eqnarray}
respectively. Suppose that the statement still holds for $n=k$, then for $n=k+1$, we have
\begin{eqnarray}
	&&\!\!\!\!\left|
	\begin{array}{cccccccc}
		\displaystyle{\frac{\mu_{1}}{\mu_{1} - \lambda_{1}}} & \!\!\cdots & \!\!\displaystyle{\frac{\mu_{1}}{\mu_{1} - \lambda_{m-1}}} & 1  & \!\!\displaystyle{\frac{\mu_{1}}{\mu_{1} - \lambda_{m+1}}} & \cdots & \!\!\displaystyle{\frac{\mu_{1}}{\mu_{1} - \lambda_{k}}}
		& \!\!\displaystyle{\frac{\mu_{1}}{\mu_{1} - \lambda_{k+1}}} \\
		\vdots &  \!\!\ddots & \!\!\vdots &  \vdots  & \vdots & \!\!\ddots & \!\!\vdots  & \!\!\vdots   \\
		\displaystyle{\frac{\mu_{m}}{\mu_{m} - \lambda_{1}}} & \cdots & \!\!\displaystyle{\frac{\mu_{m}}{\mu_{m} - \lambda_{m-1}}} & \fbox{$1$} & \!\!\displaystyle{\frac{\mu_{m}}{\mu_{m} - \lambda_{m+1}}} & \!\!\cdots & \!\!\displaystyle{\frac{\mu_{m}}{\mu_{m} - \lambda_{k}}} & \!\!\displaystyle{\frac{\mu_{m}}{\mu_{m} - \lambda_{k+1}}}  \\ 
		\vdots & \!\!\ddots & \!\!\vdots & \vdots & \vdots & \!\!\ddots & \!\!\vdots & \!\!\vdots \\
		\displaystyle{\frac{\mu_{k}}{\mu_{k} - \lambda_{1}}} & \!\!\cdots & \!\!\displaystyle{\frac{\mu_{k}}{\mu_{k} - \lambda_{m-1}}} & 1 & \displaystyle{\frac{\mu_{k}}{\mu_{k} - \lambda_{m+1}}} & \!\!\cdots & \!\!\displaystyle{\frac{\mu_{k}}{\mu_{k} - \lambda_{k}}} & \!\!\displaystyle{\frac{\mu_{k}}{\mu_{k} - \lambda_{k+1}}}
		\medskip\\
		\displaystyle{\frac{\mu_{k+1}}{\mu_{k+1} - \lambda_{1}}} & \!\!\cdots & \!\!\displaystyle{\frac{\mu_{k+1}}{\mu_{k+1} - \lambda_{m-1}}} & 1 & \displaystyle{\frac{\mu_{k+1}}{\mu_{k+1} - \lambda_{m+1}}} & \!\!\cdots & \!\!\displaystyle{\frac{\mu_{k+1}}{\mu_{k+1} - \lambda_{k}}} & \!\!\displaystyle{\frac{\mu_{k+1}}{\mu_{k+1} - \lambda_{k+1}}}
	\end{array}
	\!\!\right|  ~~~~~~~~~~
\end{eqnarray}
\begin{eqnarray}
	\nonumber \\
	&\!\!\!\!=\!\!\!\!&
	-(\mu_{m}-\mu_{k+1})	
	\!\left|
	\begin{array}{cccccccc}
		\displaystyle{\frac{\lambda_{1}}{\mu_{1} - \lambda_{1}}} & \!\!\cdots & \!\!\displaystyle{\frac{\lambda_{m-1}}{\mu_{1} - \lambda_{m-1}}} & 0  & \!\!\displaystyle{\frac{\lambda_{m+1}}{\mu_{1} - \lambda_{m+1}}} & \cdots & \!\!\displaystyle{\frac{\lambda_{k}}{\mu_{1} - \lambda_{k}}}
		& \!\!\displaystyle{\frac{\lambda_{k+1}}{\mu_{1} - \lambda_{k+1}}} \\
		\vdots &  \!\!\ddots & \!\!\vdots &  \vdots  & \vdots & \!\!\ddots & \!\!\vdots  & \!\!\vdots   \\
		\displaystyle{\frac{\lambda_{1}}{\mu_{m} - \lambda_{1}}} & \cdots & \!\!\displaystyle{\frac{\lambda_{m-1}}{\mu_{m} - \lambda_{m-1}}} & \fbox{$0$} & \!\!\displaystyle{\frac{\lambda_{m+1}}{\mu_{m} - \lambda_{m+1}}} & \!\!\cdots & \!\!\displaystyle{\frac{\lambda_{k}}{\mu_{m} - \lambda_{k}}} & \!\!\displaystyle{\frac{\lambda_{k+1}}{\mu_{m} - \lambda_{k+1}}}  \\ 
		\vdots & \!\!\ddots & \!\!\vdots & \vdots & \vdots & \!\!\ddots & \!\!\vdots & \!\!\vdots \\
		\displaystyle{\frac{\lambda_{1}}{\mu_{k} - \lambda_{1}}} & \!\!\cdots & \!\!\displaystyle{\frac{\lambda_{m-1}}{\mu_{k} - \lambda_{m-1}}} & 0 & \displaystyle{\frac{\lambda_{m+1}}{\mu_{k} - \lambda_{m+1}}} & \!\!\cdots & \!\!\displaystyle{\frac{\lambda_{k}}{\mu_{k} - \lambda_{k}}} & \!\! \displaystyle{\frac{\lambda_{k+1}}{\mu_{k} - \lambda_{k+1}}}
		\medskip\\
		\mu_{k+1} & \!\!\cdots & \!\!\mu_{k+1} & 1 & \mu_{k+1} & \!\!\cdots & \!\!\mu_{k+1} & \!\!\mu_{k+1}
	\end{array}
	\!\!\right| \nonumber
\\
&&(\mbox{Here we use the row operations: $j$-th row $\longrightarrow$ $j$-th row $-$ $(k+1)$-th row for $j=1 \cdots k$,}
\nonumber \\
&&\mbox{and factor out the row and column common factors $w. r. t$ the box entry.)}  \nonumber \\
	&\!\!\!\!=\!\!\!\!&
	-\left(\frac{\mu_{m}-\mu_{k+1}}{\mu_{m}-\lambda_{k+1}}\right)	
	\!\left|
	\begin{array}{cccccccc}
		\displaystyle{\frac{\mu_{1}}{\mu_{1} - \lambda_{1}}} & \!\!\cdots & \!\!\displaystyle{\frac{\mu_{1}}{\mu_{1} - \lambda_{m-1}}} & 0  & \!\!\displaystyle{\frac{\mu_{1}}{\mu_{1} - \lambda_{m+1}}} & \cdots & \!\!\displaystyle{\frac{\mu_{1}}{\mu_{1} - \lambda_{k}}}
		& \!\!1 \\
		\vdots &  \!\!\ddots & \!\!\vdots &  \vdots  & \vdots & \!\!\ddots & \!\!\vdots  & \!\!\vdots   \\
		\displaystyle{\frac{\mu_{m}}{\mu_{m} - \lambda_{1}}} & \cdots & \!\!\displaystyle{\frac{\mu_{m}}{\mu_{m} - \lambda_{m-1}}} & \fbox{$0$} & \!\!\displaystyle{\frac{\mu_{m}}{\mu_{m} - \lambda_{m+1}}} & \!\!\cdots & \!\!\displaystyle{\frac{\mu_{m}}{\mu_{m} - \lambda_{k}}} & \!\!1  \\ 
		\vdots & \!\!\ddots & \!\!\vdots & \vdots & \vdots & \!\!\ddots & \!\!\vdots & \!\!\vdots \\
		\displaystyle{\frac{\mu_{k}}{\mu_{k} - \lambda_{1}}} & \!\!\cdots & \!\!\displaystyle{\frac{\mu_{k}}{\mu_{k} - \lambda_{m-1}}} & 0 & \displaystyle{\frac{\mu_{k}}{\mu_{k} - \lambda_{m+1}}} & \!\!\cdots & \!\!\displaystyle{\frac{\mu_{k}}{\mu_{k} - \lambda_{k}}} & \!\! 1
		\medskip\\
		0 & \!\!\cdots & \!\!0 & 1 & 0 & \!\!\cdots & \!\!0 & \!\!\displaystyle{\frac{\mu_{k+1}}{\lambda_{k+1}}}
	\end{array}
	\!\!\right| \nonumber
	\\
&&(\mbox{Here we use column operations: $j$-th column $\rightarrow$ $j$-th column $-$ $(k+1)$-th column for}
\nonumber \\
&&\mbox{$j=1 \cdots k$, and factor out the row and column common factors $w. r. t$ the box entry.)}	
	\nonumber
\end{eqnarray}
\begin{eqnarray}	
	&\!\!\!\!=\!\!\!\!&
	-\left(\frac{\mu_{m}-\mu_{k+1}}{\mu_{m}-\lambda_{k+1}}\right)	
	\!\left|
	\begin{array}{cccccccc}
		\displaystyle{\frac{\mu_{1}}{\mu_{1} - \lambda_{1}}} & \!\!\cdots & \!\!\displaystyle{\frac{\mu_{1}}{\mu_{1} - \lambda_{m-1}}} & 1  & \!\!\displaystyle{\frac{\mu_{1}}{\mu_{1} - \lambda_{m+1}}} & \cdots & \!\!\displaystyle{\frac{\mu_{1}}{\mu_{1} - \lambda_{k}}}
		& \!\!0 \\
		\vdots &  \!\!\ddots & \!\!\vdots &  \vdots  & \vdots & \!\!\ddots & \!\!\vdots  & \!\!\vdots   \\
		\displaystyle{\frac{\mu_{m}}{\mu_{m} - \lambda_{1}}} & \cdots & \!\!\displaystyle{\frac{\mu_{m}}{\mu_{m} - \lambda_{m-1}}} & \fbox{$1$} & \!\!\displaystyle{\frac{\mu_{m}}{\mu_{m} - \lambda_{m+1}}} & \!\!\cdots & \!\!\displaystyle{\frac{\mu_{m}}{\mu_{m} - \lambda_{k}}} & \!\!0  \\ 
		\vdots & \!\!\ddots & \!\!\vdots & \vdots & \vdots & \!\!\ddots & \!\!\vdots & \!\!\vdots \\
		\displaystyle{\frac{\mu_{k}}{\mu_{k} - \lambda_{1}}} & \!\!\cdots & \!\!\displaystyle{\frac{\mu_{k}}{\mu_{k} - \lambda_{m-1}}} & 1 & \displaystyle{\frac{\mu_{k}}{\mu_{k} - \lambda_{m+1}}} & \!\!\cdots & \!\!\displaystyle{\frac{\mu_{k}}{\mu_{_{N}} - \lambda_{k}}} & \!\! 0
		\medskip \nonumber \\
		0 & \!\!\cdots & \!\!0 & \displaystyle{\frac{\mu_{k+1}}{\lambda_{k+1}}} & 0 & \!\!\cdots & \!\!0 & \!\! 1
	\end{array}
	\right|
	\nonumber \\
	&&
	~~~~~~~~~~~~~~~~~~\times\left|
	\begin{array}{cccccccc}
		\displaystyle{\frac{\mu_{1}}{\mu_{1} - \lambda_{1}}} & \!\!\cdots & \!\!\displaystyle{\frac{\mu_{1}}{\mu_{1} - \lambda_{m-1}}} & 1  & \!\!\displaystyle{\frac{\mu_{1}}{\mu_{1} - \lambda_{m+1}}} & \cdots & \!\!\displaystyle{\frac{\mu_{1}}{\mu_{1} - \lambda_{k}}}
		& \!\!0 \\
		\vdots &  \!\!\ddots & \!\!\vdots &  \vdots  & \vdots & \!\!\ddots & \!\!\vdots  & \!\!\vdots   \\
		0 & \cdots & \!\!0 & 1 & \!\!0 & \!\!\cdots & \!\!0 & \!\!\fbox{$0$}  \\ 
		\vdots & \!\!\ddots & \!\!\vdots & \vdots & \vdots & \!\!\ddots & \!\!\vdots & \!\!\vdots \\
		\displaystyle{\frac{\mu_{k}}{\mu_{k} - \lambda_{1}}} & \!\!\cdots & \!\!\displaystyle{\frac{\mu_{k}}{\mu_{k} - \lambda_{m-1}}} & 1 & \displaystyle{\frac{\mu_{k}}{\mu_{k} - \lambda_{m+1}}} & \!\!\cdots & \!\!\displaystyle{\frac{\mu_{k}}{\mu_{k} - \lambda_{k}}} & \!\! 0
		\medskip \nonumber \\
		0 & \!\!\cdots & \!\!0 & \displaystyle{\frac{\mu_{k+1}}{\lambda_{k+1}}} & 0 & \!\!\cdots & \!\!0 & \!\! 1
	\end{array}
	\right|
	\nonumber \\
&&(\mbox{Here we use the permutation rule: $m$-th column $\longleftrightarrow$ $(k+1)$-th column and}
\nonumber \\
&&\mbox{the homological relation: $(m, k+1)$-th $\square$ $\longrightarrow$ $(m,m)$-th $\square$.})
\nonumber 
\end{eqnarray}
\begin{eqnarray}
	&\!\!\!\!=\!\!\!\!&
	-\left(\frac{\mu_{m}-\mu_{k+1}}{\mu_{m}-\lambda_{k+1}}\right)	
	\!\left|
	\begin{array}{cccccccc}
		\displaystyle{\frac{\mu_{1}}{\mu_{1} - \lambda_{1}}} & \!\!\cdots & \!\!\displaystyle{\frac{\mu_{1}}{\mu_{1} - \lambda_{m-1}}} & 1  & \!\!\displaystyle{\frac{\mu_{1}}{\mu_{1} - \lambda_{m+1}}} & \cdots & \!\!\displaystyle{\frac{\mu_{1}}{\mu_{1} - \lambda_{k}}}
		\\
		\vdots &  \!\!\ddots & \!\!\vdots &  \vdots  & \vdots & \!\!\ddots & \!\!\vdots
		\\
		\displaystyle{\frac{\mu_{m}}{\mu_{m} - \lambda_{1}}} & \cdots & \!\!\displaystyle{\frac{\mu_{m}}{\mu_{m} - \lambda_{m-1}}} & \fbox{$1$} & \!\!\displaystyle{\frac{\mu_{m}}{\mu_{m} - \lambda_{m+1}}} & \!\!\cdots & \!\!\displaystyle{\frac{\mu_{m}}{\mu_{m} - \lambda_{k}}}
		\\ 
		\vdots & \!\!\ddots & \!\!\vdots & \vdots & \vdots & \!\!\ddots & \!\!\vdots
		\\
		\displaystyle{\frac{\mu_{k}}{\mu_{k} - \lambda_{1}}} & \!\!\cdots & \!\!\displaystyle{\frac{\mu_{k}}{\mu_{k} - \lambda_{m-1}}} & 1 & \displaystyle{\frac{\mu_{k}}{\mu_{k} - \lambda_{m+1}}} & \!\!\cdots & \!\!\displaystyle{\frac{\mu_{k}}{\mu_{k} - \lambda_{k}}}
		\\
	\end{array}
	\right|
	\left|
	\begin{array}{ccc}
		X & e_{k}^{T}  
		\smallskip \\
		e_{m} & \fbox{$0$}
	\end{array}
	\right|   \nonumber \\
&\!\!\!\!=\!\!\!\!&
	\left(\frac{\mu_{m}-\mu_{k+1}}{\mu_{m}-\lambda_{k+1}}\right)	
	\!\!\prod_{j=1, j \neq m}^{k}\!\!(\frac{\lambda_{j}}{\mu_{j}})(\frac{\mu_{m} - \mu_{j}}{\mu_{m} - \lambda_{j}})
	(X^{-1})_{mk} 
	\nonumber 
=\prod_{j=1, j \neq m}^{k+1}\!\!(\frac{\lambda_{j}}{\mu_{j}})(\frac{\mu_{m} - \mu_{j}}{\mu_{m} - \lambda_{j}}).
	\nonumber \\
&& (\mbox{Here we use \eqref{Quasideterminat_inverse matrix}: $(X)_{mk}^{-1}=|X|_{km}^{-1}= \lambda_{k+1}/ \mu_{k+1}$}.)
\nonumber 		
\end{eqnarray}		
$\hfill\Box$ \\

\subsection{Proof of Lemma 5.4}
By Proposition 3.1, it suffices for us to show $m=1$ case. By the commutative limit of quasideterminant, we have
\begin{eqnarray}
&&\!\!\!\!\!\!\!\!\!\!\!\!\!\!\left|
\begin{array}{cccc}
 \!\!\fbox{$1$} & \!\!1 & \!\!\cdots & \!\!\!\!1 	
\medskip \\ 
\!\!1 & \!\!\displaystyle{\frac{\mu_{2}}{\mu_{2}-\lambda_{2}}} & \!\!\cdots & \!\!\!\!\displaystyle{\frac{\mu_{2}}{\mu_{n} -\lambda_{n}}}
\\
\!\!\vdots & \!\!\ddots & \!\!\vdots & \!\!\!\!\vdots
\\
\!\!1 & \!\!\displaystyle{\frac{\mu_{n}}{\mu_{n}-\lambda_{2}}} & \!\!\cdots & \!\!\!\!\displaystyle{\frac{\mu_{n}}{\mu_{n} -\lambda_{n}}} 
\end{array}
\!\!\right|
=(-1)^{2n}
\left|
\begin{array}{cccc}
	\!1 & \!\!1 & \!\!\cdots & \!\!\!\!1 	
	\medskip \\ 
	\!1 & \!\!\displaystyle{\frac{\mu_{2}}{\mu_{2}-\lambda_{2}}} & \!\!\cdots & \!\!\!\!\displaystyle{\frac{\mu_{2}}{\mu_{n} -\lambda_{n}}}
	\\
	\!\vdots & \!\!\ddots & \!\!\vdots & \!\!\!\!\vdots
	\\
	\!1 & \!\!\displaystyle{\frac{\mu_{n}}{\mu_{n}-\lambda_{2}}} & \!\!\cdots & \!\!\!\!\displaystyle{\frac{\mu_{n}}{\mu_{n} -\lambda_{n}}} 
\end{array}
\!\!\right|
\Big/
\left|
\begin{array}{cccc} 
	\!\displaystyle{\frac{\mu_{2}}{\mu_{2}-\lambda_{2}}} & \!\!\cdots & \!\!\!\!\displaystyle{\frac{\mu_{2}}{\mu_{n} -\lambda_{n}}}
	\\
	\!\vdots & \!\!\ddots & \!\!\!\!\vdots 
	\\
	\!\displaystyle{\frac{\mu_{n}}{\mu_{n}-\lambda_{2}}} & \!\!\cdots & \!\!\!\!\displaystyle{\frac{\mu_{n}}{\mu_{n} -\lambda_{n}}} 
\end{array}
\!\!\right| \nonumber 
\end{eqnarray}
\begin{eqnarray}
=
\frac{
\left|
(\displaystyle{\frac{\lambda_{j}}{\mu_{j}-\lambda_{k}}})_{2 \leq j, k \leq N}
\right|
}
{
\left|
(\displaystyle{\frac{\mu_{j}}{\mu_{j}-\lambda_{k}}})_{2 \leq j, k \leq N}
\right|
}
=\prod_{j=2}^{n}(\frac{\lambda_{j}}{\mu_{j}}).
\end{eqnarray}

$\hfill\Box$ \\

\section{Proof of Lemma 5.6}
By \eqref{J=J_pm(k)} and \eqref{Sylvester eq-pm}, we have
\begin{eqnarray}
	\label{J=J_K-pm}	
	J_{_{[N]}}
	=
	\left|
	\begin{array}{cc}
		\!\!\Xi^{T}\Omega^{\scriptscriptstyle(K)} & \!\!\theta^{{\scriptscriptstyle(K)}\dagger}
		\smallskip\\
		\!\!\theta^{\scriptscriptstyle(K)} & \!\!\fbox{$I$}
	\end{array}
	\!\!\right|,~
	(\Xi^{T}\Omega^{\scriptscriptstyle(K)})_{j\ell} =
	\gamma_{j\ell}(\theta^{{\scriptscriptstyle(K)}\dagger})_{j ~\!\!\cdot }(\theta^{\scriptscriptstyle(K)})_{\cdot ~\!\!\ell}.	
\end{eqnarray}
Note that $\Xi^{T}\Omega^{\scriptscriptstyle(K)}$ can be decomposed into
\begin{eqnarray}
	\label{Omega_decomposition}
	\Xi^{T}\Omega^{\scriptscriptstyle(K)}
	=
	\rho^{\dagger}\sigma,
\end{eqnarray}
where the $j$-th column of $\rho$ and the $\ell$-th column of $\sigma$ are defined by 
\begin{eqnarray}
	\label{rho-sigma}
	\rho_{\cdot ~\!\! j}
	:=
	c_{j}(\theta^{\scriptscriptstyle(K)})_{\cdot ~\!\! j}~, ~~
	\sigma_{\cdot ~\!\! \ell}
	:=
	d_{k}(\theta^{\scriptscriptstyle(K)})_{\cdot ~\!\! \ell}~, ~~
\end{eqnarray}
for some constant $c_{j}, d_{k}$ such that $\overline{c}_{ j}d_{k}=\gamma_{jk}$. 
Taking partial derivative $\partial_{m}$ on \eqref{Omega_decomposition}, we have
\begin{eqnarray}
	\Xi^{T}\partial_{m}\Omega^{\scriptscriptstyle(K)}
	=
	\Xi^{T}(\partial_{m}\rho^{\dagger})\sigma
	+
	\Xi^{T}\rho^{\dagger}(\partial_{m}\sigma).
\end{eqnarray}	
Using the derivative formula of quasideterminant on \eqref{J=J_K-pm}, we have
\begin{eqnarray}
	\partial_{m}{J}_{_{[N]}}
	&\!\!\!\!=\!\!\!\!&
	\partial_{m}I
	+
	\left|
	\begin{array}{cc}
		\Xi^{T}\Omega^{\scriptscriptstyle(K)} & \theta^{{\scriptscriptstyle(K)}\dagger} 
		\smallskip \\
		\partial_{m}\theta^{\scriptscriptstyle(K)} &
		\fbox{$0$}
	\end{array}
	\right|
	+
	\left|
	\begin{array}{cc}
		\Xi^{T}\Omega^{\scriptscriptstyle(K)} & \partial_{m}\theta^{{\scriptscriptstyle(K)}\dagger}
		\smallskip \\
		\theta^{\scriptscriptstyle(K)} &
		\fbox{$0$}
	\end{array}
	\right|
	\nonumber  \\
	&&\!\!\!\!+
	\left|
	\begin{array}{cc}
		\!\!\Omega^{\scriptscriptstyle(K)} & \!\!\partial_{m}\rho^{\dagger}
		\smallskip \\
		\!\!\theta^{\scriptscriptstyle(K)} &
		\!\!\fbox{$0$}
	\end{array}
	\!\!\right|
	\left|
	\begin{array}{cc}
		\!\!\Xi^{T}\Omega^{\scriptscriptstyle(K)} & \!\!\theta^{{\scriptscriptstyle(K)}\dagger} 
		\smallskip \\
		\!\!\sigma  &
		\!\!\fbox{$0$}
	\end{array}
	\!\!\right|
	+
	\left|
	\begin{array}{cc}
		\!\!\Omega^{\scriptscriptstyle(K)} & \!\!\rho^{\dagger}
		\smallskip \\
		\!\!\theta^{\scriptscriptstyle(K)} &
		\!\!\fbox{$0$}
	\end{array}
	\!\!\right|
	\left|
	\begin{array}{cc}
		\!\!\Xi^{T}\Omega^{\scriptscriptstyle(K)} & \!\!\theta^{{\scriptscriptstyle(K)}\dagger} 
		\smallskip \\
		\!\!\partial_{m}\sigma  &
		\!\!\fbox{$0$}
	\end{array}
	\!\!\right|. ~~~~~~
\end{eqnarray}
Now we find that the partial derivative terms appear in the above formula are all involving the columns of $\partial_{m}\theta^{\scriptscriptstyle(K)}$ (Cf: \eqref{rho-sigma}).
Therefore, it suffices for us to show that
\begin{eqnarray}
	\lim_{r \rightarrow \infty}\left[\partial_{m}(\theta^{\scriptscriptstyle(K)})_{\cdot ~\!\! j}\right]
	=
	\partial_{m}\left[
	\lim_{r \rightarrow  \infty}(\theta^{\scriptscriptstyle(K)})_{\cdot ~\!\!j} \right]
	=
	\partial_{m}(\vartheta^{\scriptscriptstyle(K)})_{\cdot ~\!\!j}.
\end{eqnarray}
This fact is true and can be verified directly by \eqref{theta_K-pm} and \eqref{theta_asymptotic}.
Therefore, 
\begin{eqnarray}
	\lim_{r \rightarrow \infty}(\partial_{m}{J}_{_{[N]}})
	=	  
	\partial_{m}(\mathcal{J}_{_{[N]}}^{\scriptscriptstyle(K)})
	=
	\partial_{m}\left(\lim_{r \rightarrow \infty}J_{_{[N]}}\right).
\end{eqnarray}
$\hfill\Box$ \\

\section{Verification of Theorem 5.7}
\newtheorem{Lem}{Lemma}[section]\label{lemma}
\begin{Lem}
	Let $J$ be in the form of 
	\begin{eqnarray}
		\label{J-Delta}
		J=
		\frac{1}{\Delta}
		\left(
		\begin{array}{cc}
			\Delta_{11} & \Delta_{12} \\
			\Delta_{21} & \Delta_{22}
		\end{array}
		\right), ~
		\left\{
		\begin{array}{l}
			\Delta:=
			C^{\scriptscriptstyle(+)}e^{L + \widetilde{L}} + C^{\scriptscriptstyle(-)}e^{-(L + \widetilde{L})}	
			\\
		\Delta_{11}:= A^{\scriptscriptstyle(+)}e^{L + \widetilde{L}} + A^{\scriptscriptstyle(-)}e^{-(L + \widetilde{L})},~ 
		\Delta_{12}:= Be^{L -\widetilde{L}}	
		\\
		\Delta_{21}:= \widetilde{B}e^{-(L -\widetilde{L})},~	
		\Delta_{22}:=\widetilde{A}^{\scriptscriptstyle(+)}e^{L + \widetilde{L}} + \widetilde{A}^{\scriptscriptstyle(-)}e^{-(L + \widetilde{L})}	
	\end{array},
	\right.
\end{eqnarray}
where
\begin{eqnarray}
	L:=\lambda\alpha z + \beta\widetilde{z} + \lambda\beta w + \alpha\widetilde{w}:=\ell_{m}z^{m},~~
	\widetilde{L}:=\mu\widetilde{\alpha} z + \widetilde{\beta}\widetilde{z} + \mu\widetilde{\beta} w + \widetilde{\alpha}\widetilde{w}:=\widetilde{\ell}_{m}z^{m}
\end{eqnarray}
are determined by suitable parameter choices in \eqref{Reality conditions} so that $\overline{L}=\widetilde{L}$ holds on each real space. 
If $A^{(\pm)}$, $\widetilde{A}^{(\pm)}$, $C^{(\pm)}$ satisfy the relation
\begin{eqnarray}
	\label{1-soliton_conditions}	
	\frac{A^{(\pm)}\widetilde{A}^{(\pm)}}{(C^{(\pm)})^2}=|J|, ~~\frac{C^{\scriptscriptstyle(+)}}{C^{\scriptscriptstyle(-)}} \in \mathbb{R}^{+},
\end{eqnarray}
then
\begin{eqnarray}
	\label{Trace formula_NLSM}
	&&\!\!\!\!\mbox{Tr}\left[
	(\partial_{m}J)J^{-1}(\partial_{n}J)J^{-1}\right]  \nonumber 
	\\
	&\!\!\!\!=\!\!\!\!&
	\frac{\!\!\!\!\!-\!\left[
		\begin{array}{l}
			~(\ell_{m}\ell_{n} + \widetilde{\ell}_{m}\widetilde{\ell}_{n})
			(A^{\scriptscriptstyle(+)}\!\widetilde{A}^{\scriptscriptstyle(-)}
			\!+\!\widetilde{A}^{\scriptscriptstyle(+)}\!A^{\scriptscriptstyle(-)}
			\!-\!2|J|C^{\scriptscriptstyle(+)}C^{\scriptscriptstyle(-)}
			\!-\!B\widetilde{B}
			)
			\smallskip \\
			\!\!\!+(\ell_{m}\widetilde{\ell}_{n} + \widetilde{\ell}_{m}\ell_{n})
			(A^{\scriptscriptstyle(+)}\!\widetilde{A}^{\scriptscriptstyle(-)}
			\!+\!\widetilde{A}^{\scriptscriptstyle(+)}\!A^{\scriptscriptstyle(-)}
			\!-\!2|J|C^{\scriptscriptstyle(+)}C^{\scriptscriptstyle(-)}
			\!+\!B\widetilde{B}
			)
		\end{array}
		\!\!\!\right]	}
	{2|J|C^{\scriptscriptstyle(+)}C^{\scriptscriptstyle(-)}}
	\mbox{sech}^2\!\left[
	L \!+ \!\widetilde{L} \!+ \!\frac{1}{2}\!\log(\frac{C^{\scriptscriptstyle(+)}}{C^{\scriptscriptstyle(-)}})
	\!\right], ~~~~~~
\end{eqnarray}
and
\begin{eqnarray}
	\label{Trace formula_WZ}	
	\mbox{Tr}\left[
	(\partial_{m}J)J^{-1}(\partial_{n}J)J^{-1}(\partial_{p}J)J^{-1}
	\right] 
	=
	0.	
\end{eqnarray}
\end{Lem}
\textit{Verification of Lemma D.1:}
\medskip \\
By \eqref{Tr(A_m A_n)} and direct computation, we have
\begin{eqnarray}
&&\!\!\!\!\mbox{Tr}\left[
(\partial_{m}J)J^{-1}(\partial_{n}J)J^{-1}
\right]  \nonumber 
\\
&\!\!\!\!=\!\!\!\!&
\frac{
	2\left\{
	\begin{array}{l}
		~(\ell_{m}+\widetilde{\ell}_{m})(\ell_{n}+\widetilde{\ell}_{n})\left[
		\begin{array}{l}
			~\left(A^{\scriptscriptstyle(+)}\widetilde{A}^{\scriptscriptstyle(+)}-|J|(C^{\scriptscriptstyle(+)})^2\right)
			e^{2(L + \widetilde{L})} 
			\smallskip \\
			\!\!+
			\left(A^{\scriptscriptstyle(-)}\widetilde{A}^{\scriptscriptstyle(-)}-|J|(C^{\scriptscriptstyle(-)})^2\right)
			e^{-2(L + \widetilde{L})} 
		\end{array}
		\right]
		\smallskip \\
		\!\!- (\ell_{m}+\widetilde{\ell}_{m})(\ell_{n}+\widetilde{\ell}_{n})
		\left(A^{\scriptscriptstyle(+)}\widetilde{A}^{\scriptscriptstyle(-)} + \widetilde{A}^{\scriptscriptstyle(+)}A^{\scriptscriptstyle(-)}
		-
		2|J|C^{\scriptscriptstyle(+)}C^{\scriptscriptstyle(-)}
		\right)
		\smallskip \\
		\!\!+ (\ell_{m}-\widetilde{\ell}_{m})(\ell_{n}-\widetilde{\ell}_{n})B\widetilde{B}
	\end{array}
	\!\!\!\right\}
}
{|J|\left(C^{\scriptscriptstyle(+)}e^{L + \widetilde{L}} + C^{\scriptscriptstyle(-)}e^{-(L + \widetilde{L})}\right)^{2}
}.
\end{eqnarray}
Under he conditions of \eqref{1-soliton_conditions}, one can summarize the above trace formula as \eqref{Trace formula_NLSM}.
On the other hand, using \eqref{Tr(A_m A_n A_p)} together with (2.25) and (4.2)--(4.6) of \cite{HHK-2023}, one obtains \eqref{Trace formula_WZ}.
$\hfill\Box$ \\

\textit{Verification of Theorem 5.7:}
\smallskip \\
Comparing \eqref{J-Delta} with \eqref{J_asymptotic_exact}, one can set
\begin{eqnarray}
\label{J_coeff}	
	\begin{array}{ll}	
		A^{\scriptscriptstyle(+)}
		=
		\widetilde{\gamma}_{_{KK}}p_{_K}\widetilde{p}_{_K}|a_{_K}^{\scriptscriptstyle(+)}|^4, & A^{\scriptscriptstyle(-)}
		=
		\gamma_{_{KK}}q_{_K}\widetilde{q}_{_K}|a_{_K}^{\scriptscriptstyle(-)}|^4, 
		\medskip \\
		\widetilde{A}^{\scriptscriptstyle(+)}
		=
		\gamma_{_{KK}}p_{_K}\widetilde{p}_{_K}|a_{_K}^{\scriptscriptstyle(+)}|^4, & \widetilde{A}^{\scriptscriptstyle(-)}=\widetilde{\gamma}_{_{KK}}q_{_K}\widetilde{q}_{_K}|a_{_K}^{\scriptscriptstyle(-)}|^4,
		\medskip \\
		C^{\scriptscriptstyle(+)}
		=
		p_{_K}\widetilde{p}_{_K}|a_{_K}^{\scriptscriptstyle(+)}|^2, &
		C^{\scriptscriptstyle(-)}
		=
		q_{_K}\widetilde{q}_{_K}|a_{_K}^{\scriptscriptstyle(-)}|^2,
		\smallskip \\
		B
		=
		-p_{_K}\widetilde{q}_{_K}\left(a_{_K}^{\scriptscriptstyle(+)}\overline{a}_{_K}^{\scriptscriptstyle(-)}\right)^2, &
		\widetilde{B}
		=
		-\widetilde{p}_{_K}q_{_K}\left(\overline{a}_{_K}^{\scriptscriptstyle(+)}a_{_{K}}^{\scriptscriptstyle(-)}\right)^2.
	\end{array} 
\end{eqnarray}
Note from \eqref{det(J)} that 
\begin{eqnarray}
\label{J/C_determinant_}	
	|\mathcal{J}|=|\mathcal{J}_{_{[N]}}^{\scriptscriptstyle (K)}|/|C^{\scriptscriptstyle(K)}|
	=\lambda_{_{K}}^{\scriptscriptstyle(+)}/\lambda_{_{K}}^{\scriptscriptstyle(-)}.
\end{eqnarray}
It is not difficult to check that \eqref{J_coeff} satisfies the requirement of  \eqref{1-soliton_conditions} on each real space expect in the cases of even soliton numbers $N=2n, n \in \mathbb{N}$, on $\mathbb{E}^{4}$ (cf. case (3) of \eqref{Phase shifts_each space}). 

Substituting \eqref{J_coeff}, \eqref{J/C_determinant_}  into \eqref{Trace formula_NLSM} and replacing $L$, $\widetilde{L}$ by $L_{_K}$, $\widetilde{L}_{_K}$, one can obtain
\begin{eqnarray}
&&\!\!\!\!\mbox{Tr}\left[
(\partial_{m}\mathcal{J})\mathcal{J}^{-1}(\partial^{m}\mathcal{J})\mathcal{J}^{-1}\right]  \nonumber \\
&\!\!\!\!=\!\!\!\!&
2\mbox{Tr}
\left[
(\partial_{z}\mathcal{J})\mathcal{J}^{-1}(\partial_{\widetilde{z}}\mathcal{J})\mathcal{J}^{-1}
-
(\partial_{w}\mathcal{J})\mathcal{J}^{-1}(\partial_{\widetilde{w}}\mathcal{J})\mathcal{J}^{-1}
\right]   \nonumber \\
&\!\!\!\!\!\!=\!\!\!\!\!\!&
\label{NLSM action_asymp_coeff form}
\frac{-2(\lambda_{_{K}}-\mu_{_{K}})^2}{\lambda_{_{K}}\mu_{_{K}}}
\left[
\begin{array}{l}
~\!\ell_{_{K}z}\widetilde{\ell}_{_{K}\widetilde{z}}
-\ell_{_{K}w}\widetilde{\ell}_{_{K}\widetilde{w}}
\\
\!\!\!+\widetilde{\ell}_{_{K}z}\ell_{_{K}\widetilde{z}}
-\widetilde{\ell}_{_{K}w}\ell_{_{K}\widetilde{w}}
\end{array}
\!\!\right]
\mbox{sech}^2\!\left[
L_{_{K}} \!+ \!\widetilde{L}_{_{K}} \!+ \!\log(\frac{|a_{_K}^{\scriptscriptstyle(+)}|}{|a_{_K}^{\scriptscriptstyle(-)}|}) 
\!+ \delta_{_K} \right], ~~~~~~~ 
\end{eqnarray}
where
\begin{eqnarray}
\label{Phase shift_pq form}	
\delta_{_K}:=
\frac{1}{2}\log(\frac{p_{_K}\widetilde{p}_{_K}}{q_{_K}\widetilde{q}_{_K}}).
\end{eqnarray}	
Substituting the coefficients of $L_{_{K}}$ and $\widetilde{L}_{_{K}}$ (cf. \eqref{E_j-pm}, \eqref{E-tilde_j-pm} for $j=K$ ) into \eqref{NLSM action_asymp_coeff form} and comparing \eqref{pq_K} with \eqref{Phase shift_pq form}, one obtains \eqref{NLSM action_asymptotic} and \eqref{Phase shift factors_explicit form}.

$\hfill\Box$ \\

\section{Exact $N$-Soliton Solution ($N=4$)}
\begin{eqnarray}
\Delta_{[4]}^{(11)}
=

\!\!\right|
\xi_{[4]}^3|E_{4}^{\scriptscriptstyle(-)}|^2
\right]\!\!\Phi_{1}^{\scriptscriptstyle(-)}\Phi_{2}^{\scriptscriptstyle(+)}\Phi_{3}^{\scriptscriptstyle(+)} 
\end{array}	
\!\!\right].
\end{array}
\end{eqnarray}
\end{appendices}

\newpage
\small{

}

\end{document}